\def\year{2021}\relax
\documentclass[letterpaper]{article}
\usepackage[table]{xcolor}
\usepackage{hyperref}
\hypersetup{
    pdflang=en-US,
    hidelinks
}
\usepackage{aaai21-arxiv}

\usepackage[verbose=true,letterpaper]{geometry}
\AtBeginDocument{
  \newgeometry{
    textheight=9in,
    textwidth=6in,
    top=1in,
    headheight=12pt,
    headsep=25pt,
    footskip=30pt
  }
}
\usepackage{times}
\usepackage{helvet}
\usepackage{courier}
\usepackage{graphicx}
\usepackage{graphicx}
\usepackage{natbib}
\let\cite\citep

\usepackage{caption}
\newcommand{\parencite}{\citep}
\newcommand{\textcite}{\citet}

\title{Hindsight and Sequential Rationality of Correlated Play}

\author{
Dustin Morrill\\
        University of Alberta; Amii\textsuperscript{\rm $\dagger$}\\
        Edmonton, Alberta, Canada\\
        \texttt{morrill@ualberta.ca}\\
        \And
    Ryan D'Orazio\\
        Universit\'{e} de Montr\'{e}al; Mila\\
        Montr\'{e}al, Qu\'{e}bec, Canada\\
        \texttt{ryan.dorazio@mila.quebec}\\
        \AND
    Reca Sarfati\\
        Massachusetts Institute of Technology\\
        Cambridge, Massachusetts, United States\\
        \texttt{sarfati@mit.edu}\\
        \And
    Marc Lanctot\\
        DeepMind\textsuperscript{\rm $\ddagger$}\\
        \texttt{lanctot@google.com}\\
        \AND
    James R.\ Wright\textsuperscript{\rm $\dagger$}\\
        \texttt{james.wright@ualberta.ca}\\
        \And
    Amy R.\ Greenwald\\
        Brown University\\
        Providence, Rhode Island, United States\\
        \texttt{amy\_greenwald@brown.edu}\\
        \AND
    Michael Bowling\textsuperscript{\rm $\dagger$, \rm $\ddagger$}\\
        \texttt{mbowling@ualberta.ca}
}

\def\correctionPaperReferenceKey/{hsr2020arxivCorrections}

\usepackage{mathtools}

\DeclarePairedDelimiter{\abs}{\lvert}{\rvert}

\DeclarePairedDelimiter{\subex}{(}{)}
\DeclarePairedDelimiter{\subblock}{[}{]}
\DeclarePairedDelimiter{\tuple}{(}{)}
\DeclarePairedDelimiter{\set}{\{}{\}}

\usepackage{bbold}
\usepackage{bm}

\newcommand{\reals}{\mathbb{R}}

\newcommand{\Simplex}{\Delta}
\newcommand{\simplex}{\Simplex}

\newcommand{\bs}[1]{\bm{#1}}

\newcommand{\expectation}{\mathbb{E}}
\newcommand{\E}{\expectation}

\newcommand{\as}{\doteq}

\newcommand{\given}{\,|\,}

\newcommand{\where}{\;|\;}

\newcommand{\PureStratSet}{S}
\newcommand{\pureStrat}{s}
\newcommand{\pureProfile}{\pureStrat}

\newcommand{\reward}{r}
\newcommand{\utility}{u}
\newcommand{\policy}{\pi}
\newcommand{\StrategySet}{\Pi}
\newcommand{\strategy}{\policy}
\newcommand{\strat}{\strategy}

\newcommand{\gap}{\varepsilon}
\newcommand{\recDist}{\mu}

\newcommand{\Actions}{\mathcal{A}}

\newcommand{\regret}{\rho}
\newcommand{\Regret}{R}

\newcommand{\supReward}{U}
\newcommand{\maxReward}{\supReward}

\newcommand{\EXT}{\textsc{ex}}
\newcommand{\IN}{\textsc{in}}
\newcommand{\INT}{\IN}
\newcommand{\SWAP}{\textsc{sw}}

\newcommand{\infoSet}{I}
\newcommand{\InfoSets}{\mathcal{I}}
\newcommand{\reachProb}{P}
\newcommand{\chance}{c}

\newcommand{\Histories}{\mathcal{H}}
\newcommand{\TerminalHistories}{\mathcal{Z}}

\newcommand{\playerChoice}{\mathcal{P}}
\usepackage{amssymb}
\newcommand{\emptyHistory}{\varnothing}

\newcommand{\termValue}{\reward}
\newcommand{\cfIv}{v}

\newcommand{\cfv}{\cfIv}

\usepackage{upgreek}

\newcommand{\immStrat}{\sigma}

\newcommand{\DevSet}{\Phi}
\newcommand{\dev}{\phi}

\newcommand{\TARGET}{\odot}
\newcommand{\TRIGGER}{\text{!}}

\usepackage{xspace}

\makeatletter
\DeclareRobustCommand\onedot{\futurelet\@let@token\@onedot}
\def\@onedot{\ifx\@let@token.\else.\null\fi\xspace}

\def\eg/{\emph{e.g}\onedot} \def\Eg/{\emph{E.g}\onedot}
\def\ie/{\emph{i.e}\onedot} \def\Ie/{\emph{I.e}\onedot}
\def\cf/{\emph{c.f}\onedot} \def\Cf/{\emph{C.f}\onedot}
\def\vs/{\emph{vs}\onedot} \def\Vs/{\emph{Vs}\onedot}
\def\etc/{\emph{etc}\onedot}
\def\wrt/{w.r.t\onedot} \def\dof/{d.o.f\onedot}
\def\etal/{\emph{et al}\onedot}
\def\viceversa/{\emph{vice-versa}}
\def\ow/{\emph{o.w}\onedot}
\def\whp/{w.h.p\onedot}
\def\apriori/{\emph{a priori}} \def\Apriori/{\emph{A priori}}
\def\ala/{\`{a} la}

\def\naive/{na\"{\i}ve} \def\Naive/{Na\"{\i}ve}
\def\rmPlus/{regret matching\textsuperscript{+}}
\def\rrmPlus/{RRM\textsuperscript{+}}
\def\rcfrPlus/{RCFR\textsuperscript{+}}
\def\cfrPlus/{CFR\textsuperscript{+}}
\@ifdefinable{\Politex/}{\def\Politex/{\textsc{Politex}}}

\def\NashConv/{\textsc{NashConv}}
\def\NashConvAUC/{$\overline{\textsc{NashConv}}$}

\makeatother

\newcommand{\textbxf}[1]{{\fontseries{b}\selectfont #1}}

\def\efceFootnote/{3}
 \def\EventX/{X}
\def\EventY/{Y}
\def\Upgrade/{U}
\def\NotUpgrade/{$\neg$\Upgrade/}
\def\mEventX{\text{\EventX/}}
\def\mEventY{\text{\EventY/}}
\def\mUpgrade{\text{\Upgrade/}}
\def\mNotUpgrade{\neg\mUpgrade}

\def\Rock/{R}
\def\Paper/{P}
\def\Scissors/{S}
\def\PredRock/{\Rock/?}
\def\PredNotRock/{$\neg$\PredRock/}

\def\Heads/{H}
\def\Tails/{T}
\def\Match/{M}
\def\NotMatch/{$\neg$\Match/}
\def\mHeads{\text{\Heads/}}
\def\mTails{\text{\Tails/}}
\def\mMatch{\text{\Match/}}
\def\mNotMatch{\neg\text{\Match/}}

\def\Play/{P}
\def\mPlay{\text{\Play/}}
\def\mNotPlay{\neg\mPlay}
\def\NotPlay/{$\mNotPlay$}

\def\EvColDesc/{The ``EV'' column shows the expected utility of each deviation under a uniform distribution over the two recommendations.}
\def\ArrowDescriptions/{The line of play is traced out by bold arrows leading to a single outcome highlighted in a rectangle.
Black arrows are recommendations, red arrows are deviations, and grey arrows are leftover actions.}
\def\InfoSetDesc/{Dashed lines indicate that each of player two's histories are in the same information set.}
\def\PlayerTwoDesc/{Player two's behavior is fixed given a recommendation and only serves to determine player one's value so it is shown separately at the top of the figure.
}

\def\botsExampleFollowTrees{
  \node(follow-Label) [behaveLabel] {always\\follow}; \&
  \node(follow-1) [state] {};
  \devRecTreeNodeCoords{follow-1};
  \deviationExampleStates{follow-1}{}{};
  \botsExampleTerminalsOne{follow-1}{}{}{}{draw=black};
  \devRecRootEdges{follow-1}{\Upgrade/}{alt}{\NotUpgrade/}{follow};
  \devRecLEdges{follow-1}{\EventX/}{alt}{\EventY/}{follow};
  \devRecREdges{follow-1}{\EventX/}{alt}{\EventY/}{follow}; \&

  \node(follow-2) [state] {};
  \devRecTreeNodeCoords{follow-2};
  \node(follow-2L) [state] at (follow-2LCoord) {};
  \deviationExampleStates{follow-2}{}{};
  \botsExampleTerminalsTwo{follow-2}{}{}{draw=black}{};
  \devRecRootEdges{follow-2}{\Upgrade/}{alt}{\NotUpgrade/}{follow};
  \devRecLEdges{follow-2}{\EventX/}{follow}{\EventY/}{alt};
  \devRecREdges{follow-2}{\EventX/}{follow}{\EventY/}{alt}; \&
  \node(follow-Value) {$+1.5$};
}
\def\mpExampleFollowTrees{
  \node(follow-Label) [behaveLabel] {always\\follow}; \&
  \node(follow-1) [state] {};
  \devRecTreeNodeCoords{follow-1};
  \deviationExampleStates{follow-1}{}{};
  \mpExampleTerminalsOne{follow-1}{}{}{}{draw=black};
  \devRecRootEdges{follow-1}{M}{alt}{$\neg\text{M}$}{follow};
  \devRecLEdges{follow-1}{H}{follow}{T}{alt};
  \devRecREdges{follow-1}{H}{alt}{T}{follow};
  \&

  \node(follow-2) [state] {};
  \devRecTreeNodeCoords{follow-2};
  \deviationExampleStates{follow-2}{}{};
  \mpExampleTerminalsTwo{follow-2}{draw=black}{}{}{};
  \devRecRootEdges{follow-2}{M}{follow}{$\neg\text{M}$}{alt};
  \devRecLEdges{follow-2}{H}{follow}{T}{alt};
  \devRecREdges{follow-2}{H}{alt}{T}{follow};
  \&
  \node(follow-Value) {0};
}
 \definecolor{offWhite}{RGB}{240,240,240}
\definecolor{grey}{RGB}{180,180,180}
\definecolor{darkgreen}{RGB}{0,125,0}
\definecolor{lime}{RGB}{255,200,0}

\definecolor{amiiBlue}{RGB}{16,72,118}
\definecolor{amiiPink}{RGB}{241,97,119}
\definecolor{amiiYellow}{RGB}{248,209,109}
\definecolor{amiiPurple}{RGB}{123,105,145}

\colorlet{yes}{cyan!50!white}
\colorlet{newYes}{cyan!75!white}
\colorlet{no}{red!50!white}
\colorlet{newNo}{red!75!white}
\newcommand{\eqTableEntry}[2]{
  \cellcolor{#1}{\textcolor{black}{#2}}
}

\usepackage{amsthm}
\makeatletter
\ifdefined\theorem
\else
  \newtheorem{theorem}{Theorem}
\fi
\ifdefined\lemma
\else
  \newtheorem{lemma}{Lemma}
\fi
\ifdefined\corollary
\else
  \newtheorem{corollary}{Corollary}
\fi
\ifdefined\definition
\else
  \newtheorem{definition}{Definition}
\fi
\ifdefined\proposition
\else
  \newtheorem{proposition}{Proposition}
\fi
\makeatother

\usepackage{booktabs}
\usepackage{multirow}
\usepackage{pbox}
\usepackage{array}

\usepackage{threeparttable}
\usepackage{tabularx}

\DeclareMathSymbol{\negative}{\mathbin}{AMSa}{"39}

\makeatletter
\ifdefined\algorithmic
\else
  \usepackage{algorithm}
  \usepackage{algorithmic}
\fi
\makeatother

\makeatletter
\newcommand\safeIncCounter[1]{\@ifundefined{c@#1}{\newcounter{#1}\stepcounter{#1}}{\stepcounter{#1}}}
\makeatother

\usepackage{xargs}

\def\gTwoFive/{g\textsubscript{2, 5, $\uparrow$}}
\def\gThreeFour/{g\textsubscript{3, 4, $\uparrow$}}

\usepackage{nomencl}
\makenomenclature

\makeatletter
\ifdefined\isCorrectionsPaper
  \def\correctionBlockDescription/{The changes to the original paper pertaining to this correction can be found within labeled and numbered ``correction'' blocks.
    A brief description of the specific change within a given correction block can be found at the start of that block.}
\else
  \def\correctionBlockDescription/{The changes to the original paper pertaining to this correction are highlighted with labeled and numbered ``correction'' blocks in the separate report \textcite{\correctionPaperReferenceKey/}.}
\fi
\makeatother

\newcommand{\hsrIntroCorrectionsFootnote}{
  \footnote{This paper is changed from its original version published as \textcite{hsr2020}.
    That version lacks an example of a beneficial external deviation in a CFCE and asserts that counterfactual deviations subsume external deviations without qualification.
    This assertion is incorrect, as the example illustrates.
    Instead, observable sequential rationality is required to ensure that counterfactual deviations subsume external deviations.
    \correctionBlockDescription/
  }
}

 \usepackage{tikz}
\usetikzlibrary{
    shapes.geometric,
    arrows,
    arrows.meta,
    graphs,
    graphs.standard,
    calc,
    positioning,
    intersections,
    chains,
    matrix,
    shapes.misc,
    backgrounds,
    overlay-beamer-styles,
    decorations.pathmorphing,
    decorations.pathreplacing,
    shapes.symbols,
    trees
}

\def\cmToEmFactor{2.3710630158366}

\def\shortAYLength{0.5cm}

\tikzset{
  influenceArrow/.style={>=Triangle, -{>[scale=#1]}},
  influenceArrow/.default=0.4
}

\tikzstyle{devColor} = [draw=red]
\tikzstyle{dev} = [devColor, very thick]

\tikzstyle{followColor} = [draw=black]
\tikzstyle{follow} = [followColor, very thick]

\tikzstyle{infoColor} = [draw=cyan]
\tikzstyle{info} = [infoColor, very thick]
\tikzstyle{infoArrow} = [influenceArrow, thin, infoColor]

\tikzstyle{alt} = [draw=grey]
\tikzstyle{zeroProb} = [draw=grey]
\tikzstyle{rec} = [draw=black, densely dashed, very thick]
\def\stateMinimumRadius{0.18cm}
\tikzstyle{state} = [circle, draw=black, minimum size=2*\stateMinimumRadius, inner sep=0.5mm, fill=white]
\tikzstyle{util} = [inner sep=1mm]
\tikzstyle{behaveLabel} = [text width=2cm]
\tikzstyle{actionArrow} = [thick, >=Stealth, -{>[scale=0.7]}]
\tikzstyle{showPoint} = [shape=circle, fill=black, minimum size=0.3em]

\newcommand{\drawEdgeL}[5][]{
  \draw[actionArrow, #5]
    (#3)
    --
    node[right, xshift=0.5mm] {\small #1}
    node[anchor=south east, xshift=1mm, yshift=1mm] {\small #2} (#4);
}
\newcommand{\drawEdgeR}[5][]{
  \draw[actionArrow, #5]
    (#3)
    --
    node[left, xshift=-0.5mm] {\small #1}
    node[anchor=south west, xshift=-1mm, yshift=1mm] {\small #2} (#4);
}

\newcommand{\deviationExampleMatrix}{
  \tikzstyle{column 1} = [anchor=north west]
  \tikzstyle{column 2} = [anchor=north]
  \tikzstyle{column 3} = [anchor=north]
  \tikzstyle{column 4} = [anchor=north]
  \tikzstyle{row 1 column 4} = [anchor=north]
  \matrix[row sep=2mm, column sep=0.4cm, ampersand replacement=\&]
}
\newcommand{\thinDeviationExampleMatrix}{
  \tikzstyle{column 1} = [anchor=north]
  \tikzstyle{column 2} = [anchor=north]
  \matrix[column sep=0.1cm, ampersand replacement=\&]
}
\newcommand{\devRecTreeTerminalCoords}[1]{
  \coordinate(#1LCoord) at ($(#1Coord.south west)+(-\cmToEmFactor*0.7em,-\cmToEmFactor*0.75em)$);
  \coordinate(#1RCoord) at ($(#1Coord.south east)+(\cmToEmFactor*0.7em,-\cmToEmFactor*0.75em)$);
}
\newcommand{\devRecTreeInternalCoords}[1]{
  \coordinate(#1LCoord) at ($(#1.south west)+(-\cmToEmFactor*0.9em,-\cmToEmFactor*0.35em)$);
  \coordinate(#1RCoord) at ($(#1.south east)+(\cmToEmFactor*0.9em,-\cmToEmFactor*0.35em)$);
}
\newcommand{\devRecTreeNodeCoords}[1]{
  \devRecTreeInternalCoords{#1};
  \devRecTreeTerminalCoords{#1L};
  \devRecTreeTerminalCoords{#1R};
}
\newcommand{\devRecInternalEdges}[5]{
  \drawEdgeL{#2}{#1}{#1L}{#3};
  \drawEdgeR{#4}{#1}{#1R}{#5};
}
\newcommand{\devRecRootEdges}[5]{
  \devRecInternalEdges{#1}{#2}{#3}{#4}{#5};
}
\newcommand{\devRecTerminalEdges}[5]{
  \drawEdgeL{#2}{#1}{#1L.north}{#3};
  \drawEdgeR{#4}{#1}{#1R.north}{#5};
}
\newcommand{\devRecLEdges}[5]{
  \devRecTerminalEdges{#1L}{#2}{#3}{#4}{#5};
}
\newcommand{\devRecREdges}[5]{
  \devRecTerminalEdges{#1R}{#2}{#3}{#4}{#5};
}

\newcommand{\rowSepLine}[2][thick]{
  \draw [line cap=rect,#1]
    ($(#2-Label.north west)+(0,1mm)$)
    --
    ($(#2-Value.north east)+(0,1mm)$);
}

\newcommand{\mpExampleHeader}{
  \node {\textbf{behavior}}; \&

  \node(rec1Label) {\textbf{recommendation 1:}};
  \node(rec1P2State) [state, right=0.5cm of rec1Label.east, anchor=west] {2};
  \node(rec1P2H) [below left=\shortAYLength and 0.4cm of rec1P2State.south west] {};
  \node(rec1P2T) [below right=\shortAYLength and 0.4cm of rec1P2State.south east] {};
  \drawEdgeL{H}{rec1P2State}{rec1P2H}{follow};
  \drawEdgeR{T}{rec1P2State}{rec1P2T}{alt}; \&

  \node(rec2Label) {\textbf{recommendation 2:}};
  \node(rec2P2State) [state, right=0.5cm of rec2Label.east, anchor=west] {2};
  \node(rec2P2H) [below left=\shortAYLength and 0.4cm of rec2P2State.south west] {};
  \node(rec2P2T) [below right=\shortAYLength and 0.4cm of rec2P2State.south east] {};
  \drawEdgeL{H}{rec2P2State}{rec2P2H}{alt};
  \drawEdgeR{T}{rec2P2State}{rec2P2T}{follow}; \&

  \node(valueLabel) {\textbf{EV}};
}

\newcommand{\botsExampleHeader}{
  \node {\textbf{behavior}}; \&

  \node(rec1Label) {\textbf{recommendation 1:}};
  \node(rec1P2State) [state, right=0.5cm of rec1Label.east, anchor=west] {2};
  \node(rec1P2H) [below left=\shortAYLength and 0.4cm of rec1P2State.south west] {};
  \node(rec1P2T) [below right=\shortAYLength and 0.4cm of rec1P2State.south east] {};
  \drawEdgeL{X}{rec1P2State}{rec1P2H}{zeroProb};
  \drawEdgeR{Y}{rec1P2State}{rec1P2T}{follow}; \&

  \node(rec2Label) {\textbf{recommendation 2:}};
  \node(rec2P2State) [state, right=0.5cm of rec2Label.east, anchor=west] {2};
  \node(rec2P2H) [below left=\shortAYLength and 0.4cm of rec2P2State.south west] {};
  \node(rec2P2T) [below right=\shortAYLength and 0.4cm of rec2P2State.south east] {};
  \drawEdgeL{X}{rec2P2State}{rec2P2H}{follow};
  \drawEdgeR{Y}{rec2P2State}{rec2P2T}{zeroProb}; \&

  \node(valueLabel) {\textbf{EV}};
}

\newcommand{\deviationExampleStates}[3]{
  \node(#1L) [state, #2] at (#1LCoord) {};
  \node(#1R) [state, #3] at (#1RCoord) {};
}

\newcommand{\mpExampleTerminalsOne}[5]{
  \node(#1LL) [util, #2] at (#1LLCoord) {\small $+1$};
  \node(#1LR) [util, #3] at (#1LRCoord) {\small $-1$};
  \node(#1RL) [util, #4] at (#1RLCoord) {\small $-1$};
  \node(#1RR) [util, #5] at (#1RRCoord) {\small $+1$};
}
\newcommand{\mpExampleTerminalsTwo}[5]{
  \node(#1LL) [util, #2] at (#1LLCoord) {\small $-1$};
  \node(#1LR) [util, #3] at (#1LRCoord) {\small $+1$};
  \node(#1RL) [util, #4] at (#1RLCoord) {\small $+1$};
  \node(#1RR) [util, #5] at (#1RRCoord) {\small $-1$};
}

\newcommand{\botsExampleTerminalsOne}[5]{
  \node(#1LL) [util, #2] at (#1LLCoord) {\small $0$};
  \node(#1LR) [util, #3] at (#1LRCoord) {\small $+3$};
  \node(#1RL) [util, #4] at (#1RLCoord) {\small $0$};
  \node(#1RR) [util, #5] at (#1RRCoord) {\small $+2$};
}
\newcommand{\botsExampleTerminalsTwo}[5]{
  \node(#1LL) [util, #2] at (#1LLCoord) {\small $+2$};
  \node(#1LR) [util, #3] at (#1LRCoord) {\small $0$};
  \node(#1RL) [util, #4] at (#1RLCoord) {\small $+1$};
  \node(#1RR) [util, #5] at (#1RRCoord) {\small $0$};
}
 \tikzset{
  seqPath/.style args={segment length #1 amplitude #2}{
    very thick,
    line join=round,
    decorate,
    decoration={zigzag, segment length=#1, amplitude=#2}},
  seqPath/.default={segment length 0.2cm amplitude 1mm}
}

\def\actLen{\cmToEmFactor*0.5em}
\def\seqLen{\actLen}

\tikzset{
    absDecisionDag/.style args={size #1}{very thick, inner sep=0em, fill=white, regular polygon, regular polygon sides=3, minimum size=#1+#1/3.4641},
    absDecisionDag/.default={size \seqLen}
}
\newcommandx{\inTrans}[3][1=\actLen, 2=very thick]{
  \def\inTransActLen{#1}
  \def\inTransActionStyle{#2}
  \def\inTransRoot{#3}

\coordinate(\inTransRoot/dev) at ($(\inTransRoot)-(0,\inTransActLen)$);
  \coordinate(\inTransRoot/info) at ($(\inTransRoot/dev)+(-\inTransActLen,0.05em)$);
  \coordinate(\inTransRoot/influenceSource)
    at ($(\inTransRoot)!0.8!(\inTransRoot/info)+(0.1em, 0)$);
  \coordinate(\inTransRoot/south)
    at ($(\inTransRoot/dev)!0.5!(\inTransRoot/info)$);
  \coordinate(\inTransRoot/north)
    at (\inTransRoot -| \inTransRoot/south);
  \coordinate(\inTransRoot/center)
    at ($(\inTransRoot/north)!0.5!(\inTransRoot/south)$);

  \draw[\inTransActionStyle, rounded corners]
    (\inTransRoot/info)
    edge[info]
    ($(\inTransRoot)-(0,0.05em)$)
    (\inTransRoot)
    edge[dev]
    (\inTransRoot/dev);
}

\newcommandx{\seqInTrans}[3][1=follow, 2=\seqLen]{
  \def\seqInTransStyle{#1}
  \def\seqInTransSeqLen{#2}
  \def\seqInTransRoot{#3}

\coordinate(\seqInTransRoot/seq) at ($(\seqInTransRoot.south)-(0,\seqInTransSeqLen)$);
  \inTrans{\seqInTransRoot/seq}

  \draw[seqPath, \seqInTransStyle] (\seqInTransRoot.north) -- (\seqInTransRoot/seq);
}

\tikzstyle{tree} = [absDecisionDag={size \seqLen}]

\newcommand{\inDevDiagram}[1]{
  \def\root{#1}
  \node(\root/recState)
    [tree, info, anchor=north]
    at ($(\root.south)+(-\cmToEmFactor*0.3em,0)$)
    {};

  \node(\root/devState)
    [tree, dev, anchor=north]
    at ($(\root.south)+(\cmToEmFactor*0.3em,0)$)
    {};
  \draw[infoArrow, bend left=50]
    (\root/recState)
    edge[infoColor]
    (\root/devState);

\coordinate(\root/south) at (\root.north |- \root/recState.south);
  \coordinate(\root/center) at ($(\root.north)!0.5!(\root/south)$);
  \coordinate(\root/west) at (\root/center -| \root/recState.west);
}

\newcommand{\inCsDevDiagram}[1]{
  \def\root{#1}
  \seqInTrans{\root};

  \node(\root/devTree) [tree, devColor, anchor=north] at (\root/seq/dev) {};

\coordinate(\root/south) at (\root/devTree.south);
  \coordinate(\root/center) at ($(\root.north)!0.5!(\root/south)$);
  \coordinate(\root/north east) at (\root.north -| \root/devTree.south east);
}
\newcommand{\inCfDevDiagram}[1]{
  \def\root{#1}
  \seqInTrans[dev]{\root};
  \node(\root/followTree)
    [tree, follow, anchor=north]
    at (\root/seq/dev)
    {};

\coordinate(\root/south) at (\root/followTree.south);
  \coordinate(\root/center) at ($(\root.north)!0.5!(\root/south)$);
}
\newcommand{\inActDevDiagram}[1]{
  \def\root{#1}
  \seqInTrans{\root};
  \node(\root/followTree)
    [tree, follow, anchor=north]
    at (\root/seq/dev)
    {};

\coordinate(\root/south) at (\root/followTree.south);
  \coordinate(\root/center) at ($(\root.north)!0.5!(\root/south)$);
}

 \usepackage[textwidth=2in,colorinlistoftodos,prependcaption,textsize=footnotesize]{todonotes}
\expandafter\newif\csname ifGin@setpagesize\endcsname
\newcommand{\todonote}[4][inline]{\safeIncCounter{#2NoteCounter}
  \todo[color=offWhite,bordercolor=#3,linecolor=#3,#1]{\textbf{\uppercase{#2}$_{\arabic{#2NoteCounter}}$:}~#4}}

\usepackage[normalem]{ulem}

\newcommand{\replaced}[3]{\def\counterPrefix{#1}
  \def\arrowMarker{#2}
  \def\replacedText{#3}
  \todo[color=offWhite,bordercolor=red,inline]{$\bs{\arrowMarker}$ \textbf{Replaced (\arabic{Replaced\counterPrefix{}NoteCounter})} \replacedText }}

\newcommand{\replacedStart}[2]{\def\user{#1}
  \def\text{#2}
  \safeIncCounter{Replaced#1NoteCounter}\replaced{\user}{\downarrow}{\text}}
\newcommand{\replacedEnd}[1]{\def\user{#1}
  \replaced{\user}{\uparrow}{}}

\newcommand{\issue}[3]{\todo[color=black,inline]{\textcolor{white}{$\bs{#2}$ \textbf{Issue \##3} (Part \arabic{Issue#3NoteCounter}) #1}}}

\newcommand{\issueChangeStart}[2][]{\safeIncCounter{Issue#2NoteCounter}\issue{#1}{\downarrow}{#2}}
\newcommand{\issueChangeEnd}[2][]{\issue{#1}{\uparrow}{#2}}

\newcommand{\correction}[2]{\def\arrowMarker{#1}
  \def\correctionLabel{#2}
  \todo[color=offWhite,bordercolor=red,inline]{$\bs{\arrowMarker}$ \textbf{Correction (\arabic{CorrectionNoteCounter})} \correctionLabel }}
\newcommand{\correctionStart}[1]{\def\correctionStartLabel{#1}
  \safeIncCounter{CorrectionNoteCounter}
  \correction{\downarrow}{\correctionStartLabel}}
\newcommand{\correctionEnd}{\correction{\uparrow}{}}
 \renewcommand{\todonote}[4][inline]{\ignorespaces}
\renewcommand{\issueChangeStart}[2][]{\ignorespaces}
\renewcommand{\issueChangeEnd}[2][]{\ignorespaces}
\renewcommand{\replacedStart}[2][]{\ignorespaces}
\renewcommand{\replacedEnd}[1][]{\ignorespaces}
\renewcommand{\correctionStart}[2][]{\ignorespaces}
\renewcommand{\correctionEnd}[1][]{\ignorespaces}
 
\usepackage{nicefrac}
\usepackage[capitalize,noabbrev]{cleveref}

\begin{document}
\maketitle

\frenchspacing  

\begin{abstract}
  Driven by recent successes in two-player, zero-sum game solving and playing, artificial intelligence work on games has increasingly focused on algorithms that produce equilibrium-based strategies.
  However, this approach has been less effective at producing competent players in general-sum games or those with more than two players than in two-player, zero-sum games.
  An appealing alternative is to consider adaptive algorithms that ensure strong performance in hindsight relative to what could have been achieved with modified behavior.
  This approach also leads to a game-theoretic analysis, but in the correlated play that arises from joint learning dynamics rather than factored agent behavior at equilibrium.
  We develop and advocate for this \emph{hindsight rationality} framing of learning in general sequential decision-making settings.
  To this end, we re-examine \emph{mediated equilibrium} and \emph{deviation types} in \emph{extensive-form games}, thereby gaining a more complete understanding and resolving past misconceptions.
  We present a set of examples illustrating the distinct strengths and weaknesses of each type of equilibrium in the literature, and prove that no tractable concept subsumes all others.
  This line of inquiry culminates in the definition of the deviation and equilibrium classes that correspond to algorithms in the \emph{counterfactual regret minimization} (\emph{CFR}) family, relating them to all others in the literature.
  Examining CFR in greater detail further leads to a new recursive definition of rationality in correlated play that extends \emph{sequential rationality} in a way that naturally applies to hindsight evaluation.
\end{abstract}

\section{Introduction}
\noindent
Algorithms to approximate maximin or Nash equilibrium strategies have fueled major successes in the autonomous play of human-scale games like
Go~\parencite{Silver16Go,Silver17AGZ,Silver18AlphaZero},
chess and shogi~\parencite{Silver17AChess,Silver18AlphaZero},
poker~\parencite{moravvcik2017deepstack,brown2018superhuman,brown2019superhuman} and StarCraft~\parencite{vinyals2019grandmaster}.
With the exception of multi-player poker addressed by \textcite{brown2019superhuman}, all of these are two-player, zero-sum games.
Less success has been achieved in games outside this space.

While both maximin and Nash equilibrium concepts apply in general-sum, multi-player games, it is not clear that either should be fielded against unknown strategies in these situations.
The maximin concept assumes all other players collude against you, which can produce overly pessimistic strategies.
Nash equilibria make the less pessimistic assumption that each player is rational and independent, but strategies from such equilibria have no performance guarantees against arbitrary strategies and are hard to compute.
Perhaps one reason that AI has struggled to find the same success in multi-agent, general sum games is the lack of a strong theoretical base like that provided by the maximin objective in zero-sum games.

An appealing alternative is to consider adaptive algorithms that ensure strong performance in hindsight relative to what could have been achieved with modified behavior (\eg/, \textcite{hannan1957approximation}).
We advocate for a \emph{hindsight rationality} framing of learning in sequential decision-making settings where rational correlated play, rather than factored equilibrium behavior, arises naturally through repeated play.
Hindsight rationality is the dynamic interpretation of static \emph{mediated equilibrium}~\parencite{aumann1974ce} concepts that have led to elegant decentralized learning algorithms like \emph{regret matching}~\parencite{Hart00} and connections between learning and statistical calibration~\parencite{foster1997calibrated}.

Mediated equilibria in EFGs have gained increasing interest recently with the introduction of
\emph{extensive-form (coarse\nobreakdash-)correlated equilibrium} (\emph{EF(C)CE})~\parencite{forges2002efce,von2008efce-complexity,farina2020efcce}
and \emph{agent-form~\parencite{selten1974af-and-trembling-hand-eq} (coarse\nobreakdash-)correlated equilibrium} (\emph{AF(C)CE})~\parencite{von2008efce-complexity},
along with related algorithms~\parencite{celli2019learning,celli2020noregret}.
However, there has yet to be a clear description of the general mediated equilibrium landscape in EFGs beyond \emph{causal} and \emph{action} \emph{deviations}~\parencite{von2008efce-complexity,Dudik09causal-dev}.
There has also yet to be a clear proposal about how this line of inquiry can be used to design more competent game playing algorithms.
We show how hindsight rationality operationalizes these concepts in concrete \emph{regret minimization} objectives while preserving equilibria as descriptions of jointly rational learning.

We present a more complete picture of the relationships between the deviation types and equilibrium classes discussed in the literature and resolve the misconception that all EFCEs are AFCEs\footnote{We use the definition of EFCE as used in the majority of the literature in artificial intelligence, which differs from von Stengel's intended definition (Personal communication 2020). See Footnote \efceFootnote/.}.
A central component of this analysis is a set of examples illustrating the distinct strengths and weaknesses of each type of equilibrium.
The resulting table of equilibrium relationships reveals that while some such concepts are stronger than others, no tractable concept subsumes all others.

\correctionStart{Fix wording and add footnote.}
We define \emph{counterfactual deviations} to explicitly describe CFR's underlying deviation type.
We illustrate the unique strengths and weaknesses of counterfactual deviations in examples where there is a beneficial counterfactual deviation in an EFCE and an AFCE, and where there are beneficial causal, action, and external deviations in a CFCE.

A core feature of CFR is its recursive structure, and we present a new extension of sequential rationality~\parencite{KrepsWilson82} to correlated play in order to better characterize CFR's behavior.
This \emph{observable sequential rationality} applies naturally to hindsight evaluation and thus algorithms like CFR.
Observable sequential rationality turns out to be the key ingredient that elevates the strength of counterfactual deviations so that they subsume external deviations.\hsrIntroCorrectionsFootnote{}We show that CFR is observably sequentially hindsight rational with respect to counterfactual deviations and provide an example where a \emph{correlated equilibrium} (\emph{CE})~\parencite{aumann1974ce} is not an observable sequential \emph{counterfactual coarse-correlated equilibrium} (\emph{CFCCE}).
Considering that it is generally intractable to compute a CE in EFGs, it is perhaps surprising that CFR, a simple and efficient algorithm, behaves in accordance with an equilibrium class that is not subsumed by a CE.
\correctionEnd
 
\section{Extensive-Form Games}
As we are interested in multi-agent, sequential decision-making, we use the \emph{extensive-form game} (\emph{EFG}) framework as the basis of our technical discussions.

The set of finite-length \emph{histories}, $\Histories$, defines all ways that an EFG can partially play out, \ie/, the game begins with the \emph{empty history}, $\emptyHistory$, and given a history $h$, the successor history after taking \emph{action} $a \in \Actions(h)$ is $ha$.
$\Actions(h)$ is the finite set of actions available at $h$ for the player acting at $h$, determined by the \emph{player function},
$\playerChoice : \Histories \setminus \TerminalHistories \to \set{1, \dots, N} \cup \set{\chance}$,
where $N > 0$ is the number of players and $\chance$ is the \emph{chance player}.
Let $\Histories_i = \set{ h \in \Histories \where \playerChoice(h) = i }$ be the set of histories where player $i$ must act.
We write $h \sqsubseteq h'$ to state that $h$ is a predecessor of $h'$.
A \emph{terminal history}, $z \in \TerminalHistories \subseteq \Histories$, has an empty action set and no successors, and each history must lead to at least one terminal history (since all histories have finite length), at which point rewards are given to each player according to a bounded \emph{utility function}, $\utility_i : \TerminalHistories \to [-\maxReward, \maxReward]$.

To model imperfect information, histories are partitioned into \emph{information sets},
$\InfoSets_i = \{ \infoSet \where \infoSet \subseteq \Histories_i \}$,
where players are unable to distinguish between histories within an information set.
Thus, the set of actions at each history within an information set must, must coincide, \ie/,
$\Actions(\infoSet) \as \Actions(h) = \Actions(h')$,
for any $h, h' \in \infoSet$.
We restrict ourselves to \emph{perfect-recall} information partitions, which are those that ensure players never forget the information sets they encounter during play.
Under this assumption, each player's information-set transition graph forms a forest (not a tree since other players may act first) and are partially ordered.
We use $\infoSet \preceq \infoSet'$ to denote that $\infoSet \in \InfoSets_i$ is a predecessor of $\infoSet' \in \InfoSets_i$.

A \emph{pure strategy} for player $i$ maps each information set to an action.
We denote the finite set of such strategies as
$\PureStratSet_i = \set{ \pureStrat_i : \pureStrat_i(\infoSet) \in \Actions(\infoSet), \, \forall \infoSet \in \InfoSets_i }$
and pure \emph{strategy profiles} with a strategy for each player by
$\PureStratSet = \bigtimes_{i \in \set{1, \dots, N, c}} \PureStratSet_i$.
If we let $-i$ denote the set of players other than player $i$,
we can likewise define $\PureStratSet_{-i} = \bigtimes_{j \neq i} \PureStratSet_j$ as the set of strategies for all players except $i$.

Randomized behavior can be represented with a \emph{mixed strategy},
$\strat_i \in \StrategySet_i = \simplex^{\abs{\PureStratSet_i}}$,
which is a distribution over pure strategies, or a \emph{behavioral strategy}~\parencite{Kuhn53}, which is an assignment of \emph{immediate strategies},
$\strat_i(\infoSet) \in \simplex^{\abs{\Actions(\infoSet)}}$,
to each information set\footnote{$\simplex^d$ is the $d$-dimensional probability simplex.}.
Perfect recall ensures \emph{realization equivalence} between the set of mixed and behavioral strategies~\parencite{Kuhn53}.
That is, for every mixed strategy there is a behavioral strategy, and \viceversa/, that applies the same weight to each terminal history.
Thus, we treat mixed and behavioral strategies as interchangeable representations where $\strat_i(\pureStrat_i)$ is the probability of sampling pure strategy $\pureStrat_i$ and $\strat_i(a \given \infoSet)$ is the probability of playing action $a$ in information set $\infoSet$.
Noticing as well that a pure strategy is just a behavioral strategy with deterministic immediate strategies, we default to the behavioral representation unless otherwise specified.

The \emph{reach probability function} $\reachProb$ describes the probability of playing from one history to another, following from the familiar chain rule of probability.
Let $\infoSet(h)$ denote the unique information set that contains history $h$.
Then the probability that player $i$ plays from $h$ to $h'$, through $ha$ when $h \ne h'$, according to strategy $\strat_i$, is
\begin{align*}
  \reachProb(h, h'; \strat_i) = \begin{dcases}
    1 &\mbox{if } h = h'\\
    0 &\mbox{if } h \not\sqsubseteq h'\\
    \strat_i(a \given \infoSet(h)) \reachProb(ha, h'; \strat_i) &\mbox{if } \playerChoice(h) = i\\
    \reachProb(ha, h'; \strat_i) &\mbox{o.w.}
  \end{dcases}
\end{align*}
The probability of reaching $h$ from the beginning of the game is $\reachProb(h; \strat_i) \as \reachProb(\emptyHistory, h; \strat_i)$,
and we overload $\reachProb$ for strategy tuples, \ie/,
$\reachProb(h, h'; \strat_{-i}) \as \prod_{j \neq i} \reachProb(h, h'; \strat_j)$
is the probability that players except for $i$ play from $h$ to $h'$ and
$\reachProb(h, h'; \strat) \as \reachProb(h, h'; \strat_i) \reachProb(h, h'; \strat_{-i})$.
The terminal reach probabilities of a strategy profile describes its distribution over game outcomes, so the expected utility for player $i$ is
$\utility_i(\strat) = \sum_{z \in \TerminalHistories} \reachProb(z; \strat) \utility_i(z)$.
 
\section{Factored, Correlated, and Online Play}
Factored play is that where players act entirely independently from one another.
A \emph{Nash equilibrium}~\parencite{Nash1951} models jointly rational play under such factored play.
It is a mixed strategy profile where no player can benefit from a unilateral deviation, \ie/, $\strat$ where
$\utility_i(\strat'_i, \strat_{-i}) \le \utility_i(\strat)$ for all $\strat'_i$.

In a two-player, zero-sum game, every Nash equilibrium is a pair of \emph{maximin strategies}, which are those that maximize the player's minimum value against any opponent.
In games with more than two players or general-sum payoffs, a strategy from a Nash equilibrium does not confer the same guarantee.
Additionally, a maximin strategy for player $i$ makes the pessimistic assumption that all other players will collude against $i$.
These deficiencies leave us wanting for a theoretical framework that is better suited to multi-player, general-sum games, upon which to design game playing algorithms.

We suggest that \emph{hindsight rationality}, which is the idea of learning correlated play over time that approaches optimality in hindsight, can serve as such a framework.
This concept arises from an old connection between \emph{mediated equilibria}~\parencite{aumann1974ce} and \emph{online learning} (\eg/, \textcite{hannan1957approximation}).

\subsection{Mediated Equilibria}
A mediated equilibrium is a generalization of Nash equilibrium to correlated play.
It is a \emph{pure strategy profile distribution}--\emph{deviation set} profile  pair,
\begin{align*}
  \tuple*{\recDist \in \simplex^{\abs{\PureStratSet}}, \tuple*{\DevSet_i \subseteq \set*{ \dev \where \dev : \PureStratSet_i \to \PureStratSet_i }}_{i = 1}^N},
\end{align*}
where there is no beneficial unilateral deviation in any player's deviation set.
The benefit to player $i$ of a deviation is the expected utility under the deviation compared to that of following the \emph{recommendations} (strategies) sampled for player $i$ from the recommendation distribution $\recDist$, \ie/,
$\E_{\pureStrat \sim \recDist}[\utility_i(\dev(\pureStrat_i), \pureStrat_{-i}) - \utility_i(\pureStrat)]$.
A \emph{rational} player follows the recommendations from a mediated equilibrium since their value is optimal (with respect to deviation set $\DevSet_i$) given the play of the other players.
The deviation set constrains the extent to which individual strategy modifications can condition on input strategies, thereby providing a mechanism for varying the strength and character of equilibrium rationality constraints.

For example, an equilibrium recommendation distribution under the set of \emph{external transformations}, which are those that ignore their inputs, $\DevSet^{\EXT}_{\PureStratSet_i} = \set{ \dev \where \exists \pureStrat'_i, \, \dev(\pureStrat_i) = \pureStrat'_i, \, \forall \pureStrat_i}$, is a \emph{coarse-correlated equilibrium} (\emph{CCE})~\parencite{moulin1978cce}.
In contrast, a \emph{correlated equilibrium} (\emph{CE})~\parencite{aumann1974ce} is one where the deviation set is unconstrained, which ensures that all CEs are CCEs.
The unconstrained set of deviations is called the set of \emph{swap transformations} and is denoted $\DevSet^{\SWAP}_{\PureStratSet_i}$
It turns out that the smaller set of \emph{internal transformations}~\parencite{foster1999int-regret}, which are those that act like the identity transformation except given one particular strategy, \ie/,
\begin{align*}
  \DevSet^{\INT}_{\PureStratSet_i} = \set{
    \dev \where
      \exists \pureStrat_i, \pureStrat'_i, \,
      \dev(\pureStrat_i) = \pureStrat'_i, \,
      \dev(\pureStrat''_i) = \pureStrat''_i, \,
      \forall \pureStrat''_i \neq \pureStrat_i
    },
\end{align*}
also corresponds to CE~\parencite{greenwald2003general}.
In spite of this size reduction, computing a CE in EFGs is intractable since the pure strategy space grows exponentially with the size of the game.
 
\subsection{Online Learning}

In an online learning setting, a learner repeatedly plays a game with unknown, dynamic, possibly adversarial players.
On each round $1 \le t \le T$, the learner who acts as player $i$ chooses a behavioral strategy, $\strat_i^t$, simultaneously with the other players who in aggregate choose $\strat^t_{-i}$.
The learner is then evaluated based on the expected value they achieve under the selected strategy profile, $\utility_i(\strat^t)$.

In this context, we can consider deviations to the learner's play in hindsight.
A swap transformation generates a corresponding mixed strategy transformation, so we overload $\dev(\strat_i)$ as the mixed strategy that assigns probability
$[\dev\strat_i](\pureStrat_i) = \sum_{\pureStrat'_i, \dev(\pureStrat'_i) = \pureStrat_i} \strat_i(\pureStrat'_i)$ to each $\pureStrat_i$ (where $[\dev\strat_i] \as \dev(\strat_i)$).
The cumulative benefit of deviating from the learner's strategy sequence, $\tuple{\strat^t_i}_{t = 1}^T$, is then
\begin{align}
  \label{eq:regret}
  \regret^{1:T}(\dev) = \sum_{t = 1}^T \utility_i(\dev(\strat_i), \strat^t_{-i}) - \utility_i(\strat^t).
\end{align}
We describe \cref{eq:regret} as the cumulative \emph{regret} the learner suffers for failing to modify their play according to $\dev$.
We denote the \emph{maximum regret} under the set of deviations
$\DevSet$ as
$\Regret^T(\DevSet) = \max_{\dev \in \DevSet} \regret^{1:T}(\dev)$.
A \emph{no-$\DevSet$-regret} algorithm ensures its average (positive) regret for any deviation in $\DevSet$ approaches zero over time.

In a two-player, zero-sum game, two competing no-regret learners generate strategy profiles that converge toward a Nash equilibrium when marginalized as
$\tuple{\frac{1}{T} \sum_{t = 1}^T \strat^t_i}_{i = 1}^N$.
However, marginalized empirical strategy profiles do not generally converge toward mediated or Nash equilibria~\parencite{Shapley}, nor are they necessarily maximin strategies.
Instead, hindsight rationality (as we discuss next) suggests an online learning approach to game playing since it makes a statement about the optimality of play that is actually observed, irrespective of how rational the other players behave.

\subsection{Hindsight Rationality}

The $T$ strategy profiles deployed after $T$ rounds of play forms an \emph{empirical distribution of play},
$\recDist^T \in \simplex^{\abs{\PureStratSet}}$.
That is, for all pure strategy profiles $\pureProfile$,
$\recDist^T(\pureProfile) = \frac{1}{T} \sum_{t = 1}^T \strat^t(\pureProfile)$,
according to the probability of sampling $\pureProfile$ on each round, $\strat^t(\pureProfile) = \prod_i \strat^t_i(\pureStrat_i)$.
If we treat the empirical distribution of play as a recommendation distribution, the benefit to player $i$ of deviation $\dev$ is then exactly player $i$'s average regret per-round, \ie/,
$\E_{\pureStrat \sim \recDist^T}[\utility_i(\dev(\pureStrat_i), \pureStrat_{-i}) - \utility_i(\pureStrat)] = \frac{1}{T} \regret^{1:T}(\dev)$.
A learner participating in online play as player $i$ is thus rational in hindsight with respect to $\DevSet$ if they have no regret, where we say ``in hindsight'' because $\recDist^T$ is constructed from past play.
Furthermore, a learner with with $\gap = \Regret^T(\DevSet)/T$ average regret is approximately rational in hindsight (precisely $\gap$-rational in hindsight).
Thus, a no-regret learner's behavior approaches exact rationality in hindsight as they gain experience.
We introduce the term \emph{hindsight rational} as an equivalent alternative to ``no-regret'' for situations where it is useful to emphasize this shift in perspective, from learning to perform as well as strategy modifications, to learning to correlate optimally with other players.

\section{EFG Deviation Types}

\subsubsection{Causal deviations:}
An \emph{extensive-form correlated equilibrium} (\emph{EFCE})~\parencite{forges2002efce,von2008efce-complexity} is an equilibrium without beneficial \emph{informed causal deviations}~\parencite{Dudik09causal-dev}.
A causal deviation follows the recommended strategy, $\pureStrat_i$, unless it is triggered at a particular information set, $\infoSet^{\TRIGGER} \in \InfoSets_i$, at which point an alternative strategy, $\pureStrat_i^{\TRIGGER}$, is used to play out the rest of the game.\footnote{von Stengel (Personal communication 2020) explains that it was intended for players in an EFCE to receive action recommendations at each information set regardless of earlier behavior.
However, all work that we are aware of using the EFCE concept assumes the causal deviation structure codified by \textcite{Dudik09causal-dev} where players commit to a fixed strategy after being triggered.
This structure arises from \textcite{von2008efce-complexity}'s discussion of EFCEs in the context of recommendation distributions over reduced strategies that only specify the recommendations at information sets that can be reached by following earlier recommendations. In this case, a player receives arbitrary deterministic ``uninformative'' recommendations after deviating.}
The immediate strategy of an informed causal deviation at a given information set is
\begin{align*}
  [\dev\pureStrat_i](\infoSet) = \begin{cases}
    \pureStrat_i^{\TRIGGER}(\infoSet)& \mbox{if } \pureStrat_i(\infoSet^{\TRIGGER}) = a^{\TRIGGER} \mbox{ and } \infoSet^{\TRIGGER} \preceq \infoSet \\
    \pureStrat_i(\infoSet)& \mbox{o.w.},
  \end{cases}
\end{align*}
where a trigger action, $a^{\TRIGGER}$, must be recommended in $\infoSet^{\TRIGGER}$ for the deviation to be triggered.
A \emph{blind causal deviation} is instead always triggered upon reaching $\infoSet^{\TRIGGER}$ (it does not get to ``see'' $\pureStrat_i(\infoSet^{\TRIGGER})$).
An \emph{extensive-form coarse-correlated equilibrium} (\emph{EFCCE})~\parencite{farina2020efcce} is an equilibrium without beneficial blind causal deviations.

\subsubsection{Action deviations:}
\emph{Action deviations} correspond to \emph{agent-form} equilibria~\parencite{von2008efce-complexity}.
They trigger like causal deviations but they resume following the recommended strategy (and thereby re-correlate with the sampled recommendations) after a single action, as if a different agent is in control at each information set.
An \emph{informed action deviation} is defined as
\begin{align*}
  [\dev\pureStrat_i](\infoSet) = \begin{cases}
    \pureStrat_i^{\TRIGGER}(\infoSet)& \mbox{if } \pureStrat_i(\infoSet^{\TRIGGER}) = a^{\TRIGGER} \mbox{ and } \infoSet = \infoSet^{\TRIGGER}\\
    \pureStrat_i(\infoSet)& \mbox{o.w.}
  \end{cases}
\end{align*}
and no such deviation is beneficial in an \emph{agent-form correlated equilibrium} (\emph{AFCE}).
As above, a \emph{blind action deviation} is always triggered upon reaching the trigger information set and its equilibrium concept is \emph{agent-form coarse-correlated equilibrium} (\emph{AFCCE}).

\subsubsection{Intuition:}
For some intuition about these deviation types, imagine a driver traveling through a city with route recommendations.
A causal deviation follows the recommendations until they reach a particular part of the city and then they ignore the recommendations for the rest of the trip.
An action deviation also follows the recommendations to a particular part of the city but then deviates from the recommendations at a single intersection before following them again.
Later, we introduce \emph{counterfactual deviations}, which ignore recommendations until reaching a particular part of the city and then begin following recommendations from there.

\correctionStart{Change the counterfactual deviation to external deviation arrow to be dashed and add a line to the caption explaining the notation.}
For a visual metaphor, see \cref{fig:dev-diagram}, where we visualize each of these deviation types alongside the previously described external and internal transformations.
\correctionEnd

\renewcommand{\inDevDiagram}[1]{
  \def\root{#1}
  \node(\root/recState)
    [absDecisionDag, info, anchor=north]
    at ($(\root.south)+(-0.3cm,-0.2cm)$)
    {};
  \draw[info] (\root.north) -- (\root/recState.north);

  \node(\root/devState)
    [absDecisionDag, dev, anchor=north]
    at ($(\root.south)+(0.3cm,-0.2cm)$)
    {};
  \draw[dev] (\root.north) -- (\root/devState.north);

\coordinate(\root/south) at (\root.north |- \root/recState.south);
  \coordinate(\root/center) at ($(\root.north)!0.5!(\root/south)$);
}
\renewcommand{\inCsDevDiagram}[1]{
  \def\root{#1}
  \node(\root/devRoot) [anchor=north] at ($(\root.south)+(0,-\informedSeqLen)$) {};
  \draw[seqPath, follow] (\root.north) -- (\root/devRoot.north);

  \node(\root/devTree) [absDecisionDag, dev, anchor=north] at ($(\root/devRoot.south)+(0,-\informedActLen)$) {};
  \draw[info] (\root/devRoot.north) -- ($(\root/devTree.north)+(-0.3cm,0)$);
  \draw[dev] (\root/devRoot.north) -- (\root/devTree.north);

\coordinate(\root/south) at (\root/devTree.south);
  \coordinate(\root/center) at ($(\root.north)!0.5!(\root/south)$);
  \coordinate(\root/north east) at (\root.north -| \root/devTree.south east);
}
\renewcommand{\inCfDevDiagram}[1]{
  \def\root{#1}
  \node(\root/infoDevRoot) [anchor=north] at ($(\root.south)+(0,-\informedSeqLen)$) {};
  \draw[seqPath, dev] (\root.north) -- (\root/infoDevRoot.north);

  \node(\root/followTree)
    [absDecisionDag, follow, anchor=north]
    at ($(\root/infoDevRoot.south)+(0,-\informedActLen)$)
    {};
  \draw[info] (\root/infoDevRoot.north) -- ($(\root/followTree.north)+(-0.3cm,0)$);
  \draw[dev] (\root/infoDevRoot.north) -- (\root/followTree.north);
}
\renewcommand{\inActDevDiagram}[1]{
  \def\root{#1}
  \node(\root/infoDevRoot) [anchor=north] at ($(\root.south)+(0,-\informedSeqLen)$) {};
  \draw[seqPath, follow] (\root.north) -- (\root/infoDevRoot.north);

  \node(\root/followTree)
    [absDecisionDag, follow, anchor=north]
    at ($(\root/infoDevRoot.south)+(0,-\informedActLen)$)
    {};
  \draw[info] (\root/infoDevRoot.north) -- ($(\root/followTree.north)+(-0.3cm,0)$);
  \draw[dev] (\root/infoDevRoot.north) -- (\root/followTree.north);

\coordinate(\root/south) at (\root/followTree.south);
  \coordinate(\root/center) at ($(\root.north)!0.5!(\root/south)$);
}

\def\informedActLen{0.33cm}
\def\informedSeqLen{0.4cm}
\begin{figure}[t]
\centering
\begin{tikzpicture}[inner sep=0cm]
  \tikzstyle{arrow} = [thick, >=Straight Barb, -{>[scale=0.7]}]
  \tikzstyle{column 1} = [anchor=west]
  \tikzstyle{osrArrow} = [arrow, dashed]

  \matrix(fig) [column sep=0cm, row sep=0.5cm] {
    \node[yshift=-0.4cm] {internal}; &[0.4cm] &
    \coordinate(inRoot);
    \inDevDiagram{inRoot}; \\

\node(informedLabel) [yshift=-0.5cm] {informed}; &[0.2cm]

\coordinate(iCauseRoot);
    \node(iCauseRoot/label)
      [anchor=west]
      at ($(iCauseRoot.south)+(0.1cm,-0.49cm)$)
      {\small causal};
    \inCsDevDiagram{iCauseRoot}; &

\coordinate(iCfRoot);
    \inCfDevDiagram{iCfRoot};
    \node(label)
      [anchor=west]
      at ($(iCfRoot.south)+(0.1cm,-0.49cm)$)
      {\small counterfactual};
    &[0.15cm]

\coordinate(iActRoot);
    \inActDevDiagram{iActRoot};
    \node(label) [anchor=west] at ($(iActRoot.south)+(0.1cm,-0.49cm)$) {\small action};
    \\

\node(blindLabel) [yshift=-0.4cm] {blind}; &
    {
      \node(bCauseRoot) {};
      \node(devState) [absDecisionDag, dev, anchor=north] at ($(bCauseRoot.south)+(0,-\informedSeqLen-\informedActLen)$) {};
      \draw[seqPath, follow] (bCauseRoot.north) -- (devState.north);
      \node(label) [anchor=west] at ($(bCauseRoot.south)+(0.2cm,-0.38cm)$) {\small causal};
    } &
    {
      \node(bCfRoot) {};
      \node(followState) [absDecisionDag, follow, anchor=north] at ($(bCfRoot.south)+(0,-\informedSeqLen-\informedActLen)$) {};
      \draw[seqPath, dev] (bCfRoot.north) -- (followState.north);
      \node(label) [anchor=west] at ($(bCfRoot.south)+(0.2cm,-0.38cm)$) {\small counterfactual};
    } &
    {
      \node(bActRoot) {};
      \node(devState) [anchor=north] at ($(bActRoot.north)+(0,-\informedSeqLen)$) {};
      \draw[seqPath, follow] (bActRoot.north) -- (devState.north);

      \node(followState) [absDecisionDag, follow, anchor=north] at ($(devState.south)+(0,-\informedActLen)$) {};
      \draw[dev] (devState.north) -- (followState.north);
      \node(label) [anchor=west] at ($(bActRoot.south)+(0.1cm,-0.38cm)$) {\small action};
    } \\
  };
  \node(exLabel) [below=1.1cm of blindLabel.south west, anchor=west, xshift=0.01cm] {external};
  \node(exRoot) [right=1.45cm of exLabel.east, yshift=-0.15cm] [absDecisionDag, dev] {};

  \draw[arrow]
    ($(inRoot/center)!0.3!(iCauseRoot/center)$)
    --
    ($(iCauseRoot/center)!0.25!(inRoot/center)+(0,0.1em)$);

  \coordinate(inBot) at ($(inRoot.north)+(0,-0.8cm)$);
  \draw[arrow] (inBot) -- ($(iCfRoot.north)+(0,0.15cm)$);
  \draw[arrow]
    ($(inRoot/center)!0.25!(iActRoot/center)$)
    --
    ($(iActRoot/center)!0.15!(inRoot/center)+(0,0.3em)$);

  \draw[arrow] ($(iCauseRoot.north)+(0,-1.4cm)$) -- ($(bCauseRoot.north)+(0,0.1cm)$);
  \draw[arrow] ($(iCfRoot.north)+(0,-1.4cm)$) -- ($(bCfRoot.north)+(0,0.1cm)$);
  \draw[arrow] ($(iActRoot.north)+(0,-1.4cm)$) -- ($(bActRoot.north)+(0,0.1cm)$);

  \coordinate(exTop) at ($(exRoot.north)+(0,0.1cm)$);
  \draw[arrow] ($(bCauseRoot.north)+(0.1cm,-1.35cm)$) -- ($(exRoot.north west)+(-0.1cm,0.05cm)$);
  \draw[osrArrow] ($(bCfRoot.north)+(-0.1cm,-1.35cm)$) -- ($(exRoot.north east)+(0.1cm,0.05cm)$);
\end{tikzpicture}
  \caption{A summary of deviation types and their relationships in EFGs.
    The game plays out from top to bottom.
    Straight lines are actions, zigzags are action sequences, and triangles are decision trees.
    Decisions where recommendations are followed are colored black, alterations to recommendations are colored red, and trigger information is colored cyan.
    Arrows denote ordering from a stronger to a weaker deviation type (and therefore a subset to superset equilibrium relationship), the dashed arrow denotes that this relationship holds only under observable sequential rationality.}
  \label{fig:dev-diagram}
\end{figure}
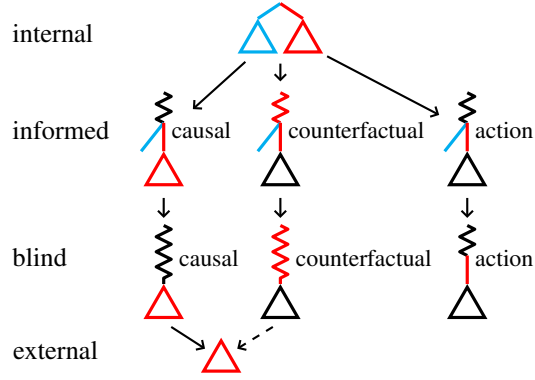

\section{Practical Implications}

We now illustrate some practical differences between tractable deviation types.
These differences are important because algorithms tied to more limited deviation types may not achieve as much reward as ones using stronger deviation types, and the types of mistakes an algorithm makes can depend on the structure of its associated deviation type.
Moreover, even if an algorithm is not designed with a deviation type in mind, it may implicitly use one.
For example, \emph{policy gradient policy iteration} (\emph{PGPI})~\parencite{srinivasan2018actor} and Monte Carlo approximations thereof (\ie/, standard formulations of policy gradient~\parencite{Williams1992reinforce,Sutton00PolicyGradient}) are implicitly tied to action deviations, as we later discuss).
Thus, these results have substantial generality and widespread impact.

\correctionStart{Correct table.}
We summarize our results in a table of relationships between equilibrium concepts (\cref{tab:eq-class-rel}).
\correctionEnd
The bottom two rows and the two rightmost columns reference new equilibrium concepts that will be explained later.

\begin{table}
  \centering
  \begin{tabular}{l|cccc|cc}
    $\downarrow$ implies $\rightarrow$ & CE & CCE & EF & AF & CF & OS-CF \\
    \toprule \addlinespace[0pt]
    CE & \eqTableEntry{yes}{${=}$} & \eqTableEntry{yes}{${\checkmark}$} & \eqTableEntry{yes}{${\checkmark}$} & \eqTableEntry{yes}{${\checkmark}$} & \eqTableEntry{newYes}{$\bs{\checkmark}$} & \eqTableEntry{newNo}{\textbxf{S}} \\
    CCE & \eqTableEntry{no}{M} & \eqTableEntry{yes}{${=}$} & \eqTableEntry{no}{{M}} & \eqTableEntry{no}{\textit{B}} & \eqTableEntry{newNo}{\textbxf{B}} & \eqTableEntry{newNo}{\textbxf{B}} \\
    EF & \eqTableEntry{no}{B} & \eqTableEntry{yes}{$\checkmark$} & \eqTableEntry{yes}{${=}$} & \eqTableEntry{no}{\textit{B}} & \eqTableEntry{newNo}{\textbxf{B}} & \eqTableEntry{newNo}{\textbxf{B}} \\
    AF & \eqTableEntry{no}{I} & \eqTableEntry{no}{I} & \eqTableEntry{no}{I} & \eqTableEntry{yes}{${=}$} & \eqTableEntry{newNo}{\textbxf{I}} & \eqTableEntry{newNo}{\textbxf{I}} \\
    \specialrule{\lightrulewidth}{0pt}{0pt}
    CF & \eqTableEntry{newNo}{\textbxf{M}} & \eqTableEntry{newNo}{\textbxf{R}} & \eqTableEntry{newNo}{\textbxf{M}} & \eqTableEntry{newNo}{\textbxf{M}} & \eqTableEntry{newYes}{$\bs{=}$} & \eqTableEntry{newNo}{\textbxf{S}} \\
    OS-CF & \eqTableEntry{newNo}{\textbxf{M}} & \eqTableEntry{newYes}{$\bs{\checkmark}$} & \eqTableEntry{newNo}{\textbxf{M}} & \eqTableEntry{newNo}{\textbxf{M}} & \eqTableEntry{newYes}{$\bs{\checkmark}$} & \eqTableEntry{newYes}{$\bs{=}$} \\
    \addlinespace[-1.7pt] \bottomrule
  \end{tabular}
  \caption{Equilibrium class relationships.
    The relationships between the coarse-correlated and correlated versions of each EFG equilibrium concept are the same, \eg/, the table is identical if ``EF'' is replaced with ``EFCE'' or ``EFCCE''.
    Cyan cells highlight where the row concept implies the column concept (\eg/, $\text{EF} \subseteq \text{CCE}$) either by deduction ($\checkmark$) or by being the same ($=$).
    Otherwise, the cell is colored red and references a game example that proves the row concept does not imply the column concept (\eg/, $\text{EF} \not\subseteq \text{AF}$), where
    ``I'' $\to$ In or Out, ``B'' $\to$ extended battle-of-the-sexes, ``M'' $\to$ extended matching pennies, ``S'' $\to$ sequential extended battle-of-the-sexes, and ``R'' $\to$ MacQueen's counterexample.
    Emphasized results were previously unclear (italicized entries, \eg/, $\text{EF} \not\Rightarrow \text{AF}$) or undefined (bold entries, \eg/, $\text{CF} \not\Rightarrow \text{EF}$).}
  \label{tab:eq-class-rel}
\end{table}
 
Since causal deviations can trigger at the start of the game, they cannot be weaker than external ones, and \textcite{von2008efce-complexity} showed that they are indeed stronger.
We provide another example proving this relationship in an extended matching pennies game that we later describe.
It is visualized in \cref{fig:mp-dev}.

Section 2.4 of \textcite{von2008efce-complexity} shows how external and causal deviations can outperform action deviations in the \emph{In or Out} game.
If the player chooses In, they choose again between In and Out.
If they ever choose Out, the game ends and they receive a reward of zero.
If they choose In twice, they receive a reward of one.
Given the recommendation $\tuple{\text{Out}, \text{Out}}$, an action deviation can swap a single Out to an In, but not both.
External and causal deviations can achieve a better value by deviating to $\tuple{\text{In}, \text{In}}$.

However, action deviations can outperform both external and causal deviations in other scenarios, ensuring that some EFCEs and CCEs are not AFCEs.
This is a different conclusion than \textcite{von2008efce-complexity} reach under their original EFCE definition.

\subsection{Action Deviations can Outperform Causal Deviations}
\label{sec:ext-bots}

Consider an extended battle-of-the-sexes game where player one first privately chooses between normal or upgraded seating (\NotUpgrade/ or \Upgrade/) for both players at no extra cost before choosing to attend event \EventX/ or \EventY/ simultaneously with their partner.
Player one prefers \EventY/ and player two prefers \EventX/, but if they attend different events, neither is happy so they both get zero.
Attending the less preferred event with their partner gives that player $+1$ with normal seating or $+2$ with upgraded seating, while the other player, who is at their preferred venue, gets $+2$ or $+3$.

The joint recommendations
$\set*{
  \tuple*{\mNotUpgrade, \; \mEventY \given \mUpgrade, \; \mEventY \given \mNotUpgrade},
  \tuple*{\mEventY}
}$
and
$\set*{
  \tuple*{\mNotUpgrade, \; \mEventX \given \mUpgrade, \; \mEventX \given \mNotUpgrade},
  \tuple*{\mEventX}
}$
gives player one $+2$ and $+1$, respectively.
A distribution that assigns $\nicefrac{1}{2}$ to both of these strategy profiles provides player one with an expected value of $+1.5$.
Since both players coordinate perfectly, player one can improve their value to $+2.5$ if they would just choose to upgrade their seating.
But there is no causal deviation that achieves more than $+1.5$ because causal deviations must play a fixed strategy after the deviation is triggered.
An external deviation must play a fixed strategy from the start of the game, which also precludes them from taking advantage of correlation in the recommendation distribution.
Figures A.2 and A.3 in the appendix enumerate all external and causal deviations in this game to prove this fully.

However, it is a valid action deviation to upgrade seating and preserve event coordination because an action deviation can re-correlate, \ie/, follow recommendations at information sets after an earlier deviation.
The resulting deviation strategy profile distribution is uniform over
\[\subblock*{
  \set*{
    \tuple*{\mUpgrade, \; \mEventY \given \mUpgrade, \; \mEventY \given \mNotUpgrade},
    \tuple*{\mEventY}
  },
  \set*{
    \tuple*{\mUpgrade, \; \mEventX \given \mUpgrade, \; \mEventX \given \mNotUpgrade},
    \tuple*{\mEventX}
  }
},\]
which achieves a value of $+2.5$.
See \cref{fig:bots-dev} for a visualization of this deviation.

\begin{figure}[t]
\centering
\begin{tikzpicture}[inner sep=0cm]
  \thinDeviationExampleMatrix{
    \node(rec1Label) {\textbf{rec 1:}};
    \node(rec1P2State) [state, right=0.5cm of rec1Label.east, anchor=west, font=\small] {2};
    \node(rec1P2H) [below left=0.4cm and 0.4cm of rec1P2State.south west] {};
    \node(rec1P2T) [below right=0.4cm and 0.4cm of rec1P2State.south east] {};
    \drawEdgeL{{\small X}}{rec1P2State}{rec1P2H}{alt};
    \drawEdgeR{{\small Y}}{rec1P2State}{rec1P2T}{follow}; \&

    \node(rec2Label) {\textbf{rec 2:}};
    \node(rec2P2State) [state, right=0.5cm of rec2Label.east, anchor=west, font=\small] {2};
    \node(rec2P2H) [below left=0.4cm and 0.4cm of rec2P2State.south west] {};
    \node(rec2P2T) [below right=0.4cm and 0.4cm of rec2P2State.south east] {};
    \drawEdgeL{{\small X}}{rec2P2State}{rec2P2H}{follow};
    \drawEdgeR{{\small Y}}{rec2P2State}{rec2P2T}{alt}; \\[-0.1cm]

    \node(cfDev1-1) [state, dev] {};
    \devRecTreeNodeCoords{cfDev1-1};
    \deviationExampleStates{cfDev1-1}{}{};
    \botsExampleTerminalsOne{cfDev1-1}{}{draw=black}{}{};
    \devRecRootEdges{cfDev1-1}{\Upgrade/}{dev}{\NotUpgrade/}{follow};
    \devRecLEdges{cfDev1-1}{\EventX/}{alt}{\EventY/}{follow};
    \devRecREdges{cfDev1-1}{\EventX/}{alt}{\EventY/}{follow}; \&

    \node(cfDev1-2) [state, dev] {};
    \devRecTreeNodeCoords{cfDev1-2};
    \deviationExampleStates{cfDev1-2}{}{};
    \botsExampleTerminalsTwo{cfDev1-2}{draw=black}{}{}{};
    \devRecRootEdges{cfDev1-2}{\Upgrade/}{dev}{\NotUpgrade/}{follow};
    \devRecLEdges{cfDev1-2}{\EventX/}{follow}{\EventY/}{alt};
    \devRecREdges{cfDev1-2}{\EventX/}{follow}{\EventY/}{alt}; \\
  };
\end{tikzpicture}
  \caption{A beneficial action and counterfactual deviation for player one in an EFCE in extended battle-of-the-sexes.
    Black lines denote recommendations, red lines denote deviations from unobserved recommendations, and grey lines mark actions that are not recommended or used by the deviation.}
  \label{fig:bots-dev}
\end{figure}
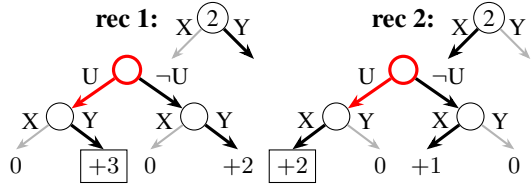

The game can be found in OpenSpiel~\cite{LanctotEtAl2019OpenSpiel} under the name {\tt extended\_bos.efg}.
See Appendix A.1 for an illustration of this game's extensive form, the EFCE recommendation distribution, and all of the possible deviations along with their values.

\section{CFR and Deviation Types}

The hindsight rationality framing encourages us to re-examine a prominent regret minimization algorithm for EFGs, \emph{counterfactual regret minimization} (\emph{CFR})~\parencite{cfr}.
CFR was designed to solve two-player, zero-sum games by minimizing external regret, but is it incidentally hindsight rational for causal or action deviations?
Or do we need a new type of deviation to fully understand CFR?

\subsection{Failure on Causal and Action Deviations}
We present an extension of Shapley's game~\parencite{Shapley} where CFR fails to behave according to an EFCCE or an AFCCE and therefore does not necessarily eliminate incentives for causal or action deviations.

In Shapley's game, both players simultaneously choose between Rock, Paper, and Scissors.
Rock beats Scissors, Scissors beats Paper, and Paper beats Rock, but both players lose if they choose the same item.
A winning player gets $+1$ while losing players get $-1$.
Our extension is that player one privately predicts whether or not player two will choose Rock after choosing their action.
If they accurately predict a Rock play, they receive a bonus, $b$, in addition to their usual reward, and if they accurately predict that they will not play Rock, they receive a smaller bonus, $b/3$.
There is no cost for inaccurate predictions, and the second player's decisions and payoffs are unchanged.
The game can be found in OpenSpiel~\cite{LanctotEtAl2019OpenSpiel} under the name {\tt extended\_shapleys.efg}.
The game's extensive-form is drawn in Figure A.5.

\cref{fig:shapley-small} shows the gap between the expected utility achieved by CFR's self-play empirical distribution (summed across players) and an optimal causal or action deviation across iterations.
In all of our experiments, optimal causal and action deviations achieved the same value.
The E/AFCE lines correspond to an optimal informed deviation while the E/AFCCE lines correspond to an optimal blind deviation.
Notice that in all figures, the gap does not continue to decrease over time as we would expect if CFR were to minimize causal or action regret.
See Figure A.6 for experiments with two other bonus values (0.3 and 30).

\begin{figure}[tb]
\centering
  \includegraphics[width=\linewidth]{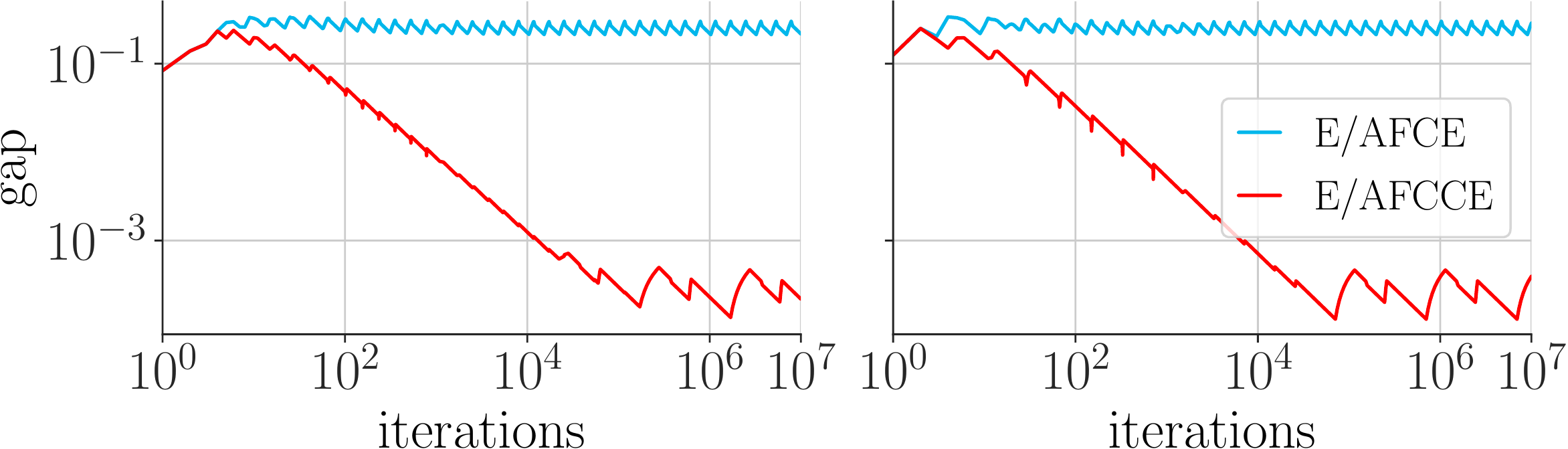}
  \caption{The gap between CFR's self-play empirical distribution and an extensive-form or agent-form (C)CE (E/AF(C)CE) in the extended Shapley's game with $b=0.003$. (Left) simultaneous-update CFR. (Right) alternating-update~\parencite{burch2019revisiting} CFR.}
  \label{fig:shapley-small}
\end{figure}

\subsection{Counterfactual Deviations}
While CFR can fail to minimize causal or action regret, it does more than minimize external regret.
We now define blind and informed \emph{counterfactual deviations} to describe the strategy deviations that arise from the definition of counterfactual value.
We show that conventional CFR with no-external-regret learners at each information set is hindsight rational for blind counterfactual deviations and that CFR with no-internal-regret learners is hindsight rational for informed counterfactual deviations.
Naturally, both types of counterfactual deviations give rise to classes of counterfactual equilibria and CFR's behavior in self-play therefore conforms to such equilibria.

Intuitively, a counterfactual deviation is one that plays from the beginning of the game to reach a particular information set and re-correlates thereafter.
Counterfactual deviations re-correlate like action deviations but they are still able to deviate across sequential information sets.
This ensures that the set of counterfactual deviations contains the set of external deviations.
Both counterfactual deviation variants are visualized in \cref{fig:dev-diagram}.

We briefly review CFR's definition and known theoretical properties before formally defining counterfactual deviations and equilibria and stating CFR's relationships with these concepts.

\subsubsection{Background:}
The \emph{counterfactual value function} at $\infoSet$, $\cfv_{\infoSet} : \Actions(\infoSet) \times \StrategySet \to \reals$, is the immediate expected value for taking action $a$ in $\infoSet$ assuming that everyone plays according to $\strat$ except that $i$ plays to reach $\infoSet$, \ie/,
\begin{align*}
  \cfv_{\infoSet}(a; \strat)
    = \sum_{\substack{
      h \in \infoSet,\\
      z \in \TerminalHistories}}
        \, \underbrace{\reachProb(h; \strat_{-i})}_{\text{Prob.~of being in $h$.}}
        \underbrace{\reachProb(ha, z; \strat) \utility_i(z)}_{\text{Future value given $ha$.}}.
\end{align*}
We overload
$\cfv_{\infoSet}(\immStrat; \strat) \as \E_{a \sim \immStrat}[\cfv_{\infoSet}(a; \strat)]$
for immediate strategies $\immStrat \in \simplex^{\abs{\Actions(\infoSet)}}$, and
$\cfv(\infoSet; \strat) \as \cfv_{\infoSet}(\strat_i(\infoSet); \strat)$
for the counterfactual value of $\infoSet$.

CFR is an online algorithm so it proposes a strategy on each round $t$, $\strat_i^t$, is evaluated based on strategies for the other players, $\strat_{-i}^t$, and uses this feedback to propose its next strategy, $\strat_i^{t + 1}$.
CFR's strategies are determined by no-external-regret learners at each $\infoSet$ that produce immediate strategies $\strat_i^t(\infoSet) \in \simplex^{\abs{\Actions}}$.
These learners are trained on
\emph{immediate counterfactual regret}~\parencite{cfr}, \ie/,
$\regret^t_{\infoSet}(\dev) = \cfv_{\infoSet}(\dev\strat^t_i(\infoSet); \strat^t) - \cfv(\infoSet; \strat^t)$.
The key result of \textcite{cfr} is that CFR's external regret decreases at the same rate as its immediate counterfactual regret, \ie/,
\begin{theorem}\cite[Theorem 3]{cfr}
  \label{thm:cfr}
  Let $\Regret^T_{\infoSet}\subex*{\DevSet^{\EXT}_{\Actions(\infoSet)}} = \max_{\dev \in \DevSet^{\EXT}_{\Actions(\infoSet)}} \regret^{1:T}_{\infoSet}(\dev)$ be the maximum immediate counterfactual regret at each information set $\infoSet$.
  Then the incentive to deviate from a sequence of $T$ strategies for player $i$, $\tuple{\strat^t_i}_{t = 1}^T$, to any other pure strategy is
  $\frac{1}{T} \Regret^T(\DevSet^{\EXT}_{\PureStratSet_i}) \le \frac{1}{T} \sum_{\infoSet \in \InfoSets_i}
    \subex*{\Regret^T_{\infoSet}\subex*{\DevSet^{\EXT}_{\Actions(\infoSet)}}}^+$.
\end{theorem}
\cref{thm:cfr} follows directly from an elementary decomposition relationship that we build on (see Appendix B.1 for a simple proof):
\begin{lemma}\cite[Equation 13, Lemma 5]{CFR-TR}
  \label{lem:cfr-decomp}
  The full regret,
$\Regret^T_{\infoSet}(\DevSet_{\PureStratSet_i}^{\EXT})$,
from information set $\infoSet \in \InfoSets_i$ that $\tuple{\strat^t_i}_{t = 1}^T$ suffers is upper bounded by
$\Regret^T_{\infoSet}(\DevSet_{\Actions(\infoSet)}^{\EXT})
+ \max_{a \in \Actions(\infoSet)}
    \sum_{\infoSet' \in \InfoSets_i(\infoSet, a)}
      \Regret^T_{\infoSet'}(\DevSet_{\PureStratSet_i}^{\EXT})$,
where $\InfoSets_i(\infoSet, a) \subset \InfoSets_i$ are the information sets that immediately follow after taking action $a$ in $\infoSet$.
 \end{lemma}

\subsubsection{Counterfactual Deviations and Intermediate Counterfactual Regret:}
\cref{lem:cfr-decomp} provides a one-step recursive connection between full counterfactual regret and immediate counterfactual regret.
When we unroll this recursion completely from the start of the game, we arrive at \cref{thm:cfr}.
But what if instead we unroll this recursion an arbitrary number of steps?

Perfect recall ensures that there is a unique sequence of player $i$'s information sets and actions leading to
any target information set, $\infoSet^{\TARGET}$.
Thus, if this target is $n$ steps away from an initial information set, $\infoSet_0$, then we can refer to the sequence of intermediate information sets as $\tuple{\infoSet_j}_{j = 0}^{n - 1}$ and the sequence of actions required to follow this sequence as $\tuple{a_j}_{j = 0}^{n - 1}$.
Once at the target, we could consider using any action transformation $\dev^{\TARGET}$.
What we have just described is a strategy deviation that plays to reach a target information set and then applies an action transformation, but otherwise leaves the strategy unmodified.
Denoting this strategy deviation as $\dev$, we can write the counterfactual value it achieves in terms of the usual one-step counterfactual values:
\begin{align*}
  \cfv(\infoSet_0; \dev(\strat_i^t), \strat_{-i}^t)
    &= \cfv_{\infoSet^{\TARGET}}(\dev^{\TARGET}(\strat_i^t(\infoSet^{\TARGET})); \strat^t)\\
    &\quad+ \sum_{j = 0}^{n - 1} \underbrace{\cfv_{\infoSet_j}(a_j; \strat^t) - \cfv(\infoSet_{j + 1}; \strat^t)}_{\text{telescoping difference}}.
\end{align*}
The counterfactual regret with respect to this value is intermediate between an immediate counterfactual regret where $n = 1$ and full counterfactual regret where all terminal information sets are treated as targets.
With regret decomposition, we can bound this \emph{intermediate counterfactual regret} in terms of immediate counterfactual regrets:
\begin{theorem}
  \label{thm:int-cfr-decomp}
  Let $\dev$ be the deviation that plays to reach $\infoSet^{\TARGET}$ from initial information set $\infoSet_0$ and uses action deviation $\dev^{\TARGET} \in \DevSet^{\TARGET} \subseteq \DevSet^{\SWAP}_{\Actions(\infoSet^{\TARGET})}$ once there, but otherwise leaves its given strategy unmodified.
The instantaneous intermediate counterfactual regret with respect to $\dev$ is
$\cfv(\infoSet_0; \dev(\strat_i^t), \strat_{-i}^t)
  - \cfv(\infoSet_0; \strat^t)$
and the cumulative intermediate counterfactual regret is bounded as
\begin{align*}
&\sum_{t = 1}^T
  \cfv(\infoSet_0; \dev(\strat_i^t), \strat_{-i}^t)
  - \cfv(\infoSet_0; \strat^t)\\
  &\le \Regret^T_{\infoSet^{\TARGET}}(\DevSet^{\TARGET})
  + \sum_{\infoSet^0 \preceq \infoSet \prec \infoSet^{\TARGET}}
    \Regret^T_{\infoSet}(\DevSet^{\EXT}_{\Actions(\infoSet)}).
\end{align*}
 \end{theorem}

If the initial information set is one at the start of the game, we call $\dev$ a \emph{counterfactual deviation}.
We distinguish between two major variants---\emph{blind counterfactual deviations} that use external deviations at the target information set and \emph{informed counterfactual deviations} that use internal deviations there instead.
We visualize both in \cref{fig:dev-diagram} and now provide formal definitions.
\begin{definition}
  \label{def:cf-dev}
  A blind counterfactual deviation is defined by a pair, $\tuple{\infoSet^{\TARGET}, a^{\TARGET}}$, where $\infoSet^{\TARGET} \in \InfoSets_i$ is a target information set that the deviation plays to reach deterministically from the start of the game and $a^{\TARGET}$ is the action taken at $\infoSet^{\TARGET}$.
  The deviation follows recommendations at every other information set.
  Formally,
  \begin{align*}
    [\dev\pureStrat_i](\infoSet(h)) = \begin{dcases}
      a &\mbox{if } \exists h' \in \infoSet^{\TARGET}, a, \, ha \sqsubseteq h'a^{\TARGET} \\
      \pureStrat_i(\infoSet(h)) &\mbox{o.w.}\\
    \end{dcases}
  \end{align*}
\end{definition}
\begin{definition}
  \label{def:in-cf-dev}
  An informed counterfactual deviation is defined by a triple, $\tuple{\infoSet^{\TARGET}, a^{\TRIGGER}, a^{\TARGET}}$, where $\infoSet^{\TARGET} \in \InfoSets_i$ is a target information set that the deviation plays to reach deterministically from the start of the game.
  If $a^{\TRIGGER}$ is the recommendation at $\infoSet^{\TARGET}$, then $a^{\TARGET}$ is played, otherwise the recommendation is followed.
  The deviation follows recommendations at every other information set.
  Formally,
  \begin{equation*}
    [\dev\pureStrat_i](\infoSet(h)) = \begin{dcases}
        a &\mbox{if } \exists h' \in \infoSet^{\TARGET}, a, \, ha \sqsubseteq h'\\
        a^{\TARGET} &\mbox{if } h \in \infoSet^{\TARGET}, \, \pureStrat_i(\infoSet^{\TARGET}) = a^{\TRIGGER}\\
        \pureStrat_i(\infoSet(h)) &\mbox{o.w.}
      \end{dcases}
  \end{equation*}
\end{definition}

The following conclusions immediately follow from \cref{thm:int-cfr-decomp}, \cref{def:cf-dev}, and the definition of CFR:
\begin{corollary}
  \label{cor:cfr-cf-devs}
  CFR is no-regret/hindsight rational with respect to blind counterfactual deviations.
Its blind counterfactual deviation regret bound is the same as that for external deviations (see, \eg/, \cref{thm:cfr}).
 \end{corollary}
\begin{corollary}
  \label{cor:cfr-in-in-cf-devs}
  CFR with no-internal-regret learners is no-regret/hindsight rational with respect to informed counterfactual deviations.
Its cumulative informed counterfactual deviation regret is upper bounded by
$\sum_{\infoSet \in \InfoSets_i}
    \subex*{\Regret^T_{\infoSet}\subex*{
      \DevSet^{\INT}_{\Actions(\infoSet)}
      \cup
      \DevSet^{\EXT}_{\Actions(\infoSet)}
      }
    }^+$.
 \end{corollary}

The counterfactual deviations we have just described correspond to two new equilibrium concepts:
\begin{definition}
A recommendation distribution is a counterfactual coarse-correlated equilibria (CFCCE) if there are no beneficial blind counterfactual deviations.
\end{definition}
\begin{definition}
A recommendation distribution is a counterfactual correlated equilibria (CFCE) if there are no beneficial informed counterfactual deviations.
\end{definition}
See the CF row and column in \cref{tab:eq-class-rel} for a summary of how counterfactual equilibria relate to other equilibrium concepts.
Since counterfactual deviations can re-correlate, they achieve the same improved value in the extended battle-of-the-sexes game as action deviations.
Counterfactual deviations can also transform multiple actions in a sequence, which is not true of action deviations, so the In or Out game does not pose a problem either.

\subsection{Counterfactual Deviation Weaknesses}
\label{sec:cf-dev-properties}

Counterfactual deviations are limited in that they cannot correlate before deviating and they cannot deviate to multiple information sets simultaneously.

\subsubsection{Beneficial Blind Causal Deviation in a CFCE:}

Consider an extended version of matching pennies where player one privately decides whether their goal is to Match (\Match/) with other player or to Not Match (\NotMatch/).
If player one achieves their goal, they get $+1$ while the other player gets $-1$ and \viceversa/ otherwise.

The joint recommendations
$\set*{
  \tuple*{\mNotMatch, \; \mHeads \given \mMatch, \; \mTails \given \mNotMatch},
  \tuple*{\mHeads}
}$
give player one $+1$ while
$\set*{
  \tuple*{\mMatch, \; \mHeads \given \mMatch, \; \mTails \given \mNotMatch},
  \tuple*{\mTails}
}$
yields $-1$, for an average of zero.
Since this time the decision holding player one back is the second action in the sequence at the Match information set where they should always play \Tails/, there is no better counterfactual deviation.
But deviating to \Tails/ in the Match information set, while following the Match or Not Match recommendation is a valid causal or action deviation.
The resulting deviation strategy profile distribution is uniform over
$\subblock*{
  \set*{
    \tuple*{\mNotMatch, \; \mTails \given \mMatch, \; \mTails \given \mNotMatch},
    \tuple*{\mHeads}
  },
  \set*{
    \tuple*{\mMatch, \; \mTails \given \mMatch, \; \mTails \given \mNotMatch},
    \tuple*{\mTails}
  }
}$,
which achieves a value of $+1$.
See \cref{fig:mp-dev} for a visualization of this deviation.

\begin{figure}[t]
\centering
\begin{tikzpicture}[inner sep=0cm]
  \thinDeviationExampleMatrix{
    \node(rec1Label) {\textbf{rec 1:}};
    \node(rec1P2State) [state, right=0.5cm of rec1Label.east, anchor=west, font=\small] {2};
    \node(rec1P2H) [below left=0.4cm and 0.4cm of rec1P2State.south west] {};
    \node(rec1P2T) [below right=0.4cm and 0.4cm of rec1P2State.south east] {};
    \drawEdgeL{{\small H}}{rec1P2State}{rec1P2H}{follow};
    \drawEdgeR{{\small T}}{rec1P2State}{rec1P2T}{alt}; \&

    \node(rec2Label) {\textbf{rec 2:}};
    \node(rec2P2State) [state, right=0.5cm of rec2Label.east, anchor=west, font=\small] {2};
    \node(rec2P2H) [below left=0.4cm and 0.4cm of rec2P2State.south west] {};
    \node(rec2P2T) [below right=0.4cm and 0.4cm of rec2P2State.south east] {};
    \drawEdgeL{{\small H}}{rec2P2State}{rec2P2H}{alt};
    \drawEdgeR{{\small T}}{rec2P2State}{rec2P2T}{follow}; \\[-0.1cm]

    \node(causDev2-1) [state] {};
    \devRecTreeNodeCoords{causDev2-1};
    \deviationExampleStates{causDev2-1}{dev}{};
    \mpExampleTerminalsOne{causDev2-1}{}{}{}{draw=black};
    \devRecRootEdges{causDev2-1}{M}{alt}{$\neg\text{M}$}{follow};
    \devRecLEdges{causDev2-1}{H}{follow}{T}{dev};
    \devRecREdges{causDev2-1}{H}{alt}{T}{follow}; \&

    \node(causDev2-2) [state] {};
    \devRecTreeNodeCoords{causDev2-2};
    \deviationExampleStates{causDev2-2}{dev}{};
    \mpExampleTerminalsTwo{causDev2-2}{}{draw=black}{}{};
    \devRecRootEdges{causDev2-2}{M}{follow}{$\neg\text{M}$}{alt};
    \devRecLEdges{causDev2-2}{H}{follow}{T}{dev};
    \devRecREdges{causDev2-2}{H}{alt}{T}{follow}; \\
  };
\end{tikzpicture}
  \caption{A beneficial causal and action deviation for player one in a CFCE in extended matching pennies.
    Black lines denote recommendations, red lines denote deviations from unobserved recommendations, and grey lines mark actions that are not recommended or utilized by the deviation.}
  \label{fig:mp-dev}
\end{figure}
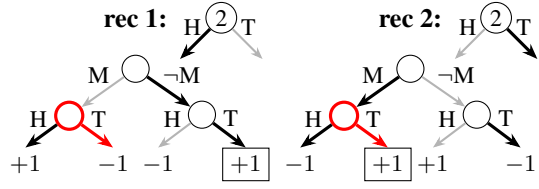
 
The game can be found in OpenSpiel~\cite{LanctotEtAl2019OpenSpiel} under the name {\tt extended\_mp.efg}.
See Appendix A.3 for an illustration of this game's extensive form, the CFCE recommendation distribution, and all of the possible deviations along with their values.

From this example, we can see why a causal or action deviation could outperform a counterfactual deviation in the extended Shapley's game.
Under the usual Shapley dynamics where players settle into a cycle that puts $1/6$ probability on Rock--Paper, Rock--Scissors, Paper--Rock, Paper--Scissors, Scissors--Rock, Scissors--Paper while avoiding Rock--Rock, Paper--Paper, and Scissors--Scissors, it is beneficial to predict Rock whenever you do not play Rock yourself.
However, counterfactual deviations cannot condition behavior at one information set according to play at a previous information set by definition, thereby preventing any counterfactual deviation from taking advantage of this fact.
This also explains why the performance of optimal causal and action deviations is the same in this game: the only causal deviations that are not counterfactual deviations in a two action game are in fact action deviations.

To be clear, it is not that a player considering a counterfactual deviation forgets that they are in an information set where they have chosen to play Rock.
It is instead that their appraisal of whether to predict Rock or not Rock is based on the average utility they receive for that prediction given that the other player chooses Rock, Paper, and Scissors equally often, which is what the player would observe if they were to always choose Rock.
Counterfactual deviations do not take into account the correlation information that may be present in previous decisions.
CFR fails to minimize regret for action or causal deviations in this case because it evaluates its actions in exactly the same way through its use of counterfactual values and regrets.

\correctionStart{Present MacQueen's counterexample.}
\subsubsection{Beneficial External Deviation in a CFCE:}

\def\leftAction/{L}
\def\rightAction/{R}

This example comes from \textcite{macQueen2022counterexampleSlackMessages}.
Player one first chooses between actions Left (\leftAction/) and Right (\rightAction/).
If they choose \leftAction/, the game ends and they receive $+1$.
If they choose \rightAction/, player two publicly chooses between \leftAction/ and \rightAction/, and player one ends up in one of two information sets with identical payoffs.
If player one's second action is \leftAction/, they receive $-2$, if they choose \rightAction/, they receive $+2$.

\begin{figure}[t]
  \centering
  \begin{tikzpicture}[inner sep=0cm]
    \tikzstyle{obsState} = [circle, draw=black, minimum size=0.4*2*\stateMinimumRadius, inner sep=0.2mm, fill=black]

    \tikzstyle{column 1} = [anchor=north]
    \tikzstyle{column 2} = [anchor=north]
    \matrix[column sep=0.05cm, ampersand replacement=\&]{
      \node(rec1Label) {\textbf{rec 1:}}; \&

      \node(rec2Label) {\textbf{rec 2:}}; \\

      \node(exDev-1) [state, dev] {};
      \devRecTreeInternalCoords{exDev-1}
      \node(exDev-1L) [util] at (exDev-1LCoord) {$+1$};
      \node(exDev-1R) [obsState] at (exDev-1RCoord) {};

      \coordinate(exDev-1RLCoord) at ($(exDev-1R.south west)+(-\cmToEmFactor*1.1em,-\cmToEmFactor*0.6em)$);
      \coordinate(exDev-1RRCoord) at ($(exDev-1R.south east)+(\cmToEmFactor*1.1em,-\cmToEmFactor*0.6em)$);
      \node(exDev-1RL) [state, dev] at (exDev-1RLCoord) {};
      \devRecTreeTerminalCoords{exDev-1RL}
      \node(exDev-1RLL) [util] at (exDev-1RLLCoord) {$-2$};
      \node(exDev-1RLR) [util, draw=black] at (exDev-1RLRCoord) {$+2$};

      \node(exDev-1RR) [state, dev] at (exDev-1RRCoord) {};
      \devRecTreeTerminalCoords{exDev-1RR}
      \node(exDev-1RRL) [util] at (exDev-1RRLCoord) {$-2$};
      \node(exDev-1RRR) [util] at (exDev-1RRRCoord) {$+2$};

      \drawEdgeL{\leftAction/}{exDev-1}{exDev-1L}{follow}
      \drawEdgeL{\rightAction/}{exDev-1}{exDev-1R}{dev}

      \devRecRootEdges{exDev-1R}{\leftAction/}{follow}{\rightAction/}{zeroProb}
      \devRecLEdges{exDev-1R}{\leftAction/}{follow}{\rightAction/}{dev}
      \devRecREdges{exDev-1R}{\leftAction/}{follow}{\rightAction/}{dev} \&

      \node(exDev-2) [state, dev] {};
      \devRecTreeInternalCoords{exDev-2}
      \node(exDev-2L) [util] at (exDev-2LCoord) {$+1$};
      \node(exDev-2R) [obsState] at (exDev-2RCoord) {};

      \coordinate(exDev-2RLCoord) at ($(exDev-2R.south west)+(-\cmToEmFactor*1.1em,-\cmToEmFactor*0.6em)$);
      \coordinate(exDev-2RRCoord) at ($(exDev-2R.south east)+(\cmToEmFactor*1.1em,-\cmToEmFactor*0.6em)$);
      \node(exDev-2RL) [state, dev] at (exDev-2RLCoord) {};
      \devRecTreeTerminalCoords{exDev-2RL}
      \node(exDev-2RLL) [util] at (exDev-2RLLCoord) {$-2$};
      \node(exDev-2RLR) [util] at (exDev-2RLRCoord) {$+2$};

      \node(exDev-2RR) [state, dev] at (exDev-2RRCoord) {};
      \devRecTreeTerminalCoords{exDev-2RR}
      \node(exDev-2RRL) [util] at (exDev-2RRLCoord) {$-2$};
      \node(exDev-2RRR) [util, draw=black] at (exDev-2RRRCoord) {$+2$};

      \drawEdgeL{\leftAction/}{exDev-2}{exDev-2L}{follow}
      \drawEdgeL{\rightAction/}{exDev-2}{exDev-2R}{dev}

      \devRecRootEdges{exDev-2R}{\leftAction/}{zeroProb}{\rightAction/}{follow}
      \devRecLEdges{exDev-2R}{\leftAction/}{follow}{\rightAction/}{dev}
      \devRecREdges{exDev-2R}{\leftAction/}{follow}{\rightAction/}{dev} \\
    };
  \end{tikzpicture}
    \caption{A beneficial external deviation for player one in a CFCE in MacQueen's counterexample.
      Black lines denote recommendations, red lines denote deviations from unobserved recommendations, and grey lines mark actions that are not recommended or used by the deviation.}
    \label{fig:macqueenCounterexample}
\end{figure}
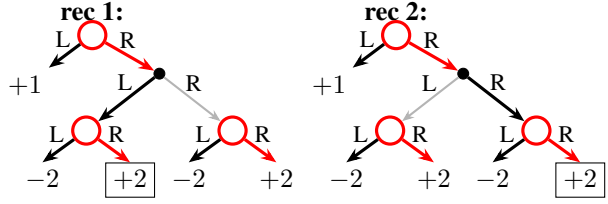

The CFCE recommendations are for player one to always choose \leftAction/ and for player two to switch between \leftAction/ and \rightAction/:
$
  \subblock*{
\set*{
        \tuple{
          \text{\leftAction/}, \; \text{\leftAction/} \given \text{\rightAction/ \leftAction/}, \; \text{\leftAction/} \given \text{\rightAction/ \rightAction/}
        },
        \tuple{\text{\leftAction/}}
      },
      \set*{
        \tuple{
          \text{\leftAction/}, \; \text{\leftAction/} \given \text{\rightAction/ \leftAction/}, \; \text{\leftAction/} \given \text{\rightAction/ \rightAction/}
        },
        \tuple{\text{\rightAction/}}
      }
}
$.
Player one achieves a value of $+1$ averaged across these two strategy profiles since they always end the game before the other player can act.
We know that these recommendations are a CFCE because the counterfactual deviations either leave the recommendations unmodified or they play the \rightAction/ action in the root information set and then play \rightAction/ in exactly \emph{one} of the two successor information sets.
Making any such modification ensures that player one achieves $+2$ under one recommendation and $-2$ under the other, resulting in an average value of $0$.
According to \cref{def:cf-dev,def:in-cf-dev}, counterfactual deviations can only modify strategies along a single path, so there is no counterfactual deviation that plays \rightAction/ in \emph{both} of the successor information sets.

However, the external deviation that always plays \rightAction/ achieves a value of $+2$ averaged across these two strategy profiles.
Thus, we have a beneficial external deviation in a CFCE.
The recommendations and the beneficial external deviation are visualized in \cref{fig:macqueenCounterexample}.
\correctionEnd

\section{Observable Sequential Rationality}

\correctionStart{Wording change.}
Since we already know from \textcite{cfr} that CFR's self-play equilibrium class is in the set of CCEs, CFR must be more than hindsight rational for the blind counterfactual deviations.
So far, we have examined properties of CFR that emerge from regret decomposition, but we have only used bounds on regrets that begin at the start of the game.
Regret decomposition, however, applies at every subtree within an information set forest in a way that resembles sequential rationality~\parencite{KrepsWilson82}.
We introduce \emph{observable sequential rationality} and \emph{observable sequential equilibrium} to define CFR's refinement of CCE and CFCCE.
\correctionEnd

Sequential rationality is based on examining the reward that follows from a particular information set $\infoSet$,
\begin{align*}
  \utility_i(\strat \given \xi_{\infoSet}) \as \E_{h \sim \xi_{\infoSet}}\subblock*{
      \sum_{\substack{
        z \sqsupseteq h\\
        z \in \TerminalHistories
      }}
        \reachProb(h, z; \strat) \utility_i(z)},
\end{align*}
where histories are distributed according to a belief distribution,
$\xi_{\infoSet} : \infoSet \to [0, 1]$.
An assessment, that is, a system of beliefs and strategy pair,
$(\xi, \strat_i)$,
is sequentially rational if
$\utility_i(\dev(\strat_i), \strat_{-i}) \given \xi_{\infoSet}) - \utility_i(\strat \given \xi_{\infoSet}) \le 0$
for each information set $\infoSet \in \InfoSets_i$
and external deviation $\dev$.

Different types of equilibria have been proposed that reuse the optimality constraint described by sequential rationality but construct systems of beliefs in different ways.
Typically, beliefs must be consistent, which means that each belief distribution is defined as
\begin{align*}
  \xi_{\infoSet} : h \to \dfrac{\reachProb(h; \strat_{-i})}{\sum_{h' \in \infoSet(h)} \reachProb(h'; \strat_{-i})},
\end{align*}
as derived by applying Bayes' rule according to $\strat_{-i}$.\footnote{Note that $\strat_i$ cannot skew the history distribution at player $i$'s information set because of perfect recall.}
However, if $\infoSet$ is unreachable, \ie/, the probability of reaching each history in $\infoSet$ is zero, then $\xi_{\infoSet}$ is undefined and it is unclear what beliefs should be held at $\infoSet$.
The method of constructing $\xi_{\infoSet}$ for unreachable positions differentiates various types of sequentially-refined equilibria.
Instead, we reuse the notion of sequential rationality under consistent beliefs without constructing beliefs at unreachable information sets at all.

A natural extension to the mediated execution setting is to define a sequentially rational recommendation distribution,
$\recDist \in \simplex^{\abs{\PureStratSet}}$,
as one where the benefit of a deviation according to a system of beliefs $\xi^{\pureStrat_{-i}}$ consistent with mediator-sampled strategies
$\pureStrat_{-i} \in \PureStratSet_{-i}$
is non-positive,
\begin{align}
  \E_{\pureProfile \sim \recDist}\subblock*{
    \utility_i(\dev(\pureStrat_i), \pureStrat_{-i}) \given \xi_{\infoSet}) - \utility_i(\pureStrat \given \xi_{\infoSet})
  } \le 0,
  \label{eq:seq-rational}
\end{align}
at every information set.
Notice that at every reachable information set, \cref{eq:seq-rational} differs from full counterfactual regret only by the Bayes normalization constant~\parencite{srinivasan2018actor}, since counterfactual values are weighted by the probability that the other players play to each history.
At unreachable information sets, the counterfactual value function is always zero so a sequential rationality-like constraint based on counterfactual values is weaker only in that all potential behavior is acceptable at unreachable positions.
However, since an unreachable position corresponds to a situation that was never experienced, the difference to a hindsight analysis is entirely inconsequential.
We call this condition \emph{observable sequential rationality} because it is indistinguishable from sequential rationality under Bayesian beliefs according to all outcomes that could be observed while playing with players who always follow their recommendations.
\begin{definition}
  \label{def:sub-seq-rationality}
  A recommendation distribution,
  $\recDist \in \simplex^{\abs{\PureStratSet}}$,
  is observably sequentially rational with respect to a set of deviations,
  $\DevSet \subseteq \DevSet^{\SWAP}_{\PureStratSet_i}$,
  on player $i$'s strategies, if and only if the maximum benefit for every deviation,
  $\dev \in \DevSet$,
  \correctionStart{Generalize original definition to account for reach-probability weights so that the definition is clear for internal and external deviations.}
  according to the deviation reach-probability-weighted counterfactual value function at every information set,
  $\infoSet \in \InfoSets_i$,
  is non-positive, \ie/,
  \begin{align*}
    \max_{\dev \in \DevSet}
      \E_{\pureProfile \sim \recDist}\subblock*{
        \reachProb(h; \dev(\pureStrat_i))
          \subex*{
            \cfv(\infoSet; \dev(\pureStrat_i), \pureStrat_{-i})
            - \cfv(\infoSet; \pureProfile)
          }
      }
      \le 0,
  \end{align*}
  where $h$ is an arbitrary history in $\infoSet$.
  \correctionEnd
\end{definition}

\correctionStart{Add theorem of how observable sequential rationality elevates the CF deviations.}
If the recommendation distribution $\recDist$ is constructed from $T$ mixed strategy profiles and each deviation is external or counterfactual, then the observable sequential rationality gap for player $i$ at information set $\infoSet$ is the positive part of their maximum average full regret from $\infoSet$, \ie/,
\begin{align}
  \max_{\dev \in \DevSet}
    \E_{\pureProfile \sim \recDist}\subblock*{
      \reachProb(h; \dev(\pureStrat_i))
        \subex*{
          \cfv(\infoSet; \dev(\pureStrat_i), \pureStrat_{-i})
          - \cfv(\infoSet; \pureProfile)
        }
    }
      = \max_{\dev \in \DevSet}
        \dfrac{1}{T} \subex*{\regret_{\infoSet}^{1:T}(\dev)}^+,
\end{align}
where $\subex*{\cdot}^+ = \max \set{\cdot, 0}$.
This equality holds because, since $\dev$ is external or counterfactual, either $\reachProb(h; \dev(\pureStrat_i)) = 1$ ($\dev$ plays to $\infoSet$) or the difference
$\subex*{
  \cfv(\infoSet; \dev(\pureStrat_i), \pureStrat_{-i})
  - \cfv(\infoSet; \pureProfile)
} = 0$
due to $\dev$'s re-correlation at or before $\infoSet$.

Using \textcite{cfr}'s regret decomposition, we can show that the strength of the blind counterfactual deviations is elevated to subsume the external deviations when combined with observable sequential rationality.
\begin{theorem}
  \label{thm:cfDevElevationWithOsr}
  If the average full regret for player $i$ at each information set $\infoSet$ for the sequence of $T$ mixed strategy profiles,
$\tuple{ \strat^t }_{t = 1}^T$,
with respect to each blind counterfactual deviation
$\dev$
is
$d_{\infoSet}(\dev)f(T) \ge 0$,
where
$d_{\infoSet}(\dev)$
is the number of external action transformations
$\dev$
applies from $\infoSet$ to the end of the game, then the observable sequential rationality gap for $i$ with respect to the blind counterfactual and external deviations is no more than
$\abs{\InfoSets_i} f(T)$. \end{theorem}
\correctionEnd

\correctionStart{Wording change.}
Two corollaries of \cref{thm:int-cfr-decomp,def:sub-seq-rationality,thm:cfDevElevationWithOsr} are:
\begin{corollary}
  \label{cor:os-cf-dev}
  CFR is observably sequentially hindsight rational with respect to blind counterfactual and external deviations.
  Its empirical play approaches exact rationality at the same rate as its average external regret vanishes (see, \eg/, \cref{thm:cfr}).
\end{corollary}
\begin{corollary}
  \label{cor:os-cfcce}
  CFR in self-play converges to an observable sequential CFCCE at the same rate as its average external regret vanishes (see, \eg/, \cref{thm:cfr}).
\end{corollary}
\correctionEnd
An analogous result applies to CFR with no-internal-regret learners, informed counterfactual deviations, and observable sequential CFCE.
\cref{cor:os-cf-dev,cor:os-cfcce} provide the most thorough characterization of CFR's behavior (both for a single player and in self-play) to date and both apply to multi-player, general sum games.

To show how observable sequential rationality differs from non-sequential rationality in practice, we provide an example in Appendix A.4 where a CE is not an observable sequential CFCCE.
See the OS-CF row and column in \cref{tab:eq-class-rel} for a summary of how observable sequential counterfactual equilibria relate to the other equilibrium concepts.

\section{CFR for Action Deviations}

We showed that CFR does not behave according to an EFCCE or AFCCE, but could we modify CFR to do so?
Modifying CFR to address causal deviations requires recently developed non-trivial modifications~\parencite{celli2020noregret}, but we can easily modify CFR for action deviations by weighting counterfactual regrets by reach probabilities.

Since by definition an action deviation applies an action transformation at a single trigger information set, the benefit of such a deviation is simply the average reach-probability-weighted immediate regret, no decomposition required.
Therefore, employing no-regret learners on the reach-probability-weighted counterfactual value functions in each information set ensures that the average reach-probability-weighted immediate regret becomes non-positive.
\begin{theorem}
  \label{thm:cfr-a-devs}
  The CFR-like algorithm that trains a no-regret learner at each information set on the reach-probability-weighted counterfactual value functions,
  $\hat{\cfv}_{\infoSet}(\cdot; \strat^t) : a \to \reachProb(\infoSet; \strat^t_i) \cfv_{\infoSet}(a; \strat^t)$,
  is no-regret/hindsight rational with respect to action deviations.
  If the learners minimize external regret, then the algorithm minimizes blind action deviation regret, and if the learners minimize internal regret, then the algorithm minimizes informed action deviation regret.
  At all times, the algorithm's action deviation regret cannot be more than the maximum regret suffered by any single learner.
\end{theorem}

\textcite{celli2020noregret}'s algorithm is also no-regret with respect to informed action deviations because it uses no-internal-regret learners trained on reach-probability-weighted immediate counterfactual values.
\cref{thm:cfr-a-devs} gives an algorithm that is weaker in that it does not minimize causal, counterfactual, or external regret, but it does not require action sampling, nor does it require multiple learners at every information set.
It remains to investigate this trade-off experimentally.

PGPI and \emph{actor-critic policy iteration} (\emph{ACPI}) describe the optimization problems underlying standard policy gradient and actor-critic methods~\parencite{srinivasan2018actor}.
(Euclidean) projected PGPI and ACPI using tabular representations are equivalent to \emph{state-local generalized infinitesimal gradient ascent} (\emph{GIGA(s)})~\parencite[Definition 2]{rpg_arxiv} and are no-external-regret with respect to reach-probability-weighted counterfactual value functions at each information set~\parencite[Lemma 4]{rpg_arxiv}.
These algorithms are also therefore hindsight rational with respect to blind action deviations.
Thus, \cref{thm:cfr-a-devs} establishes a new connection between CFR and reinforcement learning algorithms via action deviations, PGPI, and ACPI.

\section{Conclusions}

Complicated general-sum, multi-agent games, perhaps unsurprisingly, resist static solutions.
Learning online and analyzing agent behavior in hindsight is a dynamic alternative.
This hindsight rationality view returns equilibria to a descriptive role, as originally introduced within the field of game theory, rather than the prescriptive role it often holds in the field of artificial intelligence.
Agents prescriptively aim to reduce regret for deviations in hindsight and thereby learn to behave rationally over time, while equilibria describe the consequences of interactions between such agents.

The mediated equilibrium and deviation landscapes of EFGs are particularly rich because there is substantial space for deviation types that have intermediate power and computational requirements between external and internal deviations.
To develop these EFG deviation and equilibrium types, we re-examined causal and action deviations along with their corresponding equilibrium concepts, extensive-form and agent-form equilibria.
We showed how action deviations can outperform causal deviations, which dispels a common misunderstanding that an EFCE is always an AFCE.

We showed that CFR's empirical play does not converge toward an EFCCE or AFCCE in self-play and thus CFR is not hindsight rational for causal or action deviations.
Instead, we defined blind and informed counterfactual deviations, and observable sequential rationality to more precisely characterize CFR's behavior.
CFR is observably sequentially hindsight rational with respect to blind counterfactual deviations
and CFR with internal learners has the same property with respect to informed counterfactual deviations.
In self-play, these algorithms converge toward observable sequential CFCCEs or CFCEs, respectively.

\cref{tab:eq-class-rel} summarizes all of the equilibrium relationships investigated in this work.

One avenue for future work is to modify CFR to simultaneously handle each type of deviation and show the practical benefits of such an algorithm.
Additionally, investigating observable sequential rationality with non-counterfactual deviations is another avenue for interesting future study opened up by our work.

\section{Acknowledgments}
Dustin Morrill, Michael Bowling, and James Wright are supported by the Alberta Machine Intelligence Institute (Amii) and NSERC.
Michael Bowling and James Wright hold Canada CIFAR AI Chairs at Amii.
Amy Greenwald is supported in part by NSF Award CMMI-1761546.
Thanks to Bernhard von Stengel for answering our questions about the EFCE concept.
Thanks to anonymous reviewers and Julien Perolat for constructive comments and suggestions.
\correctionStart{Add acknowledgement.}
Thanks to Revan MacQueen for an example of a CFCE with a beneficial external deviation.
\correctionEnd
 
\bibliography{references}

\appendix
\renewcommand\thefigure{\thesection.\arabic{figure}}
\setcounter{figure}{0}
\renewcommand\thetable{\thesection.\arabic{table}}
\setcounter{table}{0}
\section{Game Examples}

\subsection{Extended Battle-of-the-Sexes}
\cref{fig:ext-bots} is an illustration of the extensive-form tree of the extended battle-of-the-sexes game introduced in \cref{sec:ext-bots}.
The recommendations
$\set*{
  \tuple*{\mNotUpgrade, \; \mEventY \given \mUpgrade, \; \mEventY \given \mNotUpgrade},
  \tuple*{\mEventY}
}$
and
$\set*{
  \tuple*{\mNotUpgrade, \; \mEventX \given \mUpgrade, \; \mEventX \given \mNotUpgrade},
  \tuple*{\mEventX}
}$
are visualized from player one's perspective at the top of \cref{fig:bots-strat-devs,fig:bots-caus-devs,fig:bots-cf-devs}.
\cref{fig:bots-strat-devs} shows all of the external deviations.
\cref{fig:bots-caus-devs} shows all of the blind causal deviations that are not also external deviations.
\cref{fig:bots-cf-devs} shows all of the blind counterfactual deviations that are not also external deviations.
The complete set of blind action deviations is composed of the deviations in \cref{fig:bots-caus-devs} plus counterfactual deviations 1 and 4 in \cref{fig:bots-cf-devs}.

To see that these recommendations are an EFCE, notice that conditioning a deviation on the recommendation at the root cannot improve the value of any causal deviation since the first recommended action is the same under both recommendations.
And since the recommendations after \NotUpgrade/ cannot be improved, informed causal deviations cannot improve the value by triggering on a recommendation in the \NotUpgrade/ information set.

Note that blind counterfactual deviation 1, which is also a blind action deviation, achieves a value of $+2.5 > +1.5$, and this is the only beneficial deviation.
Therefore, this recommendation distribution is an EFCE with a beneficial blind counterfactual and action deviation, proving it is not a CFCCE or an AFCCE.

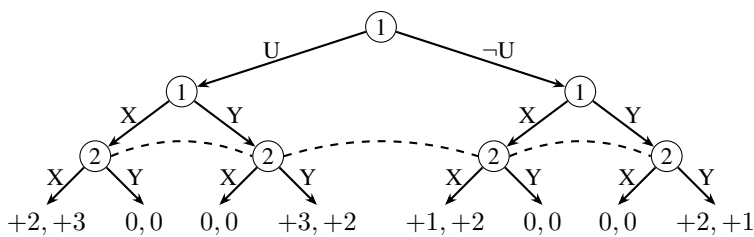
\begin{figure}[htbp]
\centering
\begin{tikzpicture}[inner sep=0cm]
  \tikzstyle{state} = [circle, draw=black, minimum size=0.3cm, inner sep=0.05cm]
  \tikzstyle{arrow} = [thick, >=Stealth, -{>[scale=0.7]}]

  \node(root) [state] {\small 1};
  \node(u) [state, below left=\shortAYLength and 2.5cm of root.south west, anchor=north] {\small 1};
  \draw[arrow] (root) -- node[anchor=south east] {\small \Upgrade/} (u);

  \node(nu) [state, below right=\shortAYLength and 2.5cm of root.south east, anchor=north] {\small 1};
  \draw[arrow] (root) -- node[anchor=south west] {\small \NotUpgrade/} (nu);

  \node(uX) [state, below left=\shortAYLength and 1cm of u.south west, anchor=north] {\small 2};
  \draw[arrow] (u) -- node[anchor=south east] {\small \EventX/} (uX);

  \node(uY) [state, below right=\shortAYLength and 1cm of u.south east, anchor=north] {\small 2};
  \draw[arrow] (u) -- node[anchor=south west] {\small \EventY/} (uY);

  \node(nuX) [state, below left=\shortAYLength and 1cm of nu.south west, anchor=north] {\small 2};
  \draw[arrow] (nu) -- node[anchor=south east] {\small \EventX/} (nuX);

  \node(nuY) [state, below right=\shortAYLength and 1cm of nu.south east, anchor=north] {\small 2};
  \draw[arrow] (nu) -- node[anchor=south west] {\small \EventY/} (nuY);

  \draw[thick, dashed] (uX.east) parabola[bend pos=0.5] bend +(0, 0.2cm) (uY.west);
  \draw[thick, dashed] (uY.east) parabola[bend pos=0.5] bend +(0, 0.2cm) (nuX.west);
  \draw[thick, dashed] (nuX.east) parabola[bend pos=0.5] bend +(0, 0.2cm) (nuY.west);

  \node(uXX) [util, below left=\shortAYLength and 0.5cm of uX.south west, anchor=north] {$+2,+3$};
  \draw[arrow] (uX) -- node[anchor=south east] {\small \EventX/} (uXX.north);

  \node(uYX) [util, below left=\shortAYLength and 0.5cm of uY.south west, anchor=north] {$0,0$};
  \draw[arrow] (uY) -- node[anchor=south east] {\small \EventX/} (uYX.north);

  \node(uXY) [util, below right=\shortAYLength and 0.5cm of uX.south east, anchor=north] {$0,0$};
  \draw[arrow] (uX) -- node[anchor=south west] {\small \EventY/} (uXY.north);

  \node(uYY) [util, below right=\shortAYLength and 0.5cm of uY.south east, anchor=north] {$+3,+2$};
  \draw[arrow] (uY) -- node[anchor=south west] {\small \EventY/} (uYY.north);

  \node(nuXX) [util, below left=\shortAYLength and 0.5cm of nuX.south west, anchor=north] {$+1,+2$};
  \draw[arrow] (nuX) -- node[anchor=south east] {\small \EventX/} (nuXX.north);

  \node(nuYX) [util, below left=\shortAYLength and 0.5cm of nuY.south west, anchor=north] {$0,0$};
  \draw[arrow] (nuY) -- node[anchor=south east] {\small \EventX/} (nuYX.north);

  \node(nuXY) [util, below right=\shortAYLength and 0.5cm of nuX.south east, anchor=north] {$0,0$};
  \draw[arrow] (nuX) -- node[anchor=south west] {\small \EventY/} (nuXY.north);

  \node(nuYY) [util, below right=\shortAYLength and 0.5cm of nuY.south east, anchor=north] {$+2,+1$};
  \draw[arrow] (nuY) -- node[anchor=south west] {\small \EventY/} (nuYY.north);
\end{tikzpicture}
\caption{The extensive-form of an extended battle-of-the-sexes game where the first player privately chooses whether or not to upgrade seating for no extra charge before selecting an event simultaneously with their partner.
\InfoSetDesc/}
\label{fig:ext-bots}
\end{figure}

\begin{figure}[htbp]
\centering
\begin{tikzpicture}[inner sep=0cm]
  \deviationExampleMatrix{
    \botsExampleHeader; \\[-0.3cm]

    \botsExampleFollowTrees; \\

\node(stratDev1-Label) [behaveLabel] {external\\deviation 1}; \&
    \node(stratDev1-1) [state, dev] {};
    \devRecTreeNodeCoords{stratDev1-1};
    \deviationExampleStates{stratDev1-1}{dev}{dev};
    \botsExampleTerminalsOne{stratDev1-1}{draw=black}{}{}{};
    \devRecRootEdges{stratDev1-1}{\Upgrade/}{dev}{\NotUpgrade/}{alt};
    \devRecLEdges{stratDev1-1}{\EventX/}{dev}{\EventY/}{alt};
    \devRecREdges{stratDev1-1}{\EventX/}{dev}{\EventY/}{alt}; \&

    \node(stratDev1-2) [state, dev] {};
    \devRecTreeNodeCoords{stratDev1-2};
    \deviationExampleStates{stratDev1-2}{dev}{dev};
    \botsExampleTerminalsTwo{stratDev1-2}{draw=black}{}{}{};
    \devRecRootEdges{stratDev1-2}{\Upgrade/}{dev}{\NotUpgrade/}{alt};
    \devRecLEdges{stratDev1-2}{\EventX/}{dev}{\EventY/}{alt};
    \devRecREdges{stratDev1-2}{\EventX/}{dev}{\EventY/}{alt}; \&
    \node(stratDev1-Value) {$+1$}; \\

\node(stratDev2-Label) [behaveLabel] {external\\deviation 2}; \&
    \node(stratDev2-1) [state, dev] {};
    \devRecTreeNodeCoords{stratDev2-1};
    \deviationExampleStates{stratDev2-1}{dev}{dev};
    \botsExampleTerminalsOne{stratDev2-1}{draw=black}{}{}{};
    \devRecRootEdges{stratDev2-1}{\Upgrade/}{dev}{\NotUpgrade/}{alt};
    \devRecLEdges{stratDev2-1}{\EventX/}{dev}{\EventY/}{alt};
    \devRecREdges{stratDev2-1}{\EventX/}{alt}{\EventY/}{dev}; \&

    \node(stratDev2-2) [state] {};
    \devRecTreeNodeCoords{stratDev2-2};
    \deviationExampleStates{stratDev2-2}{dev}{dev};
    \botsExampleTerminalsTwo{stratDev2-2}{draw=black}{}{}{};
    \devRecRootEdges{stratDev2-2}{\Upgrade/}{dev}{\NotUpgrade/}{alt};
    \devRecLEdges{stratDev2-2}{\EventX/}{dev}{\EventY/}{alt};
    \devRecREdges{stratDev2-2}{\EventX/}{alt}{\EventY/}{dev}; \&
    \node(stratDev2-Value) {$+1$}; \\

\node(stratDev3-Label) [behaveLabel] {external\\deviation 3}; \&
    \node(stratDev3-1) [state, dev] {};
    \devRecTreeNodeCoords{stratDev3-1};
    \deviationExampleStates{stratDev3-1}{dev}{dev};
    \botsExampleTerminalsOne{stratDev3-1}{}{draw=black}{}{};
    \devRecRootEdges{stratDev3-1}{\Upgrade/}{dev}{\NotUpgrade/}{alt};
    \devRecLEdges{stratDev3-1}{\EventX/}{alt}{\EventY/}{dev};
    \devRecREdges{stratDev3-1}{\EventX/}{dev}{\EventY/}{alt}; \&

    \node(stratDev3-2) [state, dev] {};
    \devRecTreeNodeCoords{stratDev3-2};
    \deviationExampleStates{stratDev3-2}{dev}{dev};
    \botsExampleTerminalsTwo{stratDev3-2}{}{draw=black}{}{};
    \devRecRootEdges{stratDev3-2}{\Upgrade/}{dev}{\NotUpgrade/}{alt};
    \devRecLEdges{stratDev3-2}{\EventX/}{alt}{\EventY/}{dev};
    \devRecREdges{stratDev3-2}{\EventX/}{dev}{\EventY/}{alt}; \&
    \node(stratDev3-Value) {$+1.5$}; \\

\node(stratDev4-Label) [behaveLabel] {external\\deviation 4}; \&
    \node(stratDev4-1) [state, dev] {};
    \devRecTreeNodeCoords{stratDev4-1};
    \deviationExampleStates{stratDev4-1}{dev}{dev};
    \botsExampleTerminalsOne{stratDev4-1}{}{draw=black}{}{};
    \devRecRootEdges{stratDev4-1}{\Upgrade/}{dev}{\NotUpgrade/}{alt};
    \devRecLEdges{stratDev4-1}{\EventX/}{alt}{\EventY/}{dev};
    \devRecREdges{stratDev4-1}{\EventX/}{alt}{\EventY/}{dev}; \&

    \node(stratDev4-2) [state, dev] {};
    \devRecTreeNodeCoords{stratDev4-2};
    \deviationExampleStates{stratDev4-2}{dev}{dev};
    \botsExampleTerminalsTwo{stratDev4-2}{}{draw=black}{}{};
    \devRecRootEdges{stratDev4-2}{\Upgrade/}{dev}{\NotUpgrade/}{alt};
    \devRecLEdges{stratDev4-2}{\EventX/}{alt}{\EventY/}{dev};
    \devRecREdges{stratDev4-2}{\EventX/}{alt}{\EventY/}{dev}; \&
    \node(stratDev4-Value) {$+1.5$}; \\

\node(stratDev5-Label) [behaveLabel] {external\\deviation 5}; \&
    \node(stratDev5-1) [state, dev] {};
    \devRecTreeNodeCoords{stratDev5-1};
    \deviationExampleStates{stratDev5-1}{dev}{dev};
    \botsExampleTerminalsOne{stratDev5-1}{}{}{draw=black}{};
    \devRecRootEdges{stratDev5-1}{\Upgrade/}{alt}{\NotUpgrade/}{dev};
    \devRecLEdges{stratDev5-1}{\EventX/}{dev}{\EventY/}{alt};
    \devRecREdges{stratDev5-1}{\EventX/}{dev}{\EventY/}{alt}; \&

    \node(stratDev5-2) [state, dev] {};
    \devRecTreeNodeCoords{stratDev5-2};
    \deviationExampleStates{stratDev5-2}{dev}{dev};
    \botsExampleTerminalsTwo{stratDev5-2}{}{}{draw=black}{};
    \devRecRootEdges{stratDev5-2}{\Upgrade/}{alt}{\NotUpgrade/}{dev};
    \devRecLEdges{stratDev5-2}{\EventX/}{dev}{\EventY/}{alt};
    \devRecREdges{stratDev5-2}{\EventX/}{dev}{\EventY/}{alt}; \&
    \node(stratDev5-Value) {$+0.5$}; \\

\node(stratDev6-Label) [behaveLabel] {external\\deviation 6}; \&
    \node(stratDev6-1) [state, dev] {};
    \devRecTreeNodeCoords{stratDev6-1};
    \deviationExampleStates{stratDev6-1}{dev}{dev};
    \botsExampleTerminalsOne{stratDev6-1}{}{}{draw=black}{};
    \devRecRootEdges{stratDev6-1}{\Upgrade/}{alt}{\NotUpgrade/}{dev};
    \devRecLEdges{stratDev6-1}{\EventX/}{alt}{\EventY/}{dev};
    \devRecREdges{stratDev6-1}{\EventX/}{dev}{\EventY/}{alt}; \&

    \node(stratDev6-2) [state, dev] {};
    \devRecTreeNodeCoords{stratDev6-2};
    \deviationExampleStates{stratDev6-2}{dev}{dev};
    \botsExampleTerminalsTwo{stratDev6-2}{}{}{draw=black}{};
    \devRecRootEdges{stratDev6-2}{\Upgrade/}{alt}{\NotUpgrade/}{dev};
    \devRecLEdges{stratDev6-2}{\EventX/}{alt}{\EventY/}{dev};
    \devRecREdges{stratDev6-2}{\EventX/}{dev}{\EventY/}{alt}; \&
    \node(stratDev6-Value) {$+0.5$}; \\

\node(stratDev7-Label) [behaveLabel] {external\\deviation 7}; \&
    \node(stratDev7-1) [state, dev] {};
    \devRecTreeNodeCoords{stratDev7-1};
    \deviationExampleStates{stratDev7-1}{dev}{dev};
    \botsExampleTerminalsOne{stratDev7-1}{}{}{}{draw=black};
    \devRecRootEdges{stratDev7-1}{\Upgrade/}{alt}{\NotUpgrade/}{dev};
    \devRecLEdges{stratDev7-1}{\EventX/}{dev}{\EventY/}{alt};
    \devRecREdges{stratDev7-1}{\EventX/}{alt}{\EventY/}{dev}; \&

    \node(stratDev7-2) [state, dev] {};
    \devRecTreeNodeCoords{stratDev7-2};
    \deviationExampleStates{stratDev7-2}{dev}{dev};
    \botsExampleTerminalsTwo{stratDev7-2}{}{}{}{draw=black};
    \devRecRootEdges{stratDev7-2}{\Upgrade/}{alt}{\NotUpgrade/}{dev};
    \devRecLEdges{stratDev7-2}{\EventX/}{dev}{\EventY/}{alt};
    \devRecREdges{stratDev7-2}{\EventX/}{alt}{\EventY/}{dev}; \&
    \node(stratDev7-Value) {$+1$}; \\

\node(stratDev8-Label) [behaveLabel] {external\\deviation 8}; \&
    \node(stratDev8-1) [state, dev] {};
    \devRecTreeNodeCoords{stratDev8-1};
    \deviationExampleStates{stratDev8-1}{dev}{dev};
    \botsExampleTerminalsOne{stratDev8-1}{}{}{}{draw=black};
    \devRecRootEdges{stratDev8-1}{\Upgrade/}{alt}{\NotUpgrade/}{dev};
    \devRecLEdges{stratDev8-1}{\EventX/}{alt}{\EventY/}{dev};
    \devRecREdges{stratDev8-1}{\EventX/}{alt}{\EventY/}{dev}; \&

    \node(stratDev8-2) [state, dev] {};
    \devRecTreeNodeCoords{stratDev8-2};
    \deviationExampleStates{stratDev8-2}{dev}{dev};
    \botsExampleTerminalsTwo{stratDev8-2}{}{}{}{draw=black};
    \devRecRootEdges{stratDev8-2}{\Upgrade/}{alt}{\NotUpgrade/}{dev};
    \devRecLEdges{stratDev8-2}{\EventX/}{alt}{\EventY/}{dev};
    \devRecREdges{stratDev8-2}{\EventX/}{alt}{\EventY/}{dev}; \&
    \node(stratDev8-Value) {$+1$}; \\
  };
  \rowSepLine{stratDev1};
  \rowSepLine[color=grey]{stratDev2};
  \rowSepLine[color=grey]{stratDev3};
  \rowSepLine[color=grey]{stratDev4};
  \rowSepLine[color=grey]{stratDev5};
  \rowSepLine[color=grey]{stratDev6};
  \rowSepLine[color=grey]{stratDev7};
  \rowSepLine[color=grey]{stratDev8};
\end{tikzpicture}
\caption{
  External deviations in the extended battle-of-the-sexes example from player one's perspective.
  \PlayerTwoDesc/
  \ArrowDescriptions/
  \EvColDesc/
}
\label{fig:bots-strat-devs}
\end{figure}

\begin{figure}[htbp]
\centering
\begin{tikzpicture}[inner sep=0cm]
  \deviationExampleMatrix{
    \botsExampleHeader; \\[-0.3cm]

    \botsExampleFollowTrees; \\

\node(causDev1-Label) [behaveLabel] {causal\\deviation 1}; \&
    \node(causDev1-1) [state] {};
    \devRecTreeNodeCoords{causDev1-1};
    \deviationExampleStates{causDev1-1}{dev}{};
    \botsExampleTerminalsOne{causDev1-1}{}{}{}{draw=black};
    \devRecRootEdges{causDev1-1}{\Upgrade/}{alt}{\NotUpgrade/}{follow};
    \devRecLEdges{causDev1-1}{\EventX/}{dev}{\EventY/}{alt};
    \devRecREdges{causDev1-1}{\EventX/}{alt}{\EventY/}{follow}; \&

    \node(causDev1-2) [state] {};
    \devRecTreeNodeCoords{causDev1-2};
    \deviationExampleStates{causDev1-2}{dev}{};
    \botsExampleTerminalsTwo{causDev1-2}{}{}{draw=black}{};
    \devRecRootEdges{causDev1-2}{\Upgrade/}{alt}{\NotUpgrade/}{follow};
    \devRecLEdges{causDev1-2}{\EventX/}{dev}{\EventY/}{alt};
    \devRecREdges{causDev1-2}{\EventX/}{follow}{\EventY/}{alt}; \&
    \node(causDev1-Value) {$+1.5$}; \\

\node(causDev2-Label) [behaveLabel] {causal\\deviation 2}; \&
    \node(causDev2-1) [state] {};
    \devRecTreeNodeCoords{causDev2-1};
    \deviationExampleStates{causDev2-1}{dev}{};
    \botsExampleTerminalsOne{causDev2-1}{}{}{}{draw=black};
    \devRecRootEdges{causDev2-1}{\Upgrade/}{alt}{\NotUpgrade/}{follow};
    \devRecLEdges{causDev2-1}{\EventX/}{alt}{\EventY/}{dev};
    \devRecREdges{causDev2-1}{\EventX/}{alt}{\EventY/}{follow}; \&

    \node(causDev2-2) [state] {};
    \devRecTreeNodeCoords{causDev2-2};
    \deviationExampleStates{causDev2-2}{dev}{};
    \botsExampleTerminalsTwo{causDev2-2}{}{}{draw=black}{};
    \devRecRootEdges{causDev2-2}{\Upgrade/}{alt}{\NotUpgrade/}{follow};
    \devRecLEdges{causDev2-2}{\EventX/}{alt}{\EventY/}{dev};
    \devRecREdges{causDev2-2}{\EventX/}{follow}{\EventY/}{alt}; \&
    \node(causDev2-Value) {$+1.5$}; \\

\node(causDev3-Label) [behaveLabel] {causal\\deviation 3}; \&
    \node(causDev3-1) [state] {};
    \devRecTreeNodeCoords{causDev3-1};
    \deviationExampleStates{causDev3-1}{}{dev};
    \botsExampleTerminalsOne{causDev3-1}{}{}{draw=black}{};
    \devRecRootEdges{causDev3-1}{\Upgrade/}{alt}{\NotUpgrade/}{follow};
    \devRecLEdges{causDev3-1}{\EventX/}{alt}{\EventY/}{follow};
    \devRecREdges{causDev3-1}{\EventX/}{dev}{\EventY/}{alt}; \&

    \node(causDev3-2) [state] {};
    \devRecTreeNodeCoords{causDev3-2};
    \deviationExampleStates{causDev3-2}{}{dev};
    \botsExampleTerminalsTwo{causDev3-2}{}{}{draw=black}{};
    \devRecRootEdges{causDev3-2}{\Upgrade/}{alt}{\NotUpgrade/}{follow};
    \devRecLEdges{causDev3-2}{\EventX/}{follow}{\EventY/}{alt};
    \devRecREdges{causDev3-2}{\EventX/}{dev}{\EventY/}{alt}; \&
    \node(causDev3-Value) {$+0.5$}; \\

\node(causDev4-Label) [behaveLabel] {causal\\deviation 4}; \&
    \node(causDev4-1) [state] {};
    \devRecTreeNodeCoords{causDev4-1};
    \deviationExampleStates{causDev4-1}{}{dev};
    \botsExampleTerminalsOne{causDev4-1}{}{}{}{draw=black};
    \devRecRootEdges{causDev4-1}{\Upgrade/}{alt}{\NotUpgrade/}{follow};
    \devRecLEdges{causDev4-1}{\EventX/}{alt}{\EventY/}{follow};
    \devRecREdges{causDev4-1}{\EventX/}{alt}{\EventY/}{dev}; \&

    \node(causDev4-2) [state] {};
    \devRecTreeNodeCoords{causDev4-2};
    \deviationExampleStates{causDev4-2}{}{dev};
    \botsExampleTerminalsTwo{causDev4-2}{}{}{}{draw=black};
    \devRecRootEdges{causDev4-2}{\Upgrade/}{alt}{\NotUpgrade/}{follow};
    \devRecLEdges{causDev4-2}{\EventX/}{follow}{\EventY/}{alt};
    \devRecREdges{causDev4-2}{\EventX/}{alt}{\EventY/}{dev}; \&
    \node(causDev4-Value) {$+1$}; \\
  };
  \rowSepLine{causDev1};
  \rowSepLine[color=grey]{causDev2};
  \rowSepLine[color=grey]{causDev3};
  \rowSepLine[color=grey]{causDev4};
\end{tikzpicture}
\caption{
  Blind causal deviations in the extended battle-of-the-sexes example from player one's perspective.
  The full set of causal deviations includes all of the external deviations, but these are duplicates of \cref{fig:bots-strat-devs} so they are omitted.
  \PlayerTwoDesc/
  \ArrowDescriptions/
  \EvColDesc/
}
\label{fig:bots-caus-devs}
\end{figure}

\begin{figure}[htbp]
\centering
\begin{tikzpicture}[inner sep=0cm]
  \deviationExampleMatrix{
    \botsExampleHeader; \\[-0.3cm]

    \botsExampleFollowTrees; \\

\node(cfDev1-Label) [behaveLabel] {counterfactual\\deviation 1}; \&
    \node(cfDev1-1) [state, dev] {};
    \devRecTreeNodeCoords{cfDev1-1};
    \deviationExampleStates{cfDev1-1}{}{};
    \botsExampleTerminalsOne{cfDev1-1}{}{draw=black}{}{};
    \devRecRootEdges{cfDev1-1}{\Upgrade/}{dev}{\NotUpgrade/}{alt};
    \devRecLEdges{cfDev1-1}{\EventX/}{alt}{\EventY/}{follow};
    \devRecREdges{cfDev1-1}{\EventX/}{alt}{\EventY/}{follow}; \&

    \node(cfDev1-2) [state, dev] {};
    \devRecTreeNodeCoords{cfDev1-2};
    \deviationExampleStates{cfDev1-2}{}{};
    \botsExampleTerminalsTwo{cfDev1-2}{draw=black}{}{}{};
    \devRecRootEdges{cfDev1-2}{\Upgrade/}{dev}{\NotUpgrade/}{alt};
    \devRecLEdges{cfDev1-2}{\EventX/}{follow}{\EventY/}{alt};
    \devRecREdges{cfDev1-2}{\EventX/}{follow}{\EventY/}{alt}; \&
    \node(cfDev1-Value) {$+2.5$}; \\

\node(cfDev2-Label) [behaveLabel] {counterfactual\\deviation 2}; \&
    \node(cfDev2-1) [state, dev] {};
    \devRecTreeNodeCoords{cfDev2-1};
    \deviationExampleStates{cfDev2-1}{dev}{};
    \botsExampleTerminalsOne{cfDev2-1}{draw=black}{}{}{};
    \devRecRootEdges{cfDev2-1}{\Upgrade/}{dev}{\NotUpgrade/}{alt};
    \devRecLEdges{cfDev2-1}{\EventX/}{dev}{\EventY/}{alt};
    \devRecREdges{cfDev2-1}{\EventX/}{alt}{\EventY/}{follow}; \&

    \node(cfDev2-2) [state, dev] {};
    \devRecTreeNodeCoords{cfDev2-2};
    \deviationExampleStates{cfDev2-2}{dev}{};
    \botsExampleTerminalsTwo{cfDev2-2}{draw=black}{}{}{};
    \devRecRootEdges{cfDev2-2}{\Upgrade/}{dev}{\NotUpgrade/}{alt};
    \devRecLEdges{cfDev2-2}{\EventX/}{dev}{\EventY/}{alt};
    \devRecREdges{cfDev2-2}{\EventX/}{follow}{\EventY/}{alt}; \&
    \node(cfDev2-Value) {$+1$}; \\

\node(cfDev3-Label) [behaveLabel] {counterfactual\\deviation 3}; \&
    \node(cfDev3-1) [state, dev] {};
    \devRecTreeNodeCoords{cfDev3-1};
    \deviationExampleStates{cfDev3-1}{dev}{};
    \botsExampleTerminalsOne{cfDev3-1}{}{draw=black}{}{};
    \devRecRootEdges{cfDev3-1}{\Upgrade/}{dev}{\NotUpgrade/}{alt};
    \devRecLEdges{cfDev3-1}{\EventX/}{alt}{\EventY/}{dev};
    \devRecREdges{cfDev3-1}{\EventX/}{alt}{\EventY/}{follow}; \&

    \node(cfDev3-2) [state, dev] {};
    \devRecTreeNodeCoords{cfDev3-2};
    \deviationExampleStates{cfDev3-2}{dev}{};
    \botsExampleTerminalsTwo{cfDev3-2}{}{draw=black}{}{};
    \devRecRootEdges{cfDev3-2}{\Upgrade/}{dev}{\NotUpgrade/}{alt};
    \devRecLEdges{cfDev3-2}{\EventX/}{alt}{\EventY/}{dev};
    \devRecREdges{cfDev3-2}{\EventX/}{follow}{\EventY/}{alt}; \&
    \node(cfDev3-Value) {$+1.5$}; \\

\node(cfDev4-Label) [behaveLabel] {counterfactual\\deviation 4}; \&
    \node(cfDev4-1) [state, dev] {};
    \devRecTreeNodeCoords{cfDev4-1};
    \deviationExampleStates{cfDev4-1}{}{};
    \botsExampleTerminalsOne{cfDev4-1}{}{}{}{draw=black};
    \devRecRootEdges{cfDev4-1}{\Upgrade/}{alt}{\NotUpgrade/}{dev};
    \devRecLEdges{cfDev4-1}{\EventX/}{alt}{\EventY/}{follow};
    \devRecREdges{cfDev4-1}{\EventX/}{alt}{\EventY/}{follow}; \&

    \node(cfDev4-2) [state, dev] {};
    \devRecTreeNodeCoords{cfDev4-2};
    \deviationExampleStates{cfDev4-2}{}{};
    \botsExampleTerminalsTwo{cfDev4-2}{}{}{draw=black}{};
    \devRecRootEdges{cfDev4-2}{\Upgrade/}{alt}{\NotUpgrade/}{dev};
    \devRecLEdges{cfDev4-2}{\EventX/}{follow}{\EventY/}{alt};
    \devRecREdges{cfDev4-2}{\EventX/}{follow}{\EventY/}{alt}; \&
    \node(cfDev4-Value) {$+1.5$}; \\

\node(cfDev5-Label) [behaveLabel] {counterfactual\\deviation 5}; \&
    \node(cfDev5-1) [state, dev] {};
    \devRecTreeNodeCoords{cfDev5-1};
    \deviationExampleStates{cfDev5-1}{}{dev};
    \botsExampleTerminalsOne{cfDev5-1}{}{}{draw=black}{};
    \devRecRootEdges{cfDev5-1}{\Upgrade/}{alt}{\NotUpgrade/}{dev};
    \devRecLEdges{cfDev5-1}{\EventX/}{alt}{\EventY/}{follow};
    \devRecREdges{cfDev5-1}{\EventX/}{dev}{\EventY/}{alt}; \&

    \node(cfDev5-2) [state, dev] {};
    \devRecTreeNodeCoords{cfDev5-2};
    \deviationExampleStates{cfDev5-2}{}{dev};
    \botsExampleTerminalsTwo{cfDev5-2}{}{}{draw=black}{};
    \devRecRootEdges{cfDev5-2}{\Upgrade/}{alt}{\NotUpgrade/}{dev};
    \devRecLEdges{cfDev5-2}{\EventX/}{follow}{\EventY/}{alt};
    \devRecREdges{cfDev5-2}{\EventX/}{dev}{\EventY/}{alt}; \&
    \node(cfDev5-Value) {$+0.5$}; \\

\node(cfDev6-Label) [behaveLabel] {counterfactual\\deviation 6}; \&
    \node(cfDev6-1) [state, dev] {};
    \devRecTreeNodeCoords{cfDev6-1};
    \deviationExampleStates{cfDev6-1}{}{dev};
    \botsExampleTerminalsOne{cfDev6-1}{}{}{}{draw=black};
    \devRecRootEdges{cfDev6-1}{\Upgrade/}{alt}{\NotUpgrade/}{dev};
    \devRecLEdges{cfDev6-1}{\EventX/}{alt}{\EventY/}{follow};
    \devRecREdges{cfDev6-1}{\EventX/}{alt}{\EventY/}{dev}; \&

    \node(cfDev6-2) [state, dev] {};
    \devRecTreeNodeCoords{cfDev6-2};
    \deviationExampleStates{cfDev6-2}{}{dev};
    \botsExampleTerminalsTwo{cfDev6-2}{}{}{}{draw=black};
    \devRecRootEdges{cfDev6-2}{\Upgrade/}{alt}{\NotUpgrade/}{dev};
    \devRecLEdges{cfDev6-2}{\EventX/}{follow}{\EventY/}{alt};
    \devRecREdges{cfDev6-2}{\EventX/}{alt}{\EventY/}{dev}; \&
    \node(cfDev6-Value) {$+1$}; \\
  };

  \rowSepLine{cfDev1};
  \rowSepLine[color=grey]{cfDev2};
  \rowSepLine[color=grey]{cfDev3};
  \rowSepLine[color=grey]{cfDev4};
  \rowSepLine[color=grey]{cfDev5};
  \rowSepLine[color=grey]{cfDev6};
\end{tikzpicture}
\caption{
  Blind counterfactual deviations in the extended battle-of-the-sexes example from player one's perspective.
  \PlayerTwoDesc/
  \ArrowDescriptions/
  \EvColDesc/
}
\label{fig:bots-cf-devs}
\end{figure}
 
\subsection{Extended Shapley's Game}
\cref{fig:ext-shapley} is a visualization of our extended Shapley's game.
\cref{fig:shapley-medium-big-alt} shows the EFCCE and EFCE gaps of CFR's empirical distribution across iterations in this game with bonuses of 0.3 and 30.

\begin{figure}[htbp]
\centering
\begin{tikzpicture}[inner sep=0cm]
  \tikzstyle{state} = [circle, draw=black, minimum size=0.3cm, inner sep=0.05cm]
  \tikzstyle{arrow} = [thick, >=Stealth, -{>[scale=0.7]}]

  \node(root) [state] {\small 1};
  \node(r) [state, below left=0.7cm and 4.1cm of root.south west, anchor=north] {\small 1};
  \draw[arrow] (root) -- node[anchor=south east, yshift=0.04cm] {\small \Rock/} (r);

  \node(p) [state, below=0.7cm of root.south, anchor=north] {\small 1};
  \draw[arrow] (root) -- node[anchor=east, xshift=-0.04cm] {\small \Paper/} (p);

  \node(s) [state, below right=0.7cm and 4.1cm of root.south east, anchor=north] {\small 1};
  \draw[arrow] (root) -- node[anchor=south west, yshift=0.04cm] {\small \Scissors/} (s);

\node(rr?) [state, below left=0.7cm and 0.75cm of r.south west, anchor=north] {\small 2};
  \draw[arrow] (r) -- node[anchor=south east] {\small \PredRock/} (rr?);
  \node(rNr?) [state, below right=0.7cm and 0.75cm of r.south east, anchor=north] {\small 2};
  \draw[arrow] (r) -- node[anchor=south west] {\small \PredNotRock/} (rNr?);
  \draw[thick, dashed] (rr?.east) parabola[bend pos=0.5] bend +(0, 0.2cm) (rNr?.west);

  \node(pr?) [state, below left=0.7cm and 0.75cm of p.south west, anchor=north] {\small 2};
  \draw[arrow] (p) -- node[anchor=south east] {\small \PredRock/} (pr?);
  \node(pNr?) [state, below right=0.7cm and 0.75cm of p.south east, anchor=north] {\small 2};
  \draw[arrow] (p) -- node[anchor=south west] {\small \PredNotRock/} (pNr?);
  \draw[thick, dashed] (rNr?.east) parabola[bend pos=0.5] bend +(0, 0.2cm) (pr?.west);
  \draw[thick, dashed] (pr?.east) parabola[bend pos=0.5] bend +(0, 0.2cm) (pNr?.west);

  \node(sr?) [state, below left=0.7cm and 0.75cm of s.south west, anchor=north] {\small 2};
  \draw[arrow] (s) -- node[anchor=south east] {\small \PredRock/} (sr?);
  \node(sNr?) [state, below right=0.7cm and 0.75cm of s.south east, anchor=north] {\small 2};
  \draw[arrow] (s) -- node[anchor=south west] {\small \PredNotRock/} (sNr?);
  \draw[thick, dashed] (pNr?.east) parabola[bend pos=0.5] bend +(0, 0.2cm) (sr?.west);
  \draw[thick, dashed] (sr?.east) parabola[bend pos=0.5] bend +(0, 0.2cm) (sNr?.west);

\node(rr?r) [util, below left=0.4cm and 0.8cm of rr?.south west, anchor=north] {\small $-1+b$, $-1$};
  \draw[arrow] (rr?) -- node[anchor=south east] {\small \Rock/} (rr?r);
  \node(rr?p) [util, below=2cm of rr?.south, anchor=north] {\small $-1$, $+1$};
  \draw[arrow] (rr?) -- node[anchor=east, xshift=-0.04cm] {\small \Paper/} (rr?p);
  \node(rr?s) [util, below right=3cm and 1cm of rr?.south east, anchor=north] {\small $+1$, $-1$};
  \draw[arrow] (rr?) -- node[anchor=south west] {\small \Scissors/} (rr?s);

  \node(rNr?r) [util, below left=0.4cm and 0.8cm of rNr?.south west, anchor=north] {\small $-1$, $-1$};
  \draw[arrow] (rNr?) -- node[anchor=south east] {\small \Rock/} (rNr?r);
  \node(rNr?p) [util, below=2cm of rNr?.south, anchor=north] {\small $-1 + \frac{b}{3}$, $+1$};
  \draw[arrow] (rNr?) -- node[anchor=east, xshift=-0.04cm] {\small \Paper/} (rNr?p);
  \node(rNr?s) [util, below right=3cm and 1cm of rNr?.south east, anchor=north] {\small $+1 + \frac{b}{3}$, $-1$};
  \draw[arrow] (rNr?) -- node[anchor=south west] {\small \Scissors/} (rNr?s);

  \node(pr?r) [util, below left=0.4cm and 0.8cm of pr?.south west, anchor=north] {\small $+1 + b$, $-1$};
  \draw[arrow] (pr?) -- node[anchor=south east] {\small \Rock/} (pr?r);
  \node(pr?p) [util, below=2cm of pr?.south, anchor=north] {\small $-1$, $-1$};
  \draw[arrow] (pr?) -- node[anchor=east, xshift=-0.04cm] {\small \Paper/} (pr?p);
  \node(pr?s) [util, below right=3cm and 1cm of pr?.south east, anchor=north] {\small $-1$, $+1$};
  \draw[arrow] (pr?) -- node[anchor=south west] {\small \Scissors/} (pr?s);

  \node(pNr?r) [util, below left=0.4cm and 0.8cm of pNr?.south west, anchor=north] {\small $+1$, $-1$};
  \draw[arrow] (pNr?) -- node[anchor=south east] {\small \Rock/} (pNr?r);
  \node(pNr?p) [util, below=2cm of pNr?.south, anchor=north] {\small $-1 + \frac{b}{3}$, $-1$};
  \draw[arrow] (pNr?) -- node[anchor=east, xshift=-0.04cm] {\small \Paper/} (pNr?p);
  \node(pNr?s) [util, below right=3cm and 1cm of pNr?.south east, anchor=north] {\small $-1 + \frac{b}{3}$, $+1$};
  \draw[arrow] (pNr?) -- node[anchor=south west] {\small \Scissors/} (pNr?s);

  \node(sr?r) [util, below left=0.4cm and 0.8cm of sr?.south west, anchor=north] {\small $-1 + b$, $+1$};
  \draw[arrow] (sr?) -- node[anchor=south east] {\small \Rock/} (sr?r);
  \node(sr?p) [util, below=2cm of sr?.south, anchor=north] {\small $+1$, $-1$};
  \draw[arrow] (sr?) -- node[anchor=east, xshift=-0.04cm] {\small \Paper/} (sr?p);
  \node(sr?s) [util, below right=3cm and 1cm of sr?.south east, anchor=north] {\small $-1$, $-1$};
  \draw[arrow] (sr?) -- node[anchor=south west] {\small \Scissors/} (sr?s);

  \node(sNr?r) [util, below left=0.4cm and 0.8cm of sNr?.south west, anchor=north] {\small $-1$, $+1$};
  \draw[arrow] (sNr?) -- node[anchor=south east] {\small \Rock/} (sNr?r);
  \node(sNr?p) [util, below=2cm of sNr?.south, anchor=north] {\small $+1 + \frac{b}{3}$, $-1$};
  \draw[arrow] (sNr?) -- node[anchor=east, xshift=-0.04cm] {\small \Paper/} (sNr?p);
  \node(sNr?s) [util, below right=3cm and 1cm of sNr?.south east, anchor=north] {\small $-1 + \frac{b}{3}$, $-1$};
  \draw[arrow] (sNr?) -- node[anchor=south west] {\small \Scissors/} (sNr?s);
\end{tikzpicture}
\caption{An extended Shapley's game where the first player privately predicts whether or not their opponent will play Rock, denoted \PredRock/ and \PredNotRock/ respectively. Rock, Paper, and Scissors are denoted by \Rock/, \Paper/, and \Scissors/, respectively.
\InfoSetDesc/}
\label{fig:ext-shapley}
\end{figure}

\begin{figure}[htbp]
\centering
  \includegraphics[width=0.5\linewidth]{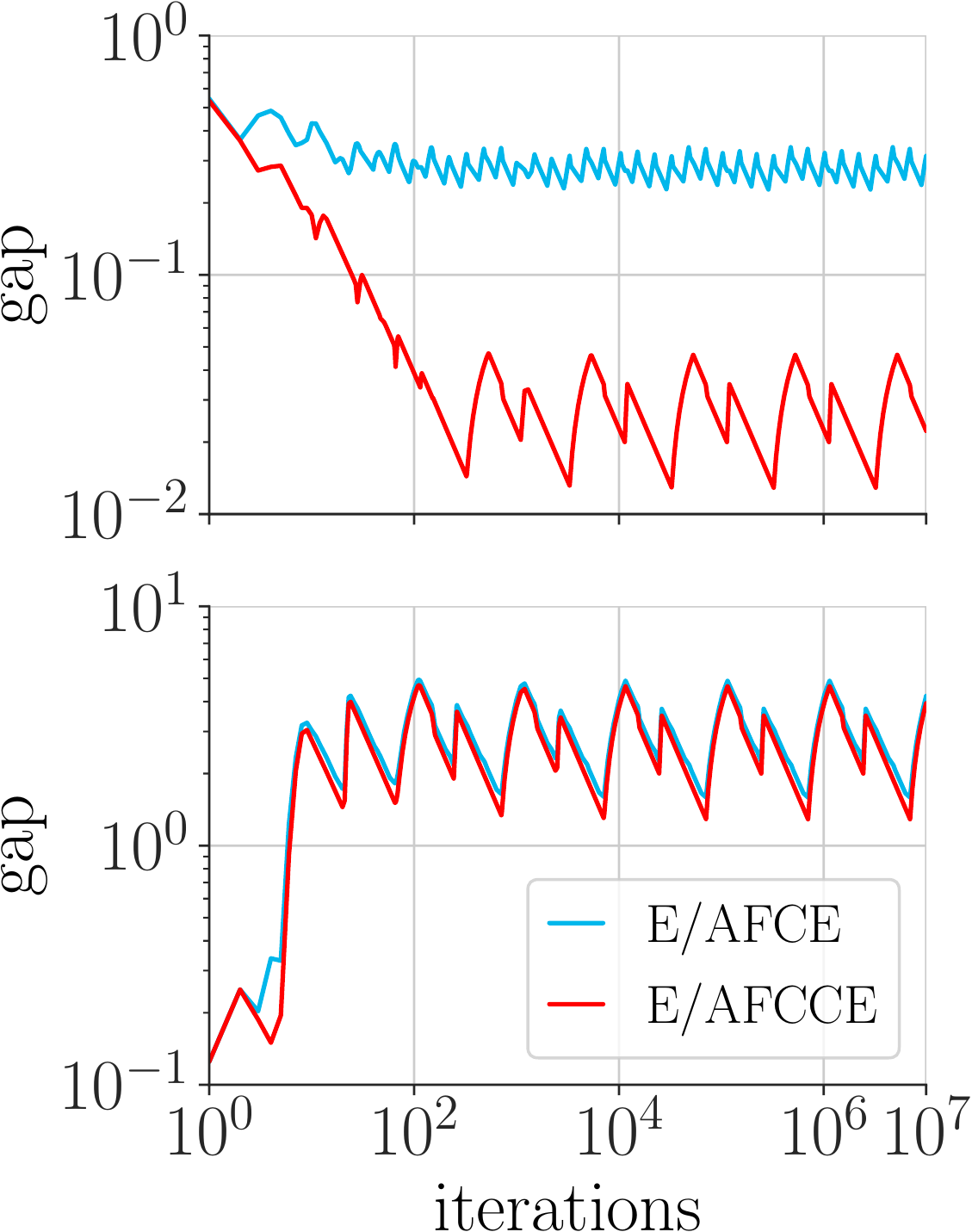}
  \caption{The gap between alternating-update~\parencite{burch2019revisiting} CFR's self-play empirical distribution and an extensive-form or agent-form (C)CE (E/AF(C)CE) in the extended Shapley's game with a medium bonus ($b=0.3$) (top) and a big bonus ($b=30$) (bottom).}
  \label{fig:shapley-medium-big-alt}
\end{figure}
 
\subsection{Extended Matching Pennies}
\cref{fig:ext-matching-pennies} is an illustration of the extensive-form tree of the extended matching pennies game introduced in \cref{sec:cf-dev-properties}.
The recommendations
$\set*{
  \tuple*{\mNotMatch, \; \mHeads \given \mMatch, \; \mTails \given \mNotMatch},
  \tuple*{\mHeads}
}$
and
$\set*{
  \tuple*{\mMatch, \; \mHeads \given \mMatch, \; \mTails \given \mNotMatch},
  \tuple*{\mTails}
}$
are visualized from player one's perspective at the top of \cref{fig:mp-strat-devs,fig:mp-caus-devs,fig:mp-cf-devs}..
\cref{fig:mp-strat-devs} shows all of the external deviations.
\cref{fig:mp-caus-devs} shows all of the blind causal deviations that are not also external deviations.
\cref{fig:mp-cf-devs} shows all of the blind counterfactual deviations that are not also external deviations.
The complete set of blind action deviations is composed of the deviations in \cref{fig:mp-caus-devs} plus counterfactual deviations 1 and 4 in \cref{fig:mp-cf-devs}.

To see that this is a CFCE, note that conditioning on the recommendations in the Match or Not Match information sets cannot improve a counterfactual deviation's value since the recommended actions are the same under both recommendations.
Informed counterfactual deviations must re-correlate after making a conditional deviation, but changing the action at the root only cannot change the player's value, so there is no beneficial informed counterfactual deviation that triggers on a root recommendation.

Note that blind causal deviations 2, which is also an action deviation, achieves a value of $+1 > 0$, and it is the only beneficial deviation.
Therefore, this recommendation distribution is a CFCE with a beneficial blind causal or action deviation, ensuring that it is not an EFCCE or an AFCCE.

\begin{figure}[htbp]
\centering
\begin{tikzpicture}[inner sep=0cm]
  \tikzstyle{state} = [circle, draw=black, minimum size=0.3cm, inner sep=0.05cm]
  \tikzstyle{arrow} = [thick, >=Stealth, -{>[scale=0.7]}]

  \node(root) [state] {\small 1};
  \node(m) [state, below left=\shortAYLength and 2.5cm of root.south west, anchor=north] {\small 1};
  \draw[arrow] (root) -- node[anchor=south east] {\small \Match/} (m);

  \node(nm) [state, below right=\shortAYLength and 2.5cm of root.south east, anchor=north] {\small 1};
  \draw[arrow] (root) -- node[anchor=south west] {\small \NotMatch/} (nm);

  \node(mH) [state, below left=\shortAYLength and 1cm of m.south west, anchor=north] {\small 2};
  \draw[arrow] (m) -- node[anchor=south east] {\small \Heads/} (mH);

  \node(mT) [state, below right=\shortAYLength and 1cm of m.south east, anchor=north] {\small 2};
  \draw[arrow] (m) -- node[anchor=south west] {\small \Tails/} (mT);

  \node(nmH) [state, below left=\shortAYLength and 1cm of nm.south west, anchor=north] {\small 2};
  \draw[arrow] (nm) -- node[anchor=south east] {\small \Heads/} (nmH);

  \node(nmT) [state, below right=\shortAYLength and 1cm of nm.south east, anchor=north] {\small 2};
  \draw[arrow] (nm) -- node[anchor=south west] {\small \Tails/} (nmT);

  \draw[thick, dashed] (mH.east) parabola[bend pos=0.5] bend +(0, 0.2cm) (mT.west);
  \draw[thick, dashed] (mT.east) parabola[bend pos=0.5] bend +(0, 0.2cm) (nmH.west);
  \draw[thick, dashed] (nmH.east) parabola[bend pos=0.5] bend +(0, 0.2cm) (nmT.west);

  \node(mHH) [util, below left=\shortAYLength and 0.5cm of mH.south west, anchor=north] {$+1$};
  \draw[arrow] (mH) -- node[anchor=south east] {\small \Heads/} (mHH);

  \node(mTH) [util, below left=\shortAYLength and 0.5cm of mT.south west, anchor=north] {$-1$};
  \draw[arrow] (mT) -- node[anchor=south east] {\small \Heads/} (mTH);

  \node(mHT) [util, below right=\shortAYLength and 0.5cm of mH.south east, anchor=north] {$-1$};
  \draw[arrow] (mH) -- node[anchor=south west] {\small \Tails/} (mHT);

  \node(mTT) [util, below right=\shortAYLength and 0.5cm of mT.south east, anchor=north] {$+1$};
  \draw[arrow] (mT) -- node[anchor=south west] {\small \Tails/} (mTT);

  \node(nmHH) [util, below left=\shortAYLength and 0.5cm of nmH.south west, anchor=north] {$-1$};
  \draw[arrow] (nmH) -- node[anchor=south east] {\small \Heads/} (nmHH);

  \node(nmTH) [util, below left=\shortAYLength and 0.5cm of nmT.south west, anchor=north] {$+1$};
  \draw[arrow] (nmT) -- node[anchor=south east] {\small \Heads/} (nmTH);

  \node(nmHT) [util, below right=\shortAYLength and 0.5cm of nmH.south east, anchor=north] {$+1$};
  \draw[arrow] (nmH) -- node[anchor=south west] {\small \Tails/} (nmHT);

  \node(nmTT) [util, below right=\shortAYLength and 0.5cm of nmT.south east, anchor=north] {$-1$};
  \draw[arrow] (nmT) -- node[anchor=south west] {\small \Tails/} (nmTT);
\end{tikzpicture}
\caption{An extended matching pennies game where the first player privately chooses whether their goal is to match or not match. The game is zero-sum so only the payoffs for player one are shown.
\InfoSetDesc/}
\label{fig:ext-matching-pennies}
\end{figure}

\begin{figure}[htbp]
\centering
\begin{tikzpicture}[inner sep=0cm]
  \deviationExampleMatrix{
    \mpExampleHeader; \\[-0.3cm]

    \mpExampleFollowTrees; \\

\node(stratDev1-Label) [behaveLabel] {external\\deviation 1}; \&
    \node(stratDev1-1) [state, dev] {};
    \devRecTreeNodeCoords{stratDev1-1};
    \deviationExampleStates{stratDev1-1}{dev}{dev};
    \mpExampleTerminalsOne{stratDev1-1}{draw=black}{}{}{};
    \devRecRootEdges{stratDev1-1}{M}{dev}{$\neg\text{M}$}{alt};
    \devRecLEdges{stratDev1-1}{H}{dev}{T}{alt};
    \devRecREdges{stratDev1-1}{H}{dev}{T}{alt}; \&

    \node(stratDev1-2) [state, dev] {};
    \devRecTreeNodeCoords{stratDev1-2};
    \deviationExampleStates{stratDev1-2}{dev}{dev};
    \mpExampleTerminalsTwo{stratDev1-2}{draw=black}{}{}{};
    \devRecRootEdges{stratDev1-2}{M}{dev}{$\neg\text{M}$}{alt};
    \devRecLEdges{stratDev1-2}{H}{dev}{T}{alt};
    \devRecREdges{stratDev1-2}{H}{dev}{T}{alt}; \&
    \node(stratDev1-Value) {0}; \\

\node(stratDev2-Label) [behaveLabel] {external\\deviation 2}; \&
    \node(stratDev2-1) [state, dev] {};
    \devRecTreeNodeCoords{stratDev2-1};
    \deviationExampleStates{stratDev2-1}{dev}{dev};
    \mpExampleTerminalsOne{stratDev2-1}{draw=black}{}{}{};
    \devRecRootEdges{stratDev2-1}{M}{dev}{$\neg\text{M}$}{alt};
    \devRecLEdges{stratDev2-1}{H}{dev}{T}{alt};
    \devRecREdges{stratDev2-1}{H}{alt}{T}{dev}; \&

    \node(stratDev2-2) [state] {};
    \devRecTreeNodeCoords{stratDev2-2};
    \deviationExampleStates{stratDev2-2}{dev}{dev};
    \mpExampleTerminalsTwo{stratDev2-2}{draw=black}{}{}{};
    \devRecRootEdges{stratDev2-2}{M}{dev}{$\neg\text{M}$}{alt};
    \devRecLEdges{stratDev2-2}{H}{dev}{T}{alt};
    \devRecREdges{stratDev2-2}{H}{alt}{T}{dev}; \&
    \node(stratDev2-Value) {0}; \\

\node(stratDev3-Label) [behaveLabel] {external\\deviation 3}; \&
    \node(stratDev3-1) [state, dev] {};
    \devRecTreeNodeCoords{stratDev3-1};
    \deviationExampleStates{stratDev3-1}{dev}{dev};
    \mpExampleTerminalsOne{stratDev3-1}{}{draw=black}{}{};
    \devRecRootEdges{stratDev3-1}{M}{dev}{$\neg\text{M}$}{alt};
    \devRecLEdges{stratDev3-1}{H}{alt}{T}{dev};
    \devRecREdges{stratDev3-1}{H}{dev}{T}{alt}; \&

    \node(stratDev3-2) [state, dev] {};
    \devRecTreeNodeCoords{stratDev3-2};
    \deviationExampleStates{stratDev3-2}{dev}{dev};
    \mpExampleTerminalsTwo{stratDev3-2}{}{draw=black}{}{};
    \devRecRootEdges{stratDev3-2}{M}{dev}{$\neg\text{M}$}{alt};
    \devRecLEdges{stratDev3-2}{H}{alt}{T}{dev};
    \devRecREdges{stratDev3-2}{H}{dev}{T}{alt}; \&
    \node(stratDev3-Value) {0}; \\

\node(stratDev4-Label) [behaveLabel] {external\\deviation 4}; \&
    \node(stratDev4-1) [state, dev] {};
    \devRecTreeNodeCoords{stratDev4-1};
    \deviationExampleStates{stratDev4-1}{dev}{dev};
    \mpExampleTerminalsOne{stratDev4-1}{}{draw=black}{}{};
    \devRecRootEdges{stratDev4-1}{M}{dev}{$\neg\text{M}$}{alt};
    \devRecLEdges{stratDev4-1}{H}{alt}{T}{dev};
    \devRecREdges{stratDev4-1}{H}{alt}{T}{dev}; \&

    \node(stratDev4-2) [state, dev] {};
    \devRecTreeNodeCoords{stratDev4-2};
    \deviationExampleStates{stratDev4-2}{dev}{dev};
    \mpExampleTerminalsTwo{stratDev4-2}{}{draw=black}{}{};
    \devRecRootEdges{stratDev4-2}{M}{dev}{$\neg\text{M}$}{alt};
    \devRecLEdges{stratDev4-2}{H}{alt}{T}{dev};
    \devRecREdges{stratDev4-2}{H}{alt}{T}{dev}; \&
    \node(stratDev4-Value) {0}; \\

\node(stratDev5-Label) [behaveLabel] {external\\deviation 5}; \&
    \node(stratDev5-1) [state, dev] {};
    \devRecTreeNodeCoords{stratDev5-1};
    \deviationExampleStates{stratDev5-1}{dev}{dev};
    \mpExampleTerminalsOne{stratDev5-1}{}{}{draw=black}{};
    \devRecRootEdges{stratDev5-1}{M}{alt}{$\neg\text{M}$}{dev};
    \devRecLEdges{stratDev5-1}{H}{dev}{T}{alt};
    \devRecREdges{stratDev5-1}{H}{dev}{T}{alt}; \&

    \node(stratDev5-2) [state, dev] {};
    \devRecTreeNodeCoords{stratDev5-2};
    \deviationExampleStates{stratDev5-2}{dev}{dev};
    \mpExampleTerminalsTwo{stratDev5-2}{}{}{draw=black}{};
    \devRecRootEdges{stratDev5-2}{M}{alt}{$\neg\text{M}$}{dev};
    \devRecLEdges{stratDev5-2}{H}{dev}{T}{alt};
    \devRecREdges{stratDev5-2}{H}{dev}{T}{alt}; \&
    \node(stratDev5-Value) {0}; \\

\node(stratDev6-Label) [behaveLabel] {external\\deviation 6}; \&
    \node(stratDev6-1) [state, dev] {};
    \devRecTreeNodeCoords{stratDev6-1};
    \deviationExampleStates{stratDev6-1}{dev}{dev};
    \mpExampleTerminalsOne{stratDev6-1}{}{}{draw=black}{};
    \devRecRootEdges{stratDev6-1}{M}{alt}{$\neg\text{M}$}{dev};
    \devRecLEdges{stratDev6-1}{H}{alt}{T}{dev};
    \devRecREdges{stratDev6-1}{H}{dev}{T}{alt}; \&

    \node(stratDev6-2) [state, dev] {};
    \devRecTreeNodeCoords{stratDev6-2};
    \deviationExampleStates{stratDev6-2}{dev}{dev};
    \mpExampleTerminalsTwo{stratDev6-2}{}{}{draw=black}{};
    \devRecRootEdges{stratDev6-2}{M}{alt}{$\neg\text{M}$}{dev};
    \devRecLEdges{stratDev6-2}{H}{alt}{T}{dev};
    \devRecREdges{stratDev6-2}{H}{dev}{T}{alt}; \&
    \node(stratDev6-Value) {0}; \\

\node(stratDev7-Label) [behaveLabel] {external\\deviation 7}; \&
    \node(stratDev7-1) [state, dev] {};
    \devRecTreeNodeCoords{stratDev7-1};
    \deviationExampleStates{stratDev7-1}{dev}{dev};
    \mpExampleTerminalsOne{stratDev7-1}{}{}{}{draw=black};
    \devRecRootEdges{stratDev7-1}{M}{alt}{$\neg\text{M}$}{dev};
    \devRecLEdges{stratDev7-1}{H}{dev}{T}{alt};
    \devRecREdges{stratDev7-1}{H}{alt}{T}{dev}; \&

    \node(stratDev7-2) [state, dev] {};
    \devRecTreeNodeCoords{stratDev7-2};
    \deviationExampleStates{stratDev7-2}{dev}{dev};
    \mpExampleTerminalsTwo{stratDev7-2}{}{}{}{draw=black};
    \devRecRootEdges{stratDev7-2}{M}{alt}{$\neg\text{M}$}{dev};
    \devRecLEdges{stratDev7-2}{H}{dev}{T}{alt};
    \devRecREdges{stratDev7-2}{H}{alt}{T}{dev}; \&
    \node(stratDev7-Value) {0}; \\

\node(stratDev8-Label) [behaveLabel] {external\\deviation 8}; \&
    \node(stratDev8-1) [state, dev] {};
    \devRecTreeNodeCoords{stratDev8-1};
    \deviationExampleStates{stratDev8-1}{dev}{dev};
    \mpExampleTerminalsOne{stratDev8-1}{}{}{}{draw=black};
    \devRecRootEdges{stratDev8-1}{M}{alt}{$\neg\text{M}$}{dev};
    \devRecLEdges{stratDev8-1}{H}{alt}{T}{dev};
    \devRecREdges{stratDev8-1}{H}{alt}{T}{dev}; \&

    \node(stratDev8-2) [state, dev] {};
    \devRecTreeNodeCoords{stratDev8-2};
    \deviationExampleStates{stratDev8-2}{dev}{dev};
    \mpExampleTerminalsTwo{stratDev8-2}{}{}{}{draw=black};
    \devRecRootEdges{stratDev8-2}{M}{alt}{$\neg\text{M}$}{dev};
    \devRecLEdges{stratDev8-2}{H}{alt}{T}{dev};
    \devRecREdges{stratDev8-2}{H}{alt}{T}{dev}; \&
    \node(stratDev8-Value) {0}; \\
  };
  \rowSepLine{stratDev1};
  \rowSepLine[color=grey]{stratDev2};
  \rowSepLine[color=grey]{stratDev3};
  \rowSepLine[color=grey]{stratDev4};
  \rowSepLine[color=grey]{stratDev5};
  \rowSepLine[color=grey]{stratDev6};
  \rowSepLine[color=grey]{stratDev7};
  \rowSepLine[color=grey]{stratDev8};
\end{tikzpicture}
\caption{
  External deviations in the extended matching pennies example from player one's perspective.
  \PlayerTwoDesc/
  \ArrowDescriptions/
  \EvColDesc/
}
\label{fig:mp-strat-devs}
\end{figure}

\begin{figure}[htbp]
\centering
\begin{tikzpicture}[inner sep=0cm]
  \deviationExampleMatrix{
    \mpExampleHeader; \\[-0.3cm]

    \mpExampleFollowTrees; \\

\node(causDev1-Label) [behaveLabel] {causal\\deviation 1}; \&
    \node(causDev1-1) [state] {};
    \devRecTreeNodeCoords{causDev1-1};
    \deviationExampleStates{causDev1-1}{dev}{};
    \mpExampleTerminalsOne{causDev1-1}{}{}{}{draw=black};
    \devRecRootEdges{causDev1-1}{M}{alt}{$\neg\text{M}$}{follow};
    \devRecLEdges{causDev1-1}{H}{dev}{T}{alt};
    \devRecREdges{causDev1-1}{H}{alt}{T}{follow}; \&

    \node(causDev1-2) [state] {};
    \devRecTreeNodeCoords{causDev1-2};
    \deviationExampleStates{causDev1-2}{dev}{};
    \mpExampleTerminalsTwo{causDev1-2}{draw=black}{}{}{};
    \devRecRootEdges{causDev1-2}{M}{follow}{$\neg\text{M}$}{alt};
    \devRecLEdges{causDev1-2}{H}{dev}{T}{alt};
    \devRecREdges{causDev1-2}{H}{alt}{T}{follow}; \&
    \node(causDev1-Value) {0}; \\

\node(causDev2-Label) [behaveLabel] {causal\\deviation 2}; \&
    \node(causDev2-1) [state] {};
    \devRecTreeNodeCoords{causDev2-1};
    \deviationExampleStates{causDev2-1}{dev}{};
    \mpExampleTerminalsOne{causDev2-1}{}{}{}{draw=black};
    \devRecRootEdges{causDev2-1}{M}{alt}{$\neg\text{M}$}{follow};
    \devRecLEdges{causDev2-1}{H}{alt}{T}{dev};
    \devRecREdges{causDev2-1}{H}{alt}{T}{follow}; \&

    \node(causDev2-2) [state] {};
    \devRecTreeNodeCoords{causDev2-2};
    \deviationExampleStates{causDev2-2}{dev}{};
    \mpExampleTerminalsTwo{causDev2-2}{}{draw=black}{}{};
    \devRecRootEdges{causDev2-2}{M}{follow}{$\neg\text{M}$}{alt};
    \devRecLEdges{causDev2-2}{H}{alt}{T}{dev};
    \devRecREdges{causDev2-2}{H}{alt}{T}{follow}; \&
    \node(causDev2-Value) {$+1$}; \\

\node(causDev3-Label) [behaveLabel] {causal\\deviation 3}; \&
    \node(causDev3-1) [state] {};
    \devRecTreeNodeCoords{causDev3-1};
    \deviationExampleStates{causDev3-1}{}{dev};
    \mpExampleTerminalsOne{causDev3-1}{}{}{draw=black}{};
    \devRecRootEdges{causDev3-1}{M}{alt}{$\neg\text{M}$}{follow};
    \devRecLEdges{causDev3-1}{H}{follow}{T}{alt};
    \devRecREdges{causDev3-1}{H}{dev}{T}{alt}; \&

    \node(causDev3-2) [state] {};
    \devRecTreeNodeCoords{causDev3-2};
    \deviationExampleStates{causDev3-2}{}{dev};
    \mpExampleTerminalsTwo{causDev3-2}{draw=black}{}{}{};
    \devRecRootEdges{causDev3-2}{M}{follow}{$\neg\text{M}$}{alt};
    \devRecLEdges{causDev3-2}{H}{follow}{T}{alt};
    \devRecREdges{causDev3-2}{H}{dev}{T}{alt}; \&
    \node(causDev3-Value) {$-1$}; \\

\node(causDev4-Label) [behaveLabel] {causal\\deviation 4}; \&
    \node(causDev4-1) [state] {};
    \devRecTreeNodeCoords{causDev4-1};
    \deviationExampleStates{causDev4-1}{}{dev};
    \mpExampleTerminalsOne{causDev4-1}{}{}{}{draw=black};
    \devRecRootEdges{causDev4-1}{M}{alt}{$\neg\text{M}$}{follow};
    \devRecLEdges{causDev4-1}{H}{follow}{T}{alt};
    \devRecREdges{causDev4-1}{H}{alt}{T}{dev}; \&

    \node(causDev4-2) [state] {};
    \devRecTreeNodeCoords{causDev4-2};
    \deviationExampleStates{causDev4-2}{}{dev};
    \mpExampleTerminalsTwo{causDev4-2}{draw=black}{}{}{};
    \devRecRootEdges{causDev4-2}{M}{follow}{$\neg\text{M}$}{alt};
    \devRecLEdges{causDev4-2}{H}{follow}{T}{alt};
    \devRecREdges{causDev4-2}{H}{alt}{T}{dev}; \&
    \node(causDev4-Value) {0}; \\
  };
  \rowSepLine{causDev1};
  \rowSepLine[color=grey]{causDev2};
  \rowSepLine[color=grey]{causDev3};
  \rowSepLine[color=grey]{causDev4};
\end{tikzpicture}
\caption{
  Blind causal deviations in the extended matching pennies example from player one's perspective.
  The full set of causal deviations includes all of the external deviations, but these are duplicates of \cref{fig:mp-strat-devs} so they are omitted.
  \PlayerTwoDesc/
  \ArrowDescriptions/
  \EvColDesc/
}
\label{fig:mp-caus-devs}
\end{figure}

\begin{figure}[htbp]
\centering
\begin{tikzpicture}[inner sep=0cm]
  \deviationExampleMatrix{
    \mpExampleHeader; \\[-0.3cm]

    \mpExampleFollowTrees; \\

\node(cfDev1-Label) [behaveLabel] {counterfactual\\deviation 1}; \&
    \node(cfDev1-1) [state, dev] {};
    \devRecTreeNodeCoords{cfDev1-1};
    \deviationExampleStates{cfDev1-1}{}{};
    \mpExampleTerminalsOne{cfDev1-1}{draw=black}{}{}{};
    \devRecRootEdges{cfDev1-1}{M}{dev}{$\neg\text{M}$}{alt};
    \devRecLEdges{cfDev1-1}{H}{follow}{T}{alt};
    \devRecREdges{cfDev1-1}{H}{alt}{T}{follow}; \&

    \node(cfDev1-2) [state, dev] {};
    \devRecTreeNodeCoords{cfDev1-2};
    \deviationExampleStates{cfDev1-2}{}{};
    \mpExampleTerminalsTwo{cfDev1-2}{draw=black}{}{}{};
    \devRecRootEdges{cfDev1-2}{M}{dev}{$\neg\text{M}$}{alt};
    \devRecLEdges{cfDev1-2}{H}{follow}{T}{alt};
    \devRecREdges{cfDev1-2}{H}{alt}{T}{follow}; \&
    \node(cfDev1-Value) {0}; \\

\node(cfDev2-Label) [behaveLabel] {counterfactual\\deviation 2}; \&
    \node(cfDev2-1) [state, dev] {};
    \devRecTreeNodeCoords{cfDev2-1};
    \deviationExampleStates{cfDev2-1}{dev}{};
    \mpExampleTerminalsOne{cfDev2-1}{draw=black}{}{}{};
    \devRecRootEdges{cfDev2-1}{M}{dev}{$\neg\text{M}$}{alt};
    \devRecLEdges{cfDev2-1}{H}{dev}{T}{alt};
    \devRecREdges{cfDev2-1}{H}{alt}{T}{follow}; \&

    \node(cfDev2-2) [state, dev] {};
    \devRecTreeNodeCoords{cfDev2-2};
    \deviationExampleStates{cfDev2-2}{dev}{};
    \mpExampleTerminalsTwo{cfDev2-2}{draw=black}{}{}{};
    \devRecRootEdges{cfDev2-2}{M}{dev}{$\neg\text{M}$}{alt};
    \devRecLEdges{cfDev2-2}{H}{dev}{T}{alt};
    \devRecREdges{cfDev2-2}{H}{alt}{T}{follow}; \&
    \node(cfDev2-Value) {0}; \\

\node(cfDev3-Label) [behaveLabel] {counterfactual\\deviation 3}; \&
    \node(cfDev3-1) [state, dev] {};
    \devRecTreeNodeCoords{cfDev3-1};
    \deviationExampleStates{cfDev3-1}{dev}{};
    \mpExampleTerminalsOne{cfDev3-1}{}{draw=black}{}{};
    \devRecRootEdges{cfDev3-1}{M}{dev}{$\neg\text{M}$}{alt};
    \devRecLEdges{cfDev3-1}{H}{alt}{T}{dev};
    \devRecREdges{cfDev3-1}{H}{alt}{T}{follow}; \&

    \node(cfDev3-2) [state, dev] {};
    \devRecTreeNodeCoords{cfDev3-2};
    \deviationExampleStates{cfDev3-2}{dev}{};
    \mpExampleTerminalsTwo{cfDev3-2}{}{draw=black}{}{};
    \devRecRootEdges{cfDev3-2}{M}{dev}{$\neg\text{M}$}{alt};
    \devRecLEdges{cfDev3-2}{H}{alt}{T}{dev};
    \devRecREdges{cfDev3-2}{H}{alt}{T}{follow}; \&
    \node(cfDev3-Value) {0}; \\

\node(cfDev4-Label) [behaveLabel] {counterfactual\\deviation 4}; \&
    \node(cfDev4-1) [state, dev] {};
    \devRecTreeNodeCoords{cfDev4-1};
    \deviationExampleStates{cfDev4-1}{}{};
    \mpExampleTerminalsOne{cfDev4-1}{}{}{}{draw=black};
    \devRecRootEdges{cfDev4-1}{M}{alt}{$\neg\text{M}$}{dev};
    \devRecLEdges{cfDev4-1}{H}{follow}{T}{alt};
    \devRecREdges{cfDev4-1}{H}{alt}{T}{follow}; \&

    \node(cfDev4-2) [state, dev] {};
    \devRecTreeNodeCoords{cfDev4-2};
    \deviationExampleStates{cfDev4-2}{}{};
    \mpExampleTerminalsTwo{cfDev4-2}{}{}{}{draw=black};
    \devRecRootEdges{cfDev4-2}{M}{alt}{$\neg\text{M}$}{dev};
    \devRecLEdges{cfDev4-2}{H}{follow}{T}{alt};
    \devRecREdges{cfDev4-2}{H}{alt}{T}{follow}; \&
    \node(cfDev4-Value) {0}; \\

\node(cfDev5-Label) [behaveLabel] {counterfactual\\deviation 5}; \&
    \node(cfDev5-1) [state, dev] {};
    \devRecTreeNodeCoords{cfDev5-1};
    \deviationExampleStates{cfDev5-1}{}{dev};
    \mpExampleTerminalsOne{cfDev5-1}{}{}{draw=black}{};
    \devRecRootEdges{cfDev5-1}{M}{alt}{$\neg\text{M}$}{dev};
    \devRecLEdges{cfDev5-1}{H}{follow}{T}{alt};
    \devRecREdges{cfDev5-1}{H}{dev}{T}{alt}; \&

    \node(cfDev5-2) [state, dev] {};
    \devRecTreeNodeCoords{cfDev5-2};
    \deviationExampleStates{cfDev5-2}{}{dev};
    \mpExampleTerminalsTwo{cfDev5-2}{}{}{draw=black}{};
    \devRecRootEdges{cfDev5-2}{M}{alt}{$\neg\text{M}$}{dev};
    \devRecLEdges{cfDev5-2}{H}{follow}{T}{alt};
    \devRecREdges{cfDev5-2}{H}{dev}{T}{alt}; \&
    \node(cfDev5-Value) {0}; \\

\node(cfDev6-Label) [behaveLabel] {counterfactual\\deviation 6}; \&
    \node(cfDev6-1) [state, dev] {};
    \devRecTreeNodeCoords{cfDev6-1};
    \deviationExampleStates{cfDev6-1}{}{dev};
    \mpExampleTerminalsOne{cfDev6-1}{}{}{}{draw=black};
    \devRecRootEdges{cfDev6-1}{M}{alt}{$\neg\text{M}$}{dev};
    \devRecLEdges{cfDev6-1}{H}{follow}{T}{alt};
    \devRecREdges{cfDev6-1}{H}{alt}{T}{dev}; \&

    \node(cfDev6-2) [state, dev] {};
    \devRecTreeNodeCoords{cfDev6-2};
    \deviationExampleStates{cfDev6-2}{}{dev};
    \mpExampleTerminalsTwo{cfDev6-2}{}{}{}{draw=black};
    \devRecRootEdges{cfDev6-2}{M}{alt}{$\neg\text{M}$}{dev};
    \devRecLEdges{cfDev6-2}{H}{follow}{T}{alt};
    \devRecREdges{cfDev6-2}{H}{alt}{T}{dev}; \&
    \node(cfDev6-Value) {0}; \\
  };

  \rowSepLine{cfDev1};
  \rowSepLine[color=grey]{cfDev2};
  \rowSepLine[color=grey]{cfDev3};
  \rowSepLine[color=grey]{cfDev4};
  \rowSepLine[color=grey]{cfDev5};
  \rowSepLine[color=grey]{cfDev6};
\end{tikzpicture}
\caption{
  Blind counterfactual deviations in the extended matching pennies example from player one's perspective.
  \PlayerTwoDesc/
  \ArrowDescriptions/
  \EvColDesc/
}
\label{fig:mp-cf-devs}
\end{figure}
 
\subsection{Observable Sequential Rationality}
We can further extend the extended battle-of-the-sexes example from \cref{sec:ext-bots} to construct a CE that is is not an observable sequential CFCCE, \ie/, a CE where there is a beneficial blind counterfactual deviation in terms of counterfactual value.
We add a private action at the start of the game where player one decides whether or not to play (\Play/ or \NotPlay/) the extended battle-of-the-sexes game.
If they opt-out of the game, they immediately receive $+3$, which matches the value they would receive for attending their preferred event with their partner.
The value for player two does not matter, but for the sake of completeness assume they receive $+3$ as well regardless of their decision.

The recommendations
\begin{itemize}
  \item $\set*{
    \tuple*{\mNotPlay, \; \mNotUpgrade \given \mPlay, \; \mEventY \given \mUpgrade, \; \mEventY \given \mNotUpgrade},
    \tuple*{\mEventY}
  }$ and
  \item $\set*{
    \tuple*{\mNotPlay, \; \mNotUpgrade \given \mPlay, \; \mEventX \given \mUpgrade, \; \mEventX \given \mNotUpgrade},
    \tuple*{\mEventX}
  }$
\end{itemize}
give player one an average value of $+3$.
There is no beneficial swap or internal deviation since there is no other way to achieve an average value of $+3$ in the extended battle-of-the-sexes game.
Therefore, the recommendations are a CE.
However, deviating to \Play/\Upgrade/ and then re-correlating with player two by choosing the event suggested by the recommendations improves player one's value at \Play/ to $+2.5 > +1.5$.
Thus, there is a beneficial counterfactual deviation at \Play/, which reveals that the recommendations are not an observable sequential CFCCE.
The recommendations and this deviation are visualized from player one's perspective in \cref{fig:seq-bots}.

\begin{figure}[htbp]
\centering
\begin{tikzpicture}[inner sep=0cm]
  \deviationExampleMatrix{
    \botsExampleHeader{}; \\

    \node(follow-Label) [behaveLabel] {always\\follow}; \&

    \node(follow-1) [state, xshift=0.7cm] {};
    \devRecTreeInternalCoords{follow-1};
    \node(follow-1L) [state] at (follow-1LCoord) {};
    \node(follow-1R) [util, draw=black] at (follow-1RCoord) {$+3$};

    \devRecTreeInternalCoords{follow-1L};
    \node(follow-1LL) [state] at (follow-1LLCoord) {};
    \node(follow-1LR) [state] at (follow-1LRCoord) {};

    \devRecTreeTerminalCoords{follow-1LL};
    \node(follow-1LLL) [util] at (follow-1LLLCoord) {$+2$};
    \node(follow-1LLR) [util] at (follow-1LLRCoord) {0};

    \devRecTreeTerminalCoords{follow-1LR};
    \node(follow-1LRL) [util] at (follow-1LRLCoord) {$+1$};
    \node(follow-1LRR) [util] at (follow-1LRRCoord) {0};

    \devRecInternalEdges{follow-1}{\Play/}{alt}{\NotPlay/}{follow};
    \devRecInternalEdges{follow-1L}{\Upgrade/}{alt}{\NotUpgrade/}{follow};
    \devRecTerminalEdges{follow-1LL}{\EventX/}{follow}{\EventY/}{alt};
    \devRecTerminalEdges{follow-1LR}{\EventX/}{follow}{\EventY/}{alt}; \&

    \node(follow-2) [state, xshift=0.7cm] {};
    \devRecTreeInternalCoords{follow-2};
    \node(follow-2L) [state] at (follow-2LCoord) {};
    \node(follow-2R) [util, draw=black] at (follow-2RCoord) {$+3$};

    \devRecTreeInternalCoords{follow-2L};
    \node(follow-2LL) [state] at (follow-2LLCoord) {};
    \node(follow-2LR) [state] at (follow-2LRCoord) {};

    \devRecTreeTerminalCoords{follow-2LL};
    \node(follow-2LLL) [util] at (follow-2LLLCoord) {0};
    \node(follow-2LLR) [util] at (follow-2LLRCoord) {$+3$};

    \devRecTreeTerminalCoords{follow-2LR};
    \node(follow-2LRL) [util] at (follow-2LRLCoord) {0};
    \node(follow-2LRR) [util] at (follow-2LRRCoord) {$+2$};

    \devRecInternalEdges{follow-2}{\Play/}{alt}{\NotPlay/}{follow};
    \devRecInternalEdges{follow-2L}{\Upgrade/}{alt}{\NotUpgrade/}{follow};
    \devRecTerminalEdges{follow-2LL}{\EventX/}{alt}{\EventY/}{follow};
    \devRecTerminalEdges{follow-2LR}{\EventX/}{alt}{\EventY/}{follow}; \&

    \node(follow-Value) {$+3$};
    \node(follow-Value2) [below=0.3cm of follow-Value.south west, anchor=north west]
      {$\bs{+1.5 \given \text{\Play/}}$}; \\

    \node(cfDev-Label) [behaveLabel] {counterfactual\\deviation}; \&

    \node(cfDev-1) [state, dev, xshift=0.7cm] {};
    \devRecTreeInternalCoords{cfDev-1};
    \node(cfDev-1L) [state, dev] at (cfDev-1LCoord) {};
    \node(cfDev-1R) [util] at (cfDev-1RCoord) {$+3$};

    \devRecTreeInternalCoords{cfDev-1L};
    \node(cfDev-1LL) [state] at (cfDev-1LLCoord) {};
    \node(cfDev-1LR) [state] at (cfDev-1LRCoord) {};

    \devRecTreeTerminalCoords{cfDev-1LL};
    \node(cfDev-1LLL) [util, draw=black] at (cfDev-1LLLCoord) {$+2$};
    \node(cfDev-1LLR) [util] at (cfDev-1LLRCoord) {0};

    \devRecTreeTerminalCoords{cfDev-1LR};
    \node(cfDev-1LRL) [util] at (cfDev-1LRLCoord) {$+1$};
    \node(cfDev-1LRR) [util] at (cfDev-1LRRCoord) {0};

    \devRecInternalEdges{cfDev-1}{\Play/}{dev}{\NotPlay/}{alt};
    \devRecInternalEdges{cfDev-1L}{\Upgrade/}{dev}{\NotUpgrade/}{alt};
    \devRecTerminalEdges{cfDev-1LL}{\EventX/}{follow}{\EventY/}{alt};
    \devRecTerminalEdges{cfDev-1LR}{\EventX/}{follow}{\EventY/}{alt}; \&

    \node(cfDev-2) [state, dev, xshift=0.7cm] {};
    \devRecTreeInternalCoords{cfDev-2};
    \node(cfDev-2L) [state, dev] at (cfDev-2LCoord) {};
    \node(cfDev-2R) [util] at (cfDev-2RCoord) {$+3$};

    \devRecTreeInternalCoords{cfDev-2L};
    \node(cfDev-2LL) [state] at (cfDev-2LLCoord) {};
    \node(cfDev-2LR) [state] at (cfDev-2LRCoord) {};

    \devRecTreeTerminalCoords{cfDev-2LL};
    \node(cfDev-2LLL) [util] at (cfDev-2LLLCoord) {0};
    \node(cfDev-2LLR) [util, draw=black] at (cfDev-2LLRCoord) {$+3$};

    \devRecTreeTerminalCoords{cfDev-2LR};
    \node(cfDev-2LRL) [util] at (cfDev-2LRLCoord) {0};
    \node(cfDev-2LRR) [util] at (cfDev-2LRRCoord) {$+2$};

    \devRecInternalEdges{cfDev-2}{\Play/}{dev}{\NotPlay/}{alt};
    \devRecInternalEdges{cfDev-2L}{\Upgrade/}{dev}{\NotUpgrade/}{alt};
    \devRecTerminalEdges{cfDev-2LL}{\EventX/}{alt}{\EventY/}{follow};
    \devRecTerminalEdges{cfDev-2LR}{\EventX/}{alt}{\EventY/}{follow}; \&

    \node(cfDev-Value) {$+2.5$};
    \node(cfDev-Value2) [below=0.3cm of cfDev-Value.south west, anchor=north west]
      {$\bs{+2.5 \given \text{\Play/}}$}; \\
  };
  \rowSepLine{cfDev};
\end{tikzpicture}
\caption{
  Player one's perspective of a CFCE based on extended battle-of-the-sexes (see \cref{sec:ext-bots}) that is not an observably sequential CFCE, and the blind counterfactual deviation that has a positive sequential incentive.
  \PlayerTwoDesc/
  \ArrowDescriptions/
  The ``EV'' column shows the expected utility of each deviation under a uniform distribution over the two recommendations, as well as the counterfactual value of the \Play/ information set.
}
\label{fig:seq-bots}
\end{figure}
 
\section{Technical Results}
\subsection{Counterfactual Deviations}
Since the probability that player $i$ plays to each history in an information set must be the same (due to perfect recall), let $\reachProb(\infoSet; \strat_i) = \reachProb(h; \strat_i)$ for all $h \in \infoSet$.

To define the benefit of a counterfactual deviation, it helps to first break counterfactual values into immediate and recursive components.
Let the information sets immediately following $\infoSet$ and action $a$ be
\begin{align}
  \InfoSets_i(\infoSet, a) = \set*{
    \infoSet' \in \InfoSets_i
    \, \Bigg| \,
    \begin{aligned}
      &\forall h' \in \infoSet', \,
      \exists h \in \infoSet, \,
      h' \sqsupseteq ha, \,\\
      &\nexists h'' \in \Histories_i, ha \sqsubseteq h'' \sqsubset h'
    \end{aligned}
  }.
\end{align}
Let the histories that terminate without further input from player $i$ after taking action $a$ in $\infoSet$ be
\begin{align}
  \TerminalHistories_i(\infoSet, a) = \set*{
    z \in \TerminalHistories
    \, \Bigg| \,
    \begin{aligned}
      &\exists h \in \infoSet, \,
      z \sqsupseteq ha, \,\\
      &\nexists h' \in \Histories_i, ha \sqsubseteq h' \sqsubset z
    \end{aligned}
  }.
\end{align}
Now we can decompose counterfactual values:
\begin{align}
    \cfv_{\infoSet}(a; \strat)
      &= \sum_{z \in \TerminalHistories_i(\infoSet, a)}
          \reachProb(z; \strat_{-i}) \utility(z)
        + \sum_{\infoSet' \in \InfoSets_i(\infoSet, a)}
          \cfv(\infoSet'; \strat).
  \shortintertext{Finally, if we define $\termValue(\infoSet, a; \strat_{-i}) = \sum_{z \in \TerminalHistories_i(\infoSet, a)} \reachProb(z; \strat_{-i}) \utility(z)$, we arrive at the simple form,
  }
    &= \underbrace{\termValue(\infoSet, a; \strat_{-i})}_{\text{expected immediate value}}
      + \underbrace{\sum_{\infoSet' \in \InfoSets_i(\infoSet, a)}
        \cfv(\infoSet'; \strat)}_{\text{expected future value}}.
  \label{eq:bellman}
\end{align}

\setcounter{lemma}{0}
\begin{lemma}\cite[Equation 13, Lemma 5]{CFR-TR}
  \label{lem:cfr-decomp}
  The full regret,
$\Regret^T_{\infoSet}(\DevSet_{\PureStratSet_i}^{\EXT})$,
from information set $\infoSet \in \InfoSets_i$ that $\tuple{\strat^t_i}_{t = 1}^T$ suffers is upper bounded by
$\Regret^T_{\infoSet}(\DevSet_{\Actions(\infoSet)}^{\EXT})
+ \max_{a \in \Actions(\infoSet)}
    \sum_{\infoSet' \in \InfoSets_i(\infoSet, a)}
      \Regret^T_{\infoSet'}(\DevSet_{\PureStratSet_i}^{\EXT})$,
where $\InfoSets_i(\infoSet, a) \subset \InfoSets_i$ are the information sets that immediately follow after taking action $a$ in $\infoSet$.
 \begin{proof}
  \begin{align*}
    \Regret^T_{\infoSet}(\DevSet_{\PureStratSet_i}^{\EXT})
      &=
        \max_{\strat'_i \in \PureStratSet_i}
          \sum_{t = 1}^T
            \cfv(\infoSet; \strat'_i, \strat^t_{-i})
            - \cfv(\infoSet; \strat^t).\\
    \shortintertext{Decompose into immediate and future value:}
    &=
      \max_{a \in \Actions(\infoSet)}
        \sum_{t = 1}^T
          \termValue(\infoSet, a; \strat^t_{-i})
          - \cfv(\infoSet; \strat^t)
      + \max_{\strat'_i \in \PureStratSet_i}
        \sum_{t = 1}^T
          \sum_{\infoSet' \in \InfoSets_i(\infoSet, a)}
            \cfv\subex{\infoSet'; \strat'_i, \strat^t_{-i}}.
  \shortintertext{Add and subtract
    $\sum_{t = 1}^T
      \sum_{\infoSet' \in \InfoSets_i(\infoSet, a)}
        \cfv\subex{\infoSet'; \strat^t}$:}
    &=
      \max_{a \in \Actions(\infoSet)}
        \sum_{t = 1}^T
          \underbrace{
            \termValue(\infoSet, a; \strat^t_{-i})
            + \sum_{\infoSet' \in \InfoSets_i(\infoSet, a)}
              \cfv\subex{\infoSet'; \strat^t}}_{
                \text{Value from $a$ assuming $\strat^t_i$ is used thereafter.}}
          - \cfv(\infoSet; \strat^t)
        \\&\quad+ \underbrace{
          \max_{\strat'_i \in \PureStratSet_i}
            \sum_{\infoSet' \in \InfoSets_i(\infoSet, a)}
              \sum_{t = 1}^T
                \cfv\subex{\infoSet'; \strat'_i, \strat^t_{-i}}
                - \cfv\subex{\infoSet'; \strat^t}}_{
                  \text{Suboptimality after $a$, }
                  \sum_{\infoSet' \in \InfoSets_i(\infoSet, a)}
                    \Regret^T_{\infoSet'}(\DevSet_{\PureStratSet_i}^{\EXT}).}.\\
    \shortintertext{Apply the max to both terms in the sum independently:}
    &\le
      \Regret^T_{\infoSet}(\DevSet_{\Actions(\infoSet)}^{\EXT})
      + \max_{a \in \Actions(\infoSet)}
          \sum_{\infoSet' \in \InfoSets_i(\infoSet, a)}
            \Regret^T_{\infoSet'}(\DevSet_{\PureStratSet_i}^{\EXT}).
  \end{align*}
\end{proof}
\end{lemma}

\setcounter{theorem}{1}
\begin{theorem}
  Let $\dev$ be the deviation that plays to reach $\infoSet^{\TARGET}$ from initial information set $\infoSet_0$ and uses action deviation $\dev^{\TARGET} \in \DevSet^{\TARGET} \subseteq \DevSet^{\SWAP}_{\Actions(\infoSet^{\TARGET})}$ once there, but otherwise leaves its given strategy unmodified.
The instantaneous intermediate counterfactual regret with respect to $\dev$ is
$\cfv(\infoSet_0; \dev(\strat_i^t), \strat_{-i}^t)
  - \cfv(\infoSet_0; \strat^t)$
and the cumulative intermediate counterfactual regret is bounded as
\begin{align*}
&\sum_{t = 1}^T
  \cfv(\infoSet_0; \dev(\strat_i^t), \strat_{-i}^t)
  - \cfv(\infoSet_0; \strat^t)\\
  &\le \Regret^T_{\infoSet^{\TARGET}}(\DevSet^{\TARGET})
  + \sum_{\infoSet^0 \preceq \infoSet \prec \infoSet^{\TARGET}}
    \Regret^T_{\infoSet}(\DevSet^{\EXT}_{\Actions(\infoSet)}).
\end{align*}
 \begin{proof}
  Let $a_{n - 1}$ be the action taken at $\infoSet^{\TARGET}$'s predecessor, $\infoSet_{n - 1}$, to reach $\infoSet^{\TARGET}$.
  Then, the intermediate regret from $\infoSet_{n - 1}$ to $\infoSet^{\TARGET}$ (which includes two deviation steps, one at $\infoSet_{n - 1}$ and another at $\infoSet^{\TARGET}$) is
  \begin{align}
  &\sum_{t = 1}^T \cfv(\infoSet_{n - 1}; \dev(\strat^t_i), \strat^t_{-i})
    - \cfv(\infoSet_{n - 1}; \strat^t)\\
    &=\sum_{t = 1}^T
      \cfv_{\infoSet^{\TARGET}}(\dev^{\TARGET}(\strat^t_i(\infoSet^{\TARGET})); \strat^t)
      - \cfv(\infoSet_{n - 1}; \strat^t)
      \\&\quad+ \underbrace{
        \termValue(\infoSet_{n - 1}, a_{n - 1}; \strat^t_{-i})
        + \sum_{\substack{
          \infoSet' \in \InfoSets_i(\infoSet_{n - 1}, a_{n - 1})\\
          \infoSet' \ne \infoSet^{\TARGET}
        }}
          \cfv(\infoSet'; \strat^t)
      }_{\text{Value from histories that do not lead into } \infoSet^{\TARGET}.}.\\
    \shortintertext{
      Add and subtract $\cfv(\infoSet^{\TARGET}; \strat^t)$:
    }\nonumber
    &=
      \sum_{t = 1}^T
        \cfv_{\infoSet^{\TARGET}}(\dev^{\TARGET}(\strat^t_i(\infoSet^{\TARGET})); \strat^t)
        - \cfv(\infoSet_{n - 1}; \strat^t)
        \\&\quad+ \underbrace{
          \termValue(\infoSet_{n - 1}, a_{n - 1}; \strat^t_{-i})
          + \sum_{
            \infoSet' \in \InfoSets_i(\infoSet_{n - 1}, a_{n - 1})
          }
            \cfv(\infoSet'; \strat^t)}_{
              \cfv_{\infoSet_{n - 1}}(a_{n - 1}; \strat^t)
            }
        - \cfv(\infoSet^{\TARGET}; \strat^t).\\
    \shortintertext{Maximizing over actions at $\infoSet_{n - 1}$:}
    &\le
      \sum_{t = 1}^T
        \cfv_{\infoSet^{\TARGET}}(\dev^{\TARGET}(\strat^t_i(\infoSet^{\TARGET})); \strat^t)
        - \cfv(\infoSet^{\TARGET}; \strat^t)
      + \max_{a \in \Actions(\infoSet_{n - 1})}
        \sum_{t = 1}^T
          \cfv_{\infoSet_{n - 1}}(a; \strat^t)
          - \cfv(\infoSet_{n - 1}; \strat^t).\\
    \shortintertext{Maximizing over transformations at $\infoSet^{\TARGET}$:}
    &\le
      \Regret^T_{\infoSet^{\TARGET}}(\DevSet^{\TARGET})
      + \Regret^T_{\infoSet_{n - 1}}(\DevSet^{\EXT}_{\Actions(\infoSet_{n - 1})}).
  \end{align}
  The same relationship holds between $\infoSet_{n - 1}$ and its predecessor where we can treat $\infoSet_{n - 1}$ as the target with $\DevSet^{\EXT}_{\Actions(\infoSet)}$ as the deviation set at $\infoSet_{n - 1}$.
  Therefore, unrolling the sum across all predecessors of $\infoSet^{\TARGET}$ to $\infoSet_0$ completes the argument.
\end{proof}
\end{theorem}

\correctionStart{Add proof.}
\setcounter{theorem}{2}
\begin{theorem}
  If the average full regret for player $i$ at each information set $\infoSet$ for the sequence of $T$ mixed strategy profiles,
$\tuple{ \strat^t }_{t = 1}^T$,
with respect to each blind counterfactual deviation
$\dev$
is
$d_{\infoSet}(\dev)f(T) \ge 0$,
where
$d_{\infoSet}(\dev)$
is the number of external action transformations
$\dev$
applies from $\infoSet$ to the end of the game, then the observable sequential rationality gap for $i$ with respect to the blind counterfactual and external deviations is no more than
$\abs{\InfoSets_i} f(T)$. \end{theorem}
\begin{proof}
  First, establish a simple fact about blind counterfactual deviations.
  At each information set $\infoSet$, the full regret with respect to each blind counterfactual deviation $\dev^{\to \infoSet, a}$ that plays to $\infoSet$, chooses an action $a$, and then immediately re-correlates is at most $f(T)$ since $d_{\infoSet}(\dev^{\to \infoSet, a}) = 1$.
  Formally, denote this value as
  \begin{align}
    \regret_{\infoSet}^{1:T}(a)
      &= \sum_{t = 1}^T
        \cfv_{\infoSet}(a; \strat^t) - \cfv(\infoSet; \strat^t)\\
      &\le f(T).
    \label{eq:immediateCfRegretBound}
  \end{align}

  Next, consider the terminal or height $1$ information sets for player $i$, \ie/, those without successors.
  The maximum full regret with respect to a blind counterfactual deviation is the positive part of the maximum full regret with respect to an external deviation at each terminal information set $\infoSet$; either the blind counterfactual deviation can change the action at $\infoSet$ as an external deviation would or it can re-correlate.
  Therefore, the theorem is proved if all information sets are terminal as $d < \abs{\InfoSets_i}$.
  This serves as the base case of a proof by induction.

  For the induction step, assume that the maximum full regret of an external and or a blind counterfactual deviation that chooses action $a$ at information set $\infoSet$ is upper bounded at each immediate successor $\infoSet' \in \InfoSets_i(\infoSet, a)$ by $(d - 1)f(T)$, where $d$ is the height of $\infoSet$ ($d - 1$ player $i$ actions leads to a terminal information set).
  The full regret of external deviation $\dev^{\to \pureStrat}$ that deviates to pure strategy $\pureStrat$ is upper bounded as
  \begin{align}
    \regret_{\infoSet}^{1:T}(\dev^{\to \pureStrat})
      &\le
        \Regret^T_{\infoSet}(\DevSet_{\Actions(\infoSet)}^{\EXT})
        + \max_{a \in \Actions(\infoSet)}
            \sum_{\infoSet' \in \InfoSets_i(\infoSet, a)}
              \Regret^T_{\infoSet'}(\DevSet_{\PureStratSet_i}^{\EXT})
    \shortintertext{according to \cref{lem:cfr-decomp}.
      We can bound the full regret at each $\infoSet'$ by the induction assumption,}
    \regret_{\infoSet}^{1:T}(\dev^{\to \pureStrat})
      &\le
        \Regret^T_{\infoSet}(\DevSet_{\Actions(\infoSet)}^{\EXT})
        + \abs{\InfoSets_i(\infoSet, a)}
          (d - 1)f(T).
    \shortintertext{We can then bound the maximum immediate regret at $\infoSet$ by \cref{eq:immediateCfRegretBound},}
    \regret_{\infoSet}^{1:T}(\dev^{\to \pureStrat})
      &\le
        f(T)
        + \abs{\InfoSets_i(\infoSet, a)}
          (d - 1)f(T)\\
      &\le
        \abs{\InfoSets_i}f(T),
      \label{eq:cfElevationCompletion}
  \end{align}
  where the last inequality follows from the fact that
  $(d - 1)\abs{\InfoSets_i(\infoSet, a)} \le \abs{\InfoSets_i} - 1$.
  \Cref{eq:cfElevationCompletion} completes the proof.
\end{proof}
\correctionEnd
  
\end{document}